\documentclass[12pt,reqno,a4paper]{amsart}
\usepackage[dvips]{color}
\usepackage{amsmath}
\usepackage{amsxtra}
\usepackage{amscd}
\usepackage{amsthm}
\usepackage{amsfonts}
\usepackage{amssymb}
\usepackage{nccmath}
\usepackage{eucal}
\usepackage{epsfig}
\usepackage{graphics}
\usepackage{dirtytalk}
\usepackage{ulem}
\usepackage[breakable]{tcolorbox}
\usepackage{ifthen}
\usepackage{varioref}
\usepackage{mathrsfs}

\usepackage{cancel}

\usepackage{xcolor}
\definecolor{superlightred}{HTML}{F5F5F5}

\allowdisplaybreaks 

\usepackage{tikz}
\usetikzlibrary{decorations.pathreplacing,
	calligraphy,
	matrix}

\usepackage{bbold}
\usepackage{mlmodern}

\usepackage{enumerate}
\usepackage{hhline}
\usepackage{multirow}

\usepackage[nosort]{cite}
\usepackage{xspace}
\usepackage{setspace}

\usepackage{titlesec}
\titleformat{\section}[block]{\large\scshape\centering}{§\,\thesection}{1em}{}[\HRule{3pt}] 
\titleformat{\subsection}[block]{\normalsize\bfseries}{\thesubsection}{1em}{}
\titleformat{\subsubsection}[block]{\normalsize\bfseries}{\thesubsubsection}{1em}{}

\usepackage[framemethod=TikZ]{mdframed}

\newcommand{\beq}{\begin{equation}}
	\newcommand{\eeq}{\end{equation}}
\newcommand{\beqa}{\begin{eqnarray}}
	\newcommand{\eeqa}{\end{eqnarray}}

\usepackage[T1]{fontenc}

\newcommand{\HRule}[1]{\rule{\linewidth}{#1}} 	

\newcommand{\comp}{\circ}

\newcommand{\tens}{\otimes}

\newcommand{\cals}[1]{\mathcal{#1}}
\newcommand{\fraks}[1]{\mathfrak{#1}}

\newcommand{\normord}[1]{:\mathrel{#1}:}
\newcommand{\bb}[1]{\mathbb{#1}}

\numberwithin{equation}{section}
\newtheorem{thm}{Theorem}[section]
\newtheorem{prop}[thm]{Proposition}
\newtheorem{lem}[thm]{Lemma}

\newtheorem{cor}[thm]{Corollary}

\newtheorem{dfn}[thm]{Definition}
\newtheorem{dfn-prp}[thm]{Definition-Proposition}

\newtheorem{nota}[thm]{Notation}

\newmdtheoremenv[backgroundcolor=superlightred]{assump}[thm]{Assumption}

\newcommand{\supermac}{\overline{SP}}
\newcommand{\quantumcorner}{\mathfrak{QC}}
\newcommand{\quantumcornerpolynomial}{\quantumcorner_\lambda(x_1,\dots,x_{N};y_1,\dots,y_{M};w_1,\dots,w_{L};q,t)}

\newcommand{\dualmap}{
\widetilde{\Psi}^{(q,\xi)}_{\lambda}
}

\theoremstyle{definition}
\newtheorem{rem}[thm]{Remark}

\usepackage{xpatch}
\xpatchcmd{\proof}{\itshape}{\normalfont\proofnamefont}{}{}
\newcommand{\proofnamefont}{}
\renewcommand{\proofnamefont}{\bfseries}

\usepackage{ytableau}
\ytableausetup{mathmode, boxframe=normal, boxsize=2.6em,centertableaux,notabloids}

\begin{document}
	
	\baselineskip = 18pt 
	
	\begin{titlepage}
		
		\bigskip
		\hfill\vbox{\baselineskip12pt
			\hbox{}
		}
		\bigskip
		\begin{center}
			\Large{ \scshape
				\HRule{3pt}
				Quantum Corner Polynomials: A Generalization of Super Macdonald Polynomials and Their 
				VOA Correspondence
				\HRule{3pt}
			}
		\end{center}
		\bigskip
		\bigskip

		\begin{center}
			\large 
			Panupong Cheewaphutthisakun$^{a}$\footnote{panupong.cheewaphutthisakun@gmail.com},
			Jun'ichi Shiraishi$^{b}$\footnote{shiraish@ms.u-tokyo.ac.jp},
			Keng Wiboonton$^{a}$\footnote{keng.w@chula.ac.th}
			\\
			\bigskip
			\bigskip
			$^a${\small {\it Department of Mathematics and Computer Science, Faculty of Science, Chulalongkorn University, 
					\\
					Bangkok, 10330, Thailand}}\\
			$^b${\small {\it Graduate School of Mathematical Sciences, University of Tokyo, Komaba, Tokyo 153-8914, Japan}} \\
		\end{center}
		\bigskip
		\bigskip
		
		%
		{\vskip 6cm}
		{\small
			\begin{quote}
				\noindent {\textbf{\textit{Abstract}}.}
				In this paper, we introduce a family of partially symmetric polynomials, which we call quantum corner polynomials, as a generalization of the Sergeev-Veselov super Macdonald polynomials. We show that these quantum corner polynomials are precisely the partially symmetric polynomials corresponding to the quantum corner VOAs. Furthermore, we provide a detailed proof of the partial symmetricity of these polynomials. 
		\end{quote}}

	\end{titlepage}
	
	
\section{Introduction}
\label{intro}

Two-dimensional conformal field theory (CFT) is a fundamental theory in theoretical physics, playing a crucial role in the study of string theory. Beyond its significance in theoretical physics, two-dimensional CFT is also profoundly important in mathematics, connecting various fields such as the representation theory of infinite-dimensional algebras, combinatorics, and algebraic geometry. In this paper, we explore the relation between quantum deformation of generalizations of $W_N$ algebra \cite{BPZ}  \cite{BS1993} \cite{Zamo}, which are symmetry algebras of conformal field theories with higher spin fields, and (partially) symmetric polynomials, which is an important object in combinatorics.

Motivated by the quantization program of universal enveloping algebras \cite{DR} \cite{Jim-original}, it is natural to ask whether the $W_N$ algebra can be quantized. This question was addressed in \cite{AKOS-950} \cite{SKAO95}, where the quantum $W_N$ algebra was constructed based on the principle that the singular vectors of the quantum $W_N$ algebra should be described by Macdonald polynomials $P_\lambda(x_1,\dots,x_N ; q,t)$. The underlying idea is that it is well-known that the singular vectors of the $W_N$ algebra are described by a class of symmetric polynomials called Jack polynomials \cite{awata-excited} \cite{mimachi-yamada-1995-singular}, and Jack polynomials have a good candidate for their quantum deformation version, namely Macdonald polynomials \cite{macbook-1998} (see also \cite{noumi-mac}). From the results in \cite{AKOS-950} \cite{SKAO95}, it is known that the quantum $W_N$ algebra is generated by currents $\widetilde{T}_1(z), \dots,\widetilde{T}_N(z)$ which satisfy relations called quadratic relations, which can be regarded as a quantum deformation version of the OPEs of the $W_N$ algebra. 

The quantum $W_N$ algebra plays a significant role in the five-dimensional version \cite{AY} \cite{AY2} \cite{Taki} of AGT correspondence \cite{AGT} \cite{Wyl}, which establishes a correspondence between the instanton partition function \cite{nekrasov} of five-dimensional $SU(N)$ gauge theory and the conformal blocks of the quantum $W_N$ algebra. 

As described above, we see the relation between quantum $W_N$ algebra and Macdonald polynomials through the study of singular vectors. In fact, we can also demonstrate the relation between quantum $W_N$ algebra and Macdonald polynomials through the correlation function of the currents that generate the quantum $W_N$ algebra. More precisely, it was shown in \cite{FHSSY-ker} that 

\footnotesize
\begin{align}
	&
	\lim_{\xi \rightarrow t^{-1}}\,\,
	(
	\dualmap
	\comp 
	\bigg|_{
		\substack{
			q_1 = q, \\
			q_2 = q^{-1}t,\\
			q_3 = t^{-1} \\
		}
	}
	)
	\left(
	\cals{N}_{\lambda}(z_1,\dots,z_k )
	\times
	\prod_{1 \leq i < j \leq k}f_{11}\left(\frac{z_j}{z_i} \right)
	\times
	\langle 0 |\widetilde{T}^{\vec{u}}_{1}(z_1 )\cdots \widetilde{T}^{\vec{u}}_{1}(z_k )|0\rangle
	\right)
	\label{eqn11-1058-7aug}
	\\
	&= 
	P_\lambda\left(
	u_1,\dots,u_N ; q,t
	\right)
	\notag 
\end{align}
\normalsize
where $P_\lambda(u_1,\dots,u_N;q,t)$ is the Macdonald polynomial, and $u_1,\dots,u_N$ are nonzero complex numbers appearing in the definition of vertex operator (see equation \eqref{eqn238-0000-7aug}). For the definition of the map $\dualmap$ and the factor $\cals{N}_{\lambda}(z_1,\dots,z_k )$, the reader is referred to Section \ref{sec4-0012}.

In \cite{GR17}, Gaiotto and Rap\v{c}\'{a}k constructed a family of vertex operator algebras (VOAs) by studying a system of D5, NS5, and $(-1,-1)$ 5-branes filled with orthogonal D3-branes. This family of VOAs is referred to as corner VOAs and denoted by $\widetilde{Y}_{M,L,N}$, where $M,L,N$ represent the numbers of orthogonal D3-branes. Soon after, alternative but equivalent definitions for these corner VOAs were discovered in \cite{PR17} \cite{PR18}.

It is known that the corner VOA $\widetilde{Y}_{M,L,N}$ is a generalization of the $W_N$ algebra in the sense that $\widetilde{Y}_{0,0,N} = W_N$. Therefore, it is natural to extend the quantization procedure of the $W_N$ algebra to the quantization of the corner VOA. This task was successfully accomplished in \cite{HMNW} (see also \cite{Misha}), and the resulting algebra from this quantization process is called the quantum corner VOA, denoted by $q\widetilde{Y}_{M,L,N}$. From its construction, we find that the quantum $W_N$ algebra mentioned earlier is none other than $q\widetilde{Y}_{0,0,N}$. Similar to the quantum $W_N$ algebra, the currents that generate the quantum corner VOA satisfy quadratic relations that are a generalization of the quadratic relations of the quantum $W_N$ algebra (see \textbf{Proposition \ref{prp29=1336-6aug}}).

Since we know that the correlation function of the currents of the quantum $W_N$ algebra is related to Macdonald polynomials, it is natural to ask whether we can generalize the relationship in equation (1.11) to the case of the quantum corner VOA $q\widetilde{Y}_{M,L,N}$. The answer to this question has been partially addressed in the paper by \cite{CSW2025}, which showed that for $q\widetilde{Y}_{M,0,N}$, the corresponding polynomial is a class of partially symmetric polynomials called super Macdonald polynomials $\supermac_\lambda(x_1,\dots,x_N;y_1,\dots,y_M;q,t)  $, constructed by Sergeev and Veselov \cite{sv2007-supermac}. More precisely, we have:

\footnotesize
\begin{align}
	&
	\lim_{\xi \rightarrow t^{-1}}\,\,
	(
	\dualmap
	\comp 
	\bigg|_{
		\substack{
			q_1 = q, \\
			q_2 = q^{-1}t,\\
			q_3 = t^{-1} \\
		}
	}
	)
	\left(
	\cals{N}_{\lambda}(z_1,\dots,z_k )
	\times
	\prod_{1 \leq i < j \leq k}f^{(3^N1^M)}_{11}\left(\frac{z_j}{z_i} \right)
	\times
	\langle 0 |\widetilde{T}^{(3^N1^M),\vec{u}}_{1}(z_1 )\cdots \widetilde{T}^{(3^N1^M),\vec{u}}_{1}(z_k )|0\rangle
	\right)
	\label{eqn12-1057-7aug}
	\\
	&= 
	\overline{SP}_\lambda\left(
	u_1,\dots,u_N, q^{-\frac{1}{2}}t^{-\frac{1}{2}}u_{N+1}, \dots, q^{-\frac{1}{2}}t^{-\frac{1}{2}}u_{N+M} ; q,t
	\right).
	\notag 
\end{align}
\normalsize
Note that equation \eqref{eqn11-1058-7aug} is a particular case of equation \eqref{eqn12-1057-7aug} when $M = 0$. 

In this paper, we extend this relation to the most general case, $q\widetilde{Y}_{M,L,N}$. We construct a class of partially symmetric polynomials that are a generalization of super Macdonald polynomials, which we call quantum corner polynomials. We will show that these quantum corner polynomials correspond to the quantum corner VOA $q\widetilde{Y}_{M,L,N}$ (\textbf{Theorem \ref{thm43-1458-25jul}}) and demonstrate that they are partially symmetric polynomials (\textbf{Theorem \ref{thm42-main-1526}}).

\subsubsection*{Organization of material}

This paper is organized as follows. In Section \ref{sec2-0010}, we review the definition of the quantum corner VOA. We begin by recalling the definition of the quantum toroidal $\fraks{gl}_1$ algebra and the horizontal Fock representation of the quantum toroidal $\fraks{gl}_1$ algebra. Subsequently, we define vertex operators through the tensor product of horizontal Fock representations. We then define the quantum corner VOA by its generating currents, which can be expressed in terms of the previously defined vertex operators.

In Section \ref{sec3-0010}, we begin by introducing the definition of a tritableau. We then define the quantum corner polynomial using a combinatorial formula expressed as a sum over tritableaus. From this definition, it can be immediately shown that in the case where $L = 0$, the quantum corner polynomial reduces to super Macdonald polynomials.

In Section \ref{sec4-0012}, we state and prove one of the main theorems of this paper, which asserts that the quantum corner polynomials constructed in Section \ref{sec3-0010} correspond to the quantum corner VOA $q\widetilde{Y}_{M,L,N}$ (\textbf{Theorem \ref{thm43-1458-25jul}}). For the proof of \textbf{Theorem \ref{thm43-1458-25jul}}, one of the crucial lemmas is \textbf{Lemma \ref{lemm42-1252-25jul}}, which states that only contributions from reverse semistandard Young tritableaus yield nonzero contributions. We prove \textbf{Lemma \ref{lemm42-1252-25jul}} in Appendix \ref{appA-1254}.

In Section \ref{sec5-0016-7aug}, we demonstrate that the quantum corner polynomials defined in Section \ref{sec3-0010} are partially symmetric polynomials (\textbf{Theorem \ref{thm42-main-1526}}). This is the other main theorem of this paper. 

\subsubsection*{Conventions}

Throughout this paper, we use the convention that if $a > b$
\begin{align}
	\displaystyle \sum_{i = a}^{b}(\text{any expression}) &= 0, 
	\\
	\displaystyle \prod_{i = a}^{b}(\text{any expression}) &= 1. 
	\label{eqn14-1159-sun17aug}
\end{align}

\subsubsection*{Acknowledgement}
This research project is supported by the Second Century Fund (C2F), Chulalongkorn University. 
Our work is supported in part by Grants-in-Aid for Scientific Research (Kakenhi); 24K06753 (J.S.), 21K03180 (J.S.).
	
\section{Quantum Corner VOA}
\label{sec2-0010}

In this section, we review the definition of the quantum toroidal $\fraks{gl}_1$ algebra and one of its representations, known as the horizontal Fock representation. This representation is essential for defining the vertex operator, which in turn plays an essential role in the definition of the quantum corner VOA. For a more detailed discussion of the quantum toroidal $\fraks{gl}_1$ algebra and its representation theory, we refer the reader to \cite{Awata-note}  \cite{BS}  \cite{DI}   \cite{FHHSY-comm} \cite{feigin-toroidal1, feigin-toroidal2,feigin-toroidal3}
\cite{Miki} \cite{Sch} 
(see also \cite{	AFS, AKMM17 , AKMM17-2 , AKMM18 , BFM17, BFM17-2 ,BJ, CK,CK2,Ma,Zen}). 

\subsection{Quantum Toroidal $\fraks{gl}_1$ Algebra}

Throughout this paper, we define $q$ and $t$ as nonzero complex numbers satisfying a generic condition: if there exist integers $a,b \in \bb{Z}$ such that $q^at^b = 1$, then $a = b = 0$.

\begin{dfn}
The \textbf{quantum toroidal $\fraks{gl}_1$ algebra}, denoted by $U_{q,t}(\widehat{\widehat{\fraks{gl}}}_1)$, is the unital associative algebra over $\bb{C}$ generated by 
\begin{align}
	E_k, F_k, K^{\pm}_0, H_{\pm r}, C  \hspace{0.3cm} (k \in \bb{Z}, r \in \bb{Z}^{\geq 1}), 
\end{align}
subject to the following defining relations:
\begin{gather}
	C \text{ is a central element}, 
	\\
	K^{\pm}(z)K^{\pm}(w) = 		K^{\pm}(w)K^{\pm}(z),
	\label{2.3main}
	\\
	K^{+}(z)K^-(w) = \frac{\cals{G}(w/C z)}{\cals{G}(C w/z)} K^-(w)K^{+}(z),
	\\
	K^+(z)E(w) = \cals{G}(w/z)E(w)K^+(z),
	\\
	K^-(C z)E(w) = \cals{G}(w/z)E(w)K^-(C z),
	\\
	K^+(C z)F(w) = \cals{G}(w/z)^{-1}F(w)K^+(C z),
	\\
	K^-(z)F(w) = \cals{G}(w/z)^{-1}F(w)K^-(z),
	\\
	\label{EEexchage}
	E(z)E(w) = \cals{G}(w/z)E(w)E(z),
	\\
	F(z)F(w) = \cals{G}(w/z)^{-1}F(w)F(z),
	\\
	[E(z),F(w)] = \frac{1}{(q_1 - 1)(q_2 - 1)(q_3 - 1)}
	\bigg(
	\delta\Big(\frac{Cw}{z}\Big)K^+(z) - \delta\Big(\frac{Cz}{w}\Big)K^-(w) 
	\bigg),
	\label{2.11main}
\end{gather}
where in the relations above, $q_1 = q, q_2 = q^{-1}t, q_3 = t^{-1}$, 
\begin{gather}
	E(z) := \sum_{k\in\mathbb{Z}} E_k z^{-k}, \quad
	F(z) := \sum_{k\in\mathbb{Z}} F_k z^{-k}, \quad
	K^{\pm}(z) := K_0^{\pm}\exp\left(
	\pm \sum_{r=1}^{\infty} H_{\pm r} z^{\mp r}
	\right), 
	\label{2.12}
	\\
	\cals{G}(z) := 
	\frac{(1 - q_1^{-1}z)(1 - q_2^{-1}z)(1 - q_3^{-1}z)}{
		(1 - q_1z)(1 - q_2z)(1 - q_3z)
	}, 
	\quad 
	\delta(z) := \sum_{k \in \bb{Z}}
	z^k. 
\end{gather}
\end{dfn}

The quantum toroidal $\fraks{gl}_1$ algebra is known to have a Hopf algebra structure, which means it is equipped with a coproduct, a counit, and an antipode. However, for our purposes, only the coproduct formula will be essential. Therefore, in the following proposition, we will only state the coproduct formula for the quantum toroidal $\fraks{gl}_1$ algebra. 

\begin{prop}
The map $\Delta : U_{q,t}(\widehat{\widehat{\fraks{gl}}}_1)
\rightarrow 
U_{q,t}(\widehat{\widehat{\fraks{gl}}}_1) \tens U_{q,t}(\widehat{\widehat{\fraks{gl}}}_1)$ defined by the formula below is an algebra homomorphism:
\begin{align}
	\Delta\big(E(z)\big) &= E(z) \tens 1 + K^-(C_1z) \tens E(C_1z),
	\label{copro1}
	\\
	\Delta\big(F(z)\big) &= F(C_2z) \tens K^+(C_2z) + 1 \tens F(z),
	\\
	\Delta\big(K^+(z)\big) &= K^+(z) \tens K^+(C_1^{-1}z),
	\\
	\Delta\big(K^-(z)\big) &= K^-(C_2^{-1}z) \tens K^-(z),
	\label{copro4}
	\\
	\Delta\big(		C		\big) &= C \tens C, 
\end{align}
where in the formulas above, $C_1 := C \tens 1, C_2 := 1 \tens C$. This map is called the \textbf{coproduct} of $U_{q,t}(\widehat{\widehat{\fraks{gl}}}_1)$. 
\end{prop}

In addition, we define the algebra homomorphism
\begin{align}
\Delta^{(n)} : U_{q,t}(\widehat{\widehat{\fraks{gl}}}_1)
\rightarrow 
\underbrace{	U_{q,t}(\widehat{\widehat{\fraks{gl}}}_1) \tens \cdots \tens U_{q,t}(\widehat{\widehat{\fraks{gl}}}_1)		}_{n+1}
\label{eqn219-1142-6aug}
\end{align}
by setting $\Delta^{(1)} = \Delta$ and for $n \in \bb{Z}^{\geq 2}$, 
\begin{align}
	\Delta^{(n)} := (\Delta \tens \underbrace{		1 \tens \cdots \tens 1		}_{n-1} ) \comp \Delta^{(n-1)}. 
\label{eqn220-1142-6aug}
\end{align}

Next, we will discuss a representation of the quantum toroidal $\fraks{gl}_1$ algebra, which is called the horizontal Fock representation. In order to define the horizontal Fock representation of the quantum toroidal $\fraks{gl}_1$ algebra, we have to first introduce the Heisenberg algebra. 

\begin{dfn}
For each $i \in \{1,2,3\}$, the \textbf{Heisenberg algebra} $\cals{B}^{(i)}$ is defined as the algebra over $\bb{C}$ generated by $\left\{a^{(i)}_{\pm n}, a^{(i)}_0~|~ n \in \bb{Z}^{\geq 1}\right\}$ which satisfies the following defining relations: 
\begin{align}
	[a_n^{(i)},a_m^{(i)}] &= \frac{n}{\kappa_n}(q_i^{n/2} - q_i^{-n/2})
	\delta_{n+m,0}a^{(i)}_0. 
\label{eqn221-1222-6aug}
\end{align}
where $q_1 = q, q_2 = q^{-1}t, q_3 = t^{-1}$
and $\kappa_n := (q^n_1 - 1)(q^n_2 - 1)(q^n_3 - 1)$. 
\end{dfn}

If we define $|0\rangle$ to satisfy the condition $a_n|0\rangle = 0$ for any $n \in \bb{Z}^{\geq 1}$, then the vector space $\cals{H}^{(i)}$ defined by 
\begin{align}
	\cals{H}^{(i)} := \operatorname{span}
	\left\{
	a^{(i)}_{-\lambda_1}\cdots a^{(i)}_{-\lambda_m}|0\rangle
	\;\middle\vert\;
	\begin{array}{@{}l@{}}
		(1) \,\,  m \in \bb{Z}^{\geq 0}\\
		(2) \,\, \lambda_1 \geq \cdots \geq \lambda_m \geq 1
	\end{array}
	\right\}.
\end{align}
has the structure of left $\cals{B}^{(i)}$-module. Note that in this module, $a^{(i)}_0$ acts as scalar multiplication by $1$. In the following proposition, we will show that we can construct a representation of quantum toroidal $\fraks{gl}_1$ algebra with $\cals{H}^{(i)}$ as its representation space. 

\begin{prop}
For each $i \in \{1,2,3\}$ and $u \in \bb{C}\backslash\{0\}$, the map $\rho_{H,u}^{(i)} : U_{q,t}(\widehat{\widehat{\fraks{gl}}}_1)
\rightarrow \operatorname{End}(\cals{H}^{(i)})$ defined by the equations below is an algebra homomorphism. We call this the \textbf{horizontal Fock representation} of the quantum toroidal $\fraks{gl}_1$ algebra:
\begin{align}
	\rho_{H,u}^{(i)}\left(
	K^+(z)
	\right)
	&= 
	\exp\left(
	\sum_{n = 1}^{\infty}\frac{\kappa_n}{n}q^{n/4}_ia^{(i)}_{n}z^{-n}
	\right),
	\\
	\rho_{H,u}^{(i)}
	\left(
	K^-(z)
	\right)
	&= 
	\exp\left(
	- \sum_{n = 1}^{\infty}\frac{\kappa_n}{n}q^{-n/4}_{i}
	a^{(i)}_{-n}z^n
	\right),
	\\
	\rho_{H,u}^{(i)}
	\left(
	E(z)
	\right)
	&= u\widetilde{d}_1
	\exp\left(
	\sum_{n = 1}^{\infty}
	\frac{\kappa_n}{n}
	\frac{q^{-n/4}_i}{q^{n/2}_{i} - q^{-n/2}_{i}}
	a^{(i)}_{-n}z^n
	\right)
	\exp\left(
	-\sum_{n = 1}^{\infty}
	\frac{\kappa_n}{n}
	\frac{q^{n/4}_i}{q^{n}_{i}  - 1}
	a^{(i)}_{n}z^{-n}
	\right),
	\\
	\rho_{H,u}^{(i)}
	\left(
	F(z)
	\right)
	&= 
	u^{-1}\widetilde{d}_2
	\exp\left(
	-\sum_{n = 1}^{\infty}
	\frac{\kappa_n}{n}
	\frac{q^{n/4}_i}{q^{n/2}_{i} - q^{-n/2}_{i}}
	a^{(i)}_{-n}z^{n}
	\right)
	\exp\left(
	\sum_{n = 1}^{\infty}
	\frac{\kappa_n}{n}
	\frac{q^{n/4}_i}{q^{n/2}_{i} - q^{-n/2}_{i}}
	a^{(i)}_{n}z^{-n}
	\right),
	\\
	\rho_{H,u}^{(i)}
	\left(
	C
	\right)
	&= q^{1/2}_{i},
\end{align}
where in the equation above, 
\begin{align}
	\widetilde{d}_1 &= 
	\frac{
		1
	}{
		(1 - q_{i^\prime})
		(1 - q_{i^{\prime\prime}})
	}, 
	\label{eqn315-2135}
	\\
	\widetilde{d}_2 &=
	\frac{
		(1 - q^{-1}_i)
	}{
		(q_1 - 1)(q_2 - 1)(q_3 - 1)
	},
\end{align}
with $i^\prime$ and $i^{\prime\prime}$ are chosen such that $i, i^\prime, i^{\prime\prime}$ are distinct elements of the set $\{1,2,3\}$. 
\end{prop}

For later convenience, we introduce the following shorthand notations:
\begin{align}
	\varphi^{(i)}(z;p) &= 
	\exp\left(
	- \sum_{n = 1}^{\infty}
	\frac{\kappa_n}{n}q^{-n/4}_{i}
	a^{(i)}_{-n}z^n
	\right),
	\\
	\eta^{(i)}(z;p) &= 
	\exp\left(
	\sum_{n = 1}^{\infty}
	\frac{\kappa_n}{n}
	\frac{q^{-n/4}_{i}}{q^{n/2}_{i} - q^{-n/2}_{i}}
	a^{(i)}_{-n}z^n
	\right)
	\exp\left(
	-\sum_{n = 1}^{\infty}
	\frac{\kappa_n}{n}
	\frac{q^{n/4}_{i}}{q^{n}_{i}  - 1}
	a^{(i)}_{n}z^{-n}
	\right). 
\end{align}

Given that quantum toroidal $\fraks{gl}_1$ algebra possesses a Hopf algebra structure, we can construct representations of the quantum toroidal $\fraks{gl}_1$ algebra by taking tensor products of representations. The tensor product of horizontal Fock representations, which we will discuss below, plays a crucial role in defining the quantum corner VOA, the subject of the next subsection. 

For each $\vec{c} = (c_1,\dots,c_n) \in \{1,2,3\}^n$ and $\vec{u} = (u_1,\dots,u_n) \in (\bb{C}\backslash\{0\})^n$, we define 
\begin{align}
	\rho^{\vec{c}}_{H,\vec{u}}
	: 
	U_{q,t}(\widehat{\widehat{\fraks{gl}}}_1)
	\rightarrow \operatorname{End}(\cals{H}^{(c_1)} \tens \cdots \tens \cals{H}^{(c_n)})
\end{align}
as the algebra homomorphism given by 
\begin{align}
	\rho^{\vec{c}}_{H,\vec{u}}
	:= 
	\left(
	\rho_{H,u_1}^{(c_1)} \tens \rho_{H,u_2}^{(c_2)} \tens \cdots \tens \rho_{H,u_n}^{(c_n)}	
	\right) \comp 
	\Delta^{(n-1)}. 
\label{eqn231-1209-6aug}
\end{align}
Here $\Delta^{(n-1)}$ is the map defined in \eqref{eqn219-1142-6aug} and \eqref{eqn220-1142-6aug}.

\subsection{Quantum Corner VOA}

In this subsection, we will define the quantum corner VOA using the ingredients discussed in the previous subsection. First, we define 
\begin{align}
	\alpha(z) &:= 
	\exp\left(
	-\sum_{r = 1}^{\infty}\frac{1}{C^r - C^{-r}}\widetilde{b}_{-r}z^r
	\right),
	\\
	\beta(z) &:=
	\exp\left(
	\sum_{r = 1}^{\infty}\frac{1}{C^r - C^{-r}}\widetilde{b}_rz^{-r}
	\right),
\end{align}
where $\widetilde{b}_{\pm r}$ appearing in the above equations are defined by:
\begin{align}
	K^+(z) &= 
	\exp\left(
	\sum_{r = 1}^{\infty}\widetilde{b}_rC^{r}z^{-r}
	\right),
	\hspace{1cm}
	K^-(z) =
	\exp\left(
	-\sum_{r = 1}^{\infty}\widetilde{b}_{-r}z^r
	\right). 
\end{align}

\begin{dfn}
For each $\vec{c} = (c_1,\dots,c_n) \in \{1,2,3\}^n$ and $\vec{u} = (u_1,\dots,u_n) \in (\bb{C} \backslash \{0\})^n$, we define the vertex operator $\widetilde{\Lambda}^{\vec{c},\vec{u}}_j(z) \,\, (j = 1,\dots,n)$ as 
\begin{align}
	\widetilde{\Lambda}^{\vec{c},\vec{u}}_j(z)
	:= 
	\rho^{\vec{c}}_{H,\vec{u}}\left(\alpha(z)\right)
	\Lambda^{\vec{c},\vec{u}}_j(z)
	\rho^{\vec{c}}_{H,\vec{u}}\left(\beta(z)\right),
\end{align}
where 
\begin{align}
	\Lambda^{\vec{c},\vec{u}}_j(z) := 
	u_j
	\varphi^{(c_1)}(q_{c_1}^{1/2}z;p) \tens \cdots \tens \varphi^{(c_{j-1})}(q_{c_1}^{1/2} \cdots q_{c_{j-1}}^{1/2}z;p) 	
	\tens \eta^{(c_j)}(q_{c_1}^{1/2} \cdots q_{c_{j-1}}^{1/2}z;p) \tens 
	\underbrace{	1 \tens \cdots \tens 1				}_{n - j}, 
\label{eqn238-0000-7aug}
\end{align}
and $\rho^{\vec{c}}_{H,\vec{u}}$ is the tensor product representation defined in \eqref{eqn231-1209-6aug}. 
\end{dfn}

\begin{dfn}
Let $\vec{c} = (c_1,\dots,c_n) \in \{1,2,3\}^n$ satisfy the condition $q_{\vec{c}} := \prod_{k = 1}^{n}q_{c_k} \neq 1$. In this situation, for $r, m \in \bb{Z}^{\geq 0}$ such that $r \leq m$, we define:
\begin{align}
	f^{\vec{c}}_{r,m}(z)
	&:=
	\exp\bigg[
	\sum_{k = 1}^{\infty}\frac{z^k}{k}
	(q_3^{\frac{r}{2}k} - q_3^{-\frac{r}{2}k})(q_{\vec{c}}^{\frac{k}{2}}q^{-\frac{m}{2}k}_{3}
	- q_{\vec{c}}^{-\frac{k}{2}}q^{\frac{m}{2}k}_{3})
	\frac{
		(q^{\frac{k}{2}}_1 - q^{-\frac{k}{2}}_1)(q^{\frac{k}{2}}_2 - q^{-\frac{k}{2}}_2)
	}{
		(q^{\frac{k}{2}}_{\vec{c}} - q^{-\frac{k}{2}}_{\vec{c}})(q^{\frac{k}{2}}_3 - q^{-\frac{k}{2}}_3)
	}
	\bigg],
	\\
	f^{\vec{c}}_{m,r}(z)
	&:=
	f^{\vec{c}}_{r,m}(z). 
\end{align}
\end{dfn}

By using the relation in equation \eqref{eqn221-1222-6aug}, we can easily show the following proposition. 

\begin{prop}
\label{prp27-1335-6aug-wed}
\begin{align}
	\label{eqn339-2150-12jan}
	\widetilde{\Lambda}^{\vec{c},\vec{u}}_i(z)			\widetilde{\Lambda}^{\vec{c},\vec{u}}_j(w)		
	= 
	\begin{cases}
		\displaystyle
		f^{\vec{c}}_{11}\left(\frac{w}{z}	  	\right)^{-1}
		\Delta\left(	q_3^{\frac{1}{2}}\frac{w}{z}			\right)
		\normord{		\widetilde{\Lambda}^{\vec{c},\vec{u}}_i(z)	 	\widetilde{\Lambda}^{\vec{c},\vec{u}}_j(w)						}
		\text{ for } i < j 
		\\
		\displaystyle
		f^{\vec{c}}_{11}\left(\frac{w}{z}  \right)^{-1}
		\gamma_{c_i}\left(\frac{w}{z} \right)
		\normord{		\widetilde{\Lambda}^{\vec{c},\vec{u}}_i(z)		\widetilde{\Lambda}^{\vec{c},\vec{u}}_i(w)						}
		\text{ for } i = j
		\\
		\displaystyle
		f^{\vec{c}}_{11}\left(\frac{w}{z} \right)^{-1}
		\Delta\left(	q_3^{-\frac{1}{2}}\frac{w}{z} 			\right)
		\normord{		\widetilde{\Lambda}^{\vec{c},\vec{u}}_i(z)		\widetilde{\Lambda}^{\vec{c},\vec{u}}_j(w)				}
		\text{ for } i > j 
	\end{cases}
\end{align}
where
\begin{align}
	\Delta(z) 
	&:= 
	\frac{
		(1 - q_1q_3^{\frac{1}{2}}z)
		(1 - q^{-1}_1q_3^{-\frac{1}{2}}z)
	}{
		(1 - q_3^{\frac{1}{2}}z)
		(1 - q_3^{-\frac{1}{2}}z)
	},
	\\
	\gamma_{c_i}(z) 
	\label{eqn340-1147-15apr}
	&:= 
	\frac{
		(1 - q_{c_i}z)
		(1 - q^{-1}_{c_i}z)
	}{
		(1 - q_{3}z)
		(1 - q^{-1}_{3}z)
	}. 
\end{align} 
\end{prop}

\begin{dfn}[\cite{HMNW}]
Let $N,M,L \in \bb{Z}^{\geq 0}$, let 
\begin{align}
	\vec{c} = (c_1,\dots,c_n) = (
	\underbrace{	3,\dots,3		}_{N},
	\underbrace{		1,\dots,1			}_{M},\underbrace{	2,\dots,2			}_{L}
	)
	= 
	(3^N1^M2^L),
\end{align}
and let $\vec{u} = (u_1,\dots,u_n) \in (\bb{C} \backslash \{0\})^n$. The \textbf{quantum corner VOA}, denoted by $q\widetilde{Y}_{M,L,N}$, is defined as the algebra generated by the set of currents $\left\{\widetilde{T}^{\vec{c},\vec{u}}_m(z)\right\}_{m \in  \bb{Z}^{\geq 0}}$ where:
\begin{align}
	\widetilde{T}^{\vec{c},\vec{u}}_m(z)
	&=
	\underbrace{			\sum_{k_1, \dots,k_n \in \bb{Z}^{\geq 0}}			}_{k_1 + \cdots + k_n = m}
	\left[
	\prod_{i = 1}^{n}
	\prod_{j_i = 1}^{k_i}
	\left(
	-(q^{1/2}_{c_i}q^{1/2}_{3})
	\frac{(1 - q_3^{j_i-1}q_{c_i}^{-1})}{(1 - q_3^{j_i})}
	\right)
	\right]
	\normord{
		\prod_{i = 1}^{n}\prod_{j_i = 1}^{k_i}
		\widetilde{		\Lambda		}^{\vec{c},\vec{u}}_i(q_3^{- \sum_{\ell = 1}^{i-1}k_\ell -j_i + 1}z)
	}. 
	\label{eqn21-1326-27jul}
\end{align}
\end{dfn}

By using \textbf{Proposition \ref{prp27-1335-6aug-wed}}, we can show that the currents $\left\{\widetilde{T}^{\vec{c},\vec{u}}_m(z)\right\}_{m \in  \bb{Z}^{\geq 0}}$ of the quantum corner VOA satisfy the quadratic relations stated in \textbf{Proposition \ref{prp29=1336-6aug}} below. 

\begin{prop}[\cite{CSw2024}]
\label{prp29=1336-6aug}
Let $\vec{c} = (c_1,\dots,c_n) \in \{1,2,3\}^n$ satisfy the condition $q_{\vec{c}} := \prod_{k = 1}^{n}q_{c_k} \neq 1$. Then, for each $\vec{u} = (u_1,\dots,u_n) \in (\bb{C} \backslash \{0\})^n$, 
and for each $r, m \in \bb{Z}^{\geq 1}$ such that $r \leq m$, we get that 
\begin{align}
	&f^{\vec{c}}_{r,m}\left(
	q^{\frac{r - m}{2}}_{3}\frac{w}{z} 
	\right)
	\widetilde{	T	}^{\vec{c},\vec{u}}_r(z )\widetilde{	T		}^{\vec{c},\vec{u}}_m(w ) 
	- f^{\vec{c}}_{m,r}\left(
	q^{\frac{m - r}{2}}_{3}\frac{z}{w} 
	\right)
	\widetilde{T}^{\vec{c},\vec{u}}_m(w )\widetilde{T}^{\vec{c},\vec{u}}_r(z )
	\notag \\
	&= 
	\frac{
		(1 - q_1)(1 - q_2)
	}{
		(1 - q_3^{-1})
	}
	\sum_{k = 1}^{r}
	\left(
	\prod_{\ell = 1}^{k-1}
	\frac{(1 -  q_1q_3^{-\ell})( 1 - q_2q_3^{-\ell})}{( 1 - q_3^{-\ell - 1})( 1 - q_3^{-\ell})}
	\right)
	\biggl\{
	\delta\left(q_3^k\frac{w}{z}\right)f^{\vec{c}}_{r-k,m+k}(q_3^{\frac{r-m}{2}} )\widetilde{T}^{\vec{c},\vec{u}}_{r-k}(q_3^{-k}z)\widetilde{T}^{\vec{c},\vec{u}}_{m+k}(q_3^kw)
	\notag \\ 
	&\hspace{0.3cm}- 
	\delta\left(q_3^{r-m-k}\frac{w}{z}\right)f^{\vec{c}}_{r-k,m+k}(q_3^{\frac{m-r}{2}})\widetilde{T}^{\vec{c},\vec{u}}_{r-k}(z )\widetilde{T}^{\vec{c},\vec{u}}_{m+k}(w)
	\biggr\}. 
\end{align}
\end{prop}

\section{Quantum Corner Polynomial}
\label{sec3-0010}

The main objective of this section is to define quantum corner polynomials. To do this, the concept of reverse semi-standard Young tritableau is essential. We therefore begin by providing a precise definition of the reverse semi-standard Young tritableau.

\subsection{Reverse Semi-standard Young Tritableau and Quantum Corner Polynomial}

\begin{dfn}[Partition]
A \textbf{partition} is a sequence $(\lambda_i)_{i = 1}^{\infty}$ of non-negative integers satisfying the following conditions:
\begin{enumerate}[(1)]
\item $\lambda_1 \geq \lambda_2 \geq \cdots$
\item There exist only finitely many $i \in \bb{Z}^{\geq 1}$ such that $\lambda_i \neq 0$. 
\end{enumerate}
The set of all partitions is denoted by $\operatorname{Par}$. 
\end{dfn}

\begin{dfn}
For $\lambda := (\lambda_i)_{i = 1}^{\infty} \in \operatorname{Par}$, we define its length, denoted by $\ell(\lambda)$, to be the number of nonzero elements in $(\lambda_i)_{i = 1}^{\infty}$.  
\end{dfn}

\begin{dfn}[Subset of Partitions]
Let $\lambda := (\lambda_i)_{i = 1}^{\infty} \in \operatorname{Par}$ and $\mu := (\mu_i)_{i = 1}^{\infty} \in \operatorname{Par}$. We say that $\mu \subseteq \lambda$ if for any $i \in \bb{Z}^{\geq 1}$, we have $\mu_i \leq \lambda_i$. 
\end{dfn}

\begin{dfn}
For each $k \in \bb{Z}^{\geq 0}$, we define 
\begin{align}
\operatorname{Par}(k) := 
\left\{
(\lambda_i)_{i = 1}^{\infty} \in \operatorname{Par}
\;\middle\vert\;
\begin{array}{@{}l@{}}
\sum_{i = 1}^{\infty}\lambda_i = k
\end{array}
\right\}
\end{align}
Sometimes, we will use the notation $|\lambda| = k$ to denote that $\lambda \in \operatorname{Par}(k)$. 
\end{dfn}

Any partition can be represented by a Young diagram, where there are $\lambda_i$ boxes in the $i$-th row. For example, the Young diagram of the partition $(3,1,1,0,0,\dots)$ is 
\begin{align}
\ydiagram{3,1,1}
\end{align}

\begin{dfn}[Transpose of a Partition]
Let $\lambda := (\lambda_i)_{i = 1}^{\infty} \in \operatorname{Par}$. We define the \textbf{transpose} of $\lambda$, denoted by $\lambda^\prime$, to be $\lambda^\prime := (\lambda^\prime_i)_{i = 1}^{\infty}$, where $\lambda^\prime_i$ is the number of boxes in the $i$-th column of the Young diagram of $\lambda$. 
\end{dfn}

\begin{dfn}[Arm Length and Leg Length]
Let $\lambda := (\lambda_i)_{i = 1}^{\infty} \in \operatorname{Par}$, and let $s$ be the box located at row $i$ and column $j$ of the Young diagram of $\lambda$. We define the \textbf{arm length} of the box $s$ to be $a_\lambda(s) := \lambda_i - j$ and the \textbf{leg length} of the box $s$ to be $\ell_\lambda(s) := \lambda^\prime_j - i$. 
\end{dfn}

\begin{dfn}[Young Tableau]
A Young tableau is a Young diagram in which each box is filled with a positive integer. 
\end{dfn}

\begin{dfn}
	Let $(i_1,\dots,i_k)$ be a finite sequence of positive integers, and let $\lambda \in \operatorname{Par}(k)$. Define 
	$T(i_1,\dots,i_k;\lambda)$ to be the Young tableau obtained by arranging the elements of $(i_1,\dots,i_k)$ onto Young diagram of $\lambda$. The elements are placed by filling the boxes of $\lambda$ from left to right, row by row, starting from the first row and proceeding to the last row. 
\end{dfn}

\begin{dfn}[Reverse Semi-Standard Young Tritableau]
\label{def35-1058-25jul}
Let $\lambda \in \operatorname{Par}$ and $N,M,L \in \bb{Z}^{\geq 1}$. A \textbf{reverse semi-standard Young tritableau} (reverse SSYTT) of type $(N,M,L)$ with shape $\lambda$ is a Young diagram of shape $\lambda$ that satisfies the following conditions:
\begin{enumerate}[(1)]
\item Each box in the Young diagram is assigned an element of the set $\{1,\dots,N+M+L\}$. 
The elements in $\{1,\dots,N\}$, $\{N+1,\dots,N+M\}$, $\{N+M+1,\dots,N+M+L\}$ are referred to as 
ordinary numbers, super numbers, and hyper numbers, respectively. 
\item The assigend numbers in each row of the Young diagram are weakly decreasing from left to right. 
\item The assigend numbers in each column of the Young diagram are weakly decreasing from top to bottom.  
\item For each column, the ordinary numbers are strictly decreasing from top to bottom.
\item For each row, the super numbers are strictly decreasing from left to right. 
\end{enumerate}
The set of all reverse SSYTTs of type $(N,M,L)$ with shape $\lambda$ is denoted by $\operatorname{RSSYTT}(N,M,L;\lambda)$. 
\end{dfn}

\begin{rem}
The concept of a reverse SSYTT can be extended to a skew shape $\lambda\backslash\mu$, where $\mu \subseteq \lambda$. The ruls for assigning ordinary numbers, super numbers, and hyper numbers remain the same as in \textbf{Definition \ref{def35-1058-25jul}}, with the only difference being that the shape of the diagram is now $\lambda\backslash\mu$. The set of all reverse SSYTTs of type $(N,M,L)$ with shape $\lambda\backslash\mu$ is denoted by $\operatorname{RSSYTT}(N,M,L;\lambda\backslash\mu)$. 
\end{rem}

\begin{rem}
The elements of the set $\operatorname{RSSYTT}(N,M,0;\lambda\backslash\mu)$ are called reverse semi-standard Young bitableaus (reverse SSYBTs). 
\end{rem}

\begin{nota}
For $T \in \operatorname{RSSYTT}(N,M,L;\lambda)$, we define 
$T_0,T_1,$ and $T_2$ to be the sub-diagrams of $T$ consisting of boxes assigned with ordinary numbers, super numbers, and hyper numbers, respectively. 
\end{nota}

\begin{dfn}
Let $T \in \operatorname{RSSYTT}(N,M,L;\lambda)$. For each $\alpha \in \{1,\dots,\ell(\lambda)\}$ and $\beta \in \{1,\dots,N+M+L\}$, we define $\theta_{\alpha,\beta}(T)$ to be the number of boxes in row $\alpha$ that are assigned the number $\beta$. 
\end{dfn}

\begin{dfn}
Let $\lambda, \mu \in \operatorname{Par}$ where $\mu \subseteq \lambda$. Define 
\begin{align}
	\fraks{C}_{\lambda/\mu}(q,t)
	=
	\prod_{
		1 \leq i \leq j \leq \ell(\mu)
	}
	\frac{
		f_{q,t}\left(	t^{j - i} q^{ \mu_{i} 	- \mu_{j}	 		}						\right)
		f_{q,t}\left(		
		t^{j - i} q^{ \lambda_{i} 	- \mu_{j + 1} 		}					\right)
	}{
		f_{q,t}\left(	t^{j - i} q^{ \lambda_{i} 	- \mu_{j}		}						\right)
		f_{q,t}\left(		
		t^{j- i} q^{ \mu_{i} 	- \mu_{j + 1}			}					\right)
	}
\end{align}
where $\displaystyle f_{q,t}(u) := \frac{(tu;q)_\infty}{(qu;q)_\infty}$. Here $(x;q)_\infty := \prod_{i = 0}^{\infty}(1 - xq^i)$. 
\end{dfn}

\begin{dfn}
\label{dfn213-1508-1aug}
Let $(N,M,L) \in (\bb{Z}^{\geq 0} \times \bb{Z}^{\geq 0} \times \bb{Z}^{\geq 0})\backslash\{(0,0,0)\}$. For each $T \in \operatorname{RSSYTT}(N,M,L;\lambda)$, we define 

\footnotesize
\begin{align}
	\cals{A}_2(T;q,t)
	=& 
	\prod_{d = N+M+1}^{N+M+L}\fraks{C}_{T^{(d)}/T^{(d+1)}}(q,t)
	\times 
	\prod_{\zeta = 1}^{\ell(T_2)}
	\prod_{d = N+M+1}^{N+M+L}
	\prod_{j = 1}^{\theta_{\zeta,d} - 1}
	\frac{
		(1 - qt^{-1}q^{ j})(1 - t)
	}{
		(1 - qt^{-1})(1 - tq^{ j})
	}
	\\
	&\times
	\prod_{d = N+M+1}^{N+M+L} 
	\prod_{\zeta = 1}^{\ell(T_2)}
	\prod_{\tau = \zeta + 1}^{\ell(T_2)} 
	\prod_{\alpha = 0}^{\theta_{\zeta,d} - 1}
	\frac{
		\left(		1 - t^{\tau - \zeta + 1} q^{ T^{(d+1)}_{\zeta} + \alpha	 -	T^{(d)}_{\tau}			}					\right)
		\left(
		1 - t^{\tau - \zeta} q^{ T^{(d+1)}_{\zeta} + \alpha	-	T^{(d+1)}_{\tau}			}	
		\right)
	}{
		\left(
		1 - t^{\tau - \zeta} q^{ T^{(d+1)}_{\zeta} + \alpha	 -	T^{(d)}_{\tau} 		}	
		\right)
		\left(		1 - t^{\tau - \zeta + 1} q^{	T^{(d+1)}_{\zeta} + \alpha	 - T^{(d+1)}_{\tau}			}					\right)
	}
	\notag 
\end{align}
\normalsize
where $T^{(d)}$ for $d = N+M+1,\dots,N+M+L$ denotes the Young sub-diagram of $T$ consisting of boxes assigned with numbers $d,\dots,N+M+L$. As a convention, we set $T^{(N+M+L+1)} = (0,0,0,\dots)$. 
\end{dfn}

\begin{rem}
It is clear that $\cals{A}_2(T;q,t)$ depends only on the sub-diagram $T_2$ of $T$. 
\end{rem}

The next lemma plays a crucial role in the proof of \textbf{Lemma \ref{lemm36-1211-27jul}}. To maintain the flow of the main body of this section, we relegate its proof to \textbf{Appendix \ref{app-A-2246-16aug}}. 

\begin{lem}
\label{prp317-2237-16aug}
Let $(N,M,L) \in (\bb{Z}^{\geq 0} \times \bb{Z}^{\geq 0} \times \bb{Z}^{\geq 0})\backslash\{(0,0,0)\}$. For each $T \in \operatorname{RSSYTT}(N,M,L;\lambda)$, we have 
\footnotesize
\begin{align}
	&\cals{A}_2(T;q,t)
	\label{eqn35-1146-17aug}
	\\
	&=\prod_{\zeta = 1}^{\ell(T_2)}
	\prod_{d = N+M+1}^{N+M+L}
	\prod_{j = 1}^{\theta_{\zeta,d} - 1}
	\frac{
		(1 - qt^{-1}q^{\theta_{\zeta,d} - j})(1 - t)
	}{
		(1 - t^{-1}q)(1 - tq^{\theta_{\zeta,d} - j})
	}
	\notag 
	\\
	&\times 
	\prod_{\zeta = 1}^{\ell(T_2)}
	\prod_{d = N+M+1}^{N+M+L}
	\prod_{\alpha = N+M+1}^{d - 1} 
	\prod_{\omega = 1}^{\theta_{\zeta,\alpha}}
	\frac{
		\left(	1 - q^{	\sum_{\gamma = \alpha + 1}^{N + M + L} \theta_{\zeta,\gamma}	+ \omega - \sum_{\gamma = d+1}^{N + M + L}\theta_{\zeta,\gamma} 		}			\right)
		\left(	1 - tq^{	\sum_{\gamma = \alpha + 1}^{N + M + L} \theta_{\zeta,\gamma}	+ \omega - (\sum_{\gamma = d}^{N + M + L}\theta_{\zeta,\gamma} 	+ 1 )	}		\right)
	}{
		\left(	1 - q^{	\sum_{\gamma = \alpha + 1}^{N + M + L} \theta_{\zeta,\gamma}	+ \omega - \sum_{\gamma = d}^{N + M + L}\theta_{\zeta,\gamma} 		}			\right)
		\left(	1 - 	tq^{	\sum_{\gamma = \alpha + 1}^{N + M + L} \theta_{\zeta,\gamma}	+ \omega - (\sum_{\gamma = d + 1}^{N + M + L}\theta_{\zeta,\gamma} 	+ 1 )	}			\right)
	}
	\notag 
	\\
	&\times
	\prod_{\zeta = 1}^{\ell(T_2)}
	\prod_{d = N+M+1}^{N+M+L} 
	\prod_{\tau = \zeta + 1}^{\ell(T_2)} 
	\prod_{\omega = 1}^{\theta_{\tau,d}}
	\frac{
		\left(		1 - t^{-1} (t^{-1})^{\tau - \zeta}	q^{	\sum_{\gamma = d+1}^{N + M + L}\theta_{\tau,\gamma}	+ \omega - 	\sum_{\gamma = d+ 1}^{N + M + L} \theta_{\zeta,\gamma}				}			\right)
		\left(
		1 - (t^{-1})^{\tau - \zeta}	q^{	\sum_{\gamma = d+1}^{N + M + L}\theta_{\tau,\gamma}	+ \omega - 	\sum_{\gamma = d}^{N + M + L} \theta_{\zeta,\gamma}				}		
		\right)
	}{
		\left(		1 - t^{-1} (t^{-1})^{\tau - \zeta}	q^{	\sum_{\gamma = d+1}^{N + M + L}\theta_{\tau,\gamma}	+ \omega - 	\sum_{\gamma = d}^{N + M + L} \theta_{\zeta,\gamma}				}			\right)
		\left(
		1 - (t^{-1})^{\tau - \zeta}	q^{	\sum_{\gamma = d+1}^{N + M + L}\theta_{\tau,\gamma}	+ \omega - 	\sum_{\gamma = d+1}^{N + M + L} \theta_{\zeta,\gamma}				}		
		\right)
	}
	\notag 
	\\
	&\times 
	\prod_{\zeta = 1}^{\ell(T_2)}
	\prod_{d = N+M+1}^{N+M+L}
	\prod_{\tau = \zeta + 1}^{\ell(T_2)}
	\prod_{p = d+1}^{N+M+L}
	\prod_{\Xi = 1}^{\theta_{\zeta,d}}
	\prod_{\omega = 1}^{\theta_{\tau,p}}
	\Biggl\{
	\notag 
	\\
	&\frac{
		\left(	1 - q^{-1}					
		(t^{-1})^{\tau - \zeta} q^{	\sum_{\gamma = p + 1}^{N + M + L} \theta_{\tau,\gamma}	+ \omega	- (\sum_{\gamma = d + 1}^{N + M + L} \theta_{\zeta,\gamma} + \Xi)		}	
		\right)
		\left(	1 - qt^{-1}			
		(t^{-1})^{\tau - \zeta} q^{	\sum_{\gamma = p + 1}^{N + M + L} \theta_{\tau,\gamma}	+ \omega	- (\sum_{\gamma = d + 1}^{N + M + L} \theta_{\zeta,\gamma} + \Xi)		}
		\right)
	}{
		\left(	1 - q			
		(t^{-1})^{\tau - \zeta} q^{	\sum_{\gamma = p + 1}^{N + M + L} \theta_{\tau,\gamma}	+ \omega	- (\sum_{\gamma = d + 1}^{N + M + L} \theta_{\zeta,\gamma} + \Xi)		}
		\right)
		\left(	1 - q^{-1}t			
		(t^{-1})^{\tau - \zeta} q^{	\sum_{\gamma = p + 1}^{N + M + L} \theta_{\tau,\gamma}	+ \omega	- (\sum_{\gamma = d + 1}^{N + M + L} \theta_{\zeta,\gamma} + \Xi)		}
		\right)
	}
	\notag 
	\\
	&\times
	\frac{
		\left(	1 - t		
		(t^{-1})^{\tau - \zeta} q^{	\sum_{\gamma = p + 1}^{N + M + L} \theta_{\tau,\gamma}	+ \omega	- (\sum_{\gamma = d + 1}^{N + M + L} \theta_{\zeta,\gamma} + \Xi)		}
		\right)
	}{
		\left(	1 - t^{-1}			
		(t^{-1})^{\tau - \zeta} q^{	\sum_{\gamma = p + 1}^{N + M + L} \theta_{\tau,\gamma}	+ \omega	- (\sum_{\gamma = d + 1}^{N + M + L} \theta_{\zeta,\gamma} + \Xi)		}
		\right)
	}
	\Biggr\}
	\notag 
\end{align}
\normalsize
\end{lem}
\begin{proof}
See \textbf{Appendix \ref{app-A-2246-16aug}}. 
\end{proof}

\begin{rem}
Examining the quantity on the RHS of equation \eqref{eqn35-1146-17aug}, a natural question arises as to whether the denominator can be zero. We address this concern here. First, we note that the only factors in the denominator that could potentially evaluate to zero are 
\begin{align}
\left(	1 - q^{	\sum_{\gamma = \alpha + 1}^{N + M + L} \theta_{\zeta,\gamma}	+ \omega - \sum_{\gamma = d}^{N + M + L}\theta_{\zeta,\gamma} 		}			\right)
\label{eqn36-1149-sun17aug}
\end{align}
and
\begin{align}
\left(	1 - q^{-1}t			
(t^{-1})^{\tau - \zeta} q^{	\sum_{\gamma = p + 1}^{N + M + L} \theta_{\tau,\gamma}	+ \omega	- (\sum_{\gamma = d + 1}^{N + M + L} \theta_{\zeta,\gamma} + \Xi)		}
\right)
\label{eqn37-1200-sun17aug}
\end{align}

Regarding the factor \eqref{eqn36-1149-sun17aug}, since $\sum_{\gamma = \alpha + 1}^{N + M + L} \theta_{\zeta,\gamma} \geq \sum_{\gamma = d}^{N + M + L}\theta_{\zeta,\gamma}$, the only situation in which the factor \eqref{eqn36-1149-sun17aug} could be zero is if both of the following conditions are satisfied:
\begin{enumerate}[(1)]
	\item $\sum_{\gamma = \alpha + 1}^{N + M + L} \theta_{\zeta,\gamma} = \sum_{\gamma = d}^{N + M + L}\theta_{\zeta,\gamma}$,
	\item $\omega = 0$. 
\end{enumerate}
However, we assert that $\omega$ can not be zero. The only way for $\omega = 0$ would be if $\theta_{\zeta,\alpha} = 0$. But if $\theta_{\zeta,\alpha} = 0$, then by our convention \eqref{eqn14-1159-sun17aug}, the product $\prod_{\omega = 1}^{\theta_{\zeta,\alpha}}$ is an empty product, which evaluates to $1$.

Regarding the factor \eqref{eqn37-1200-sun17aug}, it will be zero if and only if the following conditions are satisfied:
\begin{enumerate}[(1)]
\item $\tau - \zeta = 1$
\item $\sum_{\gamma = p + 1}^{N + M + L} \theta_{\tau,\gamma}	+ \omega	- (\sum_{\gamma = d + 1}^{N + M + L} \theta_{\zeta,\gamma} + \Xi)	 = 1$
\end{enumerate}
We will now show that these conditions can never be satisfied. First, note that if either $\theta_{\zeta,d} = 0$ or
$\theta_{\tau,p} = 0$, our convention \eqref{eqn14-1159-sun17aug} implies that the product 
$\prod_{\Xi = 1}^{\theta_{\zeta,d}}\prod_{\omega = 1}^{\theta_{\tau,p}}$ becomes an empty product, which yields a contribution of $1$, and thus no singularity arises. For this reason, we only need to consider the case where $\theta_{\zeta,d} \geq 1$ and $\theta_{\tau,p} \geq 1$. In this situation, we observe that the expression 
$\sum_{\gamma = p + 1}^{N + M + L} \theta_{\zeta + 1,\gamma}	+ \omega	- (\sum_{\gamma = d + 1}^{N + M + L} \theta_{\zeta,\gamma} + \Xi)$
attains its maximum value when $p = d+1$, $\omega = \theta_{\tau,p}$ and $\Xi = 1$. In this case, we find its value to be $-1$. Therefore, we conclude that for any $\Xi \in \{1,\dots,\theta_{\zeta,d}\}$ and $\omega \in \{1,\dots,\theta_{\tau,p}\}$, we have 
$\sum_{\gamma = p + 1}^{N + M + L} \theta_{\zeta + 1,\gamma}	+ \omega	- (\sum_{\gamma = d + 1}^{N + M + L} \theta_{\zeta,\gamma} + \Xi) \neq 1$. 
\end{rem}

\begin{dfn}\mbox{}
\label{dfn214-2113-26jul}
\begin{enumerate}[(1)]
\item Let $\lambda, \mu \in \operatorname{Par}$ where $\mu \subseteq \lambda$. Define 
\begin{align}
	\psi_{\lambda/\mu}(q,t) := \prod_{1 \leq i \leq j \leq \ell(\mu)}
	\frac{
		f_{q,t}(q^{\mu_i - \mu_j}t^{j - i})
		f_{q,t}(q^{\lambda_i - \lambda_{j+1}}t^{j - i})
	}{
		f_{q,t}(q^{\lambda_i - \mu_j}t^{j - i})
		f_{q,t}(q^{\mu_i - \lambda_{j+1}}t^{j - i})
	}, 
	\label{eqn27-2052}
\end{align}
where $\displaystyle f_{q,t}(u) := \frac{(tu;q)_\infty}{(qu;q)_\infty}$. Here $(x;q)_\infty := \prod_{i = 0}^{\infty}(1 - xq^i)$. 
\vspace{0.2cm}
\item Let $\lambda, \mu \in \operatorname{Par}$ where $\mu \subseteq \lambda$.  For each $T \in \operatorname{RSSYTT}(N,0,0;\lambda\backslash \mu)$, we define 
\begin{align}
	\psi_{T}(q,t) := \prod_{k = 1}^{N}\psi_{T^{(k)}/T^{(k+1)}}(q,t), 
\end{align}
where $T^{(k)}$ for $k = 1,\dots,N$ denotes the Young sub-diagram of $T$ consisting of boxes assigned with numbers $k,\dots,N$. As a convention, we set $T^{(N+1)} = \mu$. 
\end{enumerate}
\end{dfn}

\begin{dfn}
Let $(N,M,L) \in (\bb{Z}^{\geq 0} \times \bb{Z}^{\geq 0} \times \bb{Z}^{\geq 0})\backslash\{(0,0,0)\}$, and $T \in \operatorname{RSSYTT}(N,M,L;\lambda)$. Then, we define 

\footnotesize
\begin{align}
	R_T(q,t)
	:=&  
	\cals{A}_2(T;q,t) \times 
	\psi_{T^{\prime}_{1}}(t,q) \times \psi_{T_0}(q,t)
	\times
	\frac{
		H(\operatorname{sh}(T_1 \sqcup T_2),q,t^{-1})
	}{
		H(\operatorname{sh}(T_1 \sqcup T_2)^\prime,t^{-1},q)
	}
	\times 
	\prod_{\square \in \operatorname{sh}(T_2)}
	\frac{
		(t^{-1})^{\ell_{\operatorname{sh}(T_2)}(\square) + 1} - q^{a_{\operatorname{sh}(T_2)}(\square)}
	}{
		q^{a_{\operatorname{sh}(T_2)}(\square) + 1} - (t^{-1})^{\ell_{\operatorname{sh}(T_2)}(\square)}
	}
	\notag 
	\\
	&\times 
	\left(	\frac{q^{-1}t - 1}{t^{-1} - 1}				\right)^{|T_2|}
\end{align}
\normalsize
where $\operatorname{sh}(T)$ and 
$|T|$ are the shape and the number of boxes in the Young tableau $T$, respectively. Here for each $\mu \in \operatorname{Par}$, 
$H(\mu,q,t) := 
\prod_{s \in \mu}\left(
q^{a_{\mu}(s) + 1} - t^{\ell_\mu(s)}
\right)$. 
\end{dfn}

\begin{dfn}[Quantum Corner Polynomial]
\label{defn314-1234-25jul}
Let $\lambda \in \operatorname{Par}$, $(N,M,L) \in (\bb{Z}^{\geq 0} \times \bb{Z}^{\geq 0} \times \bb{Z}^{\geq 0})\backslash\{(0,0,0)\}$. We define \textbf{quantum corner polynomial} to be 
\begin{align}
	\quantumcornerpolynomial
	:= 
	\sum_{T \in \operatorname{RSSYTT}(N,M,L;\lambda)}
	R_T(q,t)x_T
\end{align}
where $x_T := x_{i_1}\cdots x_{i_k}$. For each $i \in \{1,\dots,M\}$ and $j \in \{1,\dots,L\}$, we write $x_{N+i} = y_i$ and $x_{N+M+j} = w_j$. 
\end{dfn}

\begin{prop}
Let $\supermac_\lambda(x_1,\dots,x_N;y_1,\dots,y_M;q,t)$ be the super Macdonald polynomial defined in \cite{CSW2025}. Then, 
\begin{align}
\quantumcorner_\lambda(x_1,\dots,x_N;y_1,\dots,y_M;q,t) = 
\supermac_\lambda(x_1,\dots,x_N;y_1,\dots,y_M;q,t)  
\end{align}
\end{prop}
\begin{proof}
When $L = 0$, we get that $T_2 = \emptyset$. Then, from \textbf{Definition \ref{defn314-1234-25jul}}, we obtain that 
\begin{align}
&\quantumcorner_\lambda(x_1,\dots,x_N;y_1,\dots,y_M;q,t) 
\\
&=
\sum_{T \in \operatorname{RSSYTT}(N,M,0;\lambda)}
R_T(q,t)x_T,
\notag 
\\
&= 
\sum_{T \in \operatorname{RSSYTT}(N,M,0;\lambda)}
\left(
\psi_{T^{\prime}_{1}}(t,q) \times \psi_{T_0}(q,t)
\times
\frac{
	H(\operatorname{sh}(T_1),q,t^{-1})
}{
	H(\operatorname{sh}(T_1)^\prime,t^{-1},q)
}
\right)x_T.
\notag  
\end{align}
On the other hand, we know from the combinatorial formula of super Macdonald polynomial \cite{CSW2025} that 
\footnotesize
\begin{align}
\supermac_\lambda(x_1,\dots,x_N;y_1,\dots,y_M;q,t)  
= 
\sum_{T \in \operatorname{RSSYTT}(N,M,0;\lambda)}
\left(
\psi_{T^{\prime}_{1}}(t,q) \times \psi_{T_0}(q,t)
\times
\frac{
	H(\operatorname{sh}(T_1),q,t^{-1})
}{
	H(\operatorname{sh}(T_1)^\prime,t^{-1},q)
}
\right)x_T. 
\end{align}
\normalsize
Thus, we get that $\quantumcorner_\lambda(x_1,\dots,x_N;y_1,\dots,y_M;q,t) = 
\supermac_\lambda(x_1,\dots,x_N;y_1,\dots,y_M;q,t)$. 
\end{proof}

\section{Quantum Corner VOA/Quantum Corner Polynomial Correspondence}
\label{sec4-0012}

In this section, we state and prove one of the main result of this paper (\textbf{Theorem \ref{thm43-1458-25jul}}). 

\begin{dfn}
Let $\lambda \in \operatorname{Par}$. We define the map $\dualmap$ to be the map that sends $f(z_1,\dots,z_k) \in \bb{C}(z_1,\dots,z_k)$ to 
\begin{align}
	&f(y,qy,\dots,q^{\lambda_1 - 1}y
	\notag	\\
	&\hspace{0.4cm} \xi y,q\xi y,\dots,q^{\lambda_2 - 1}\xi y
	\notag \\
	&\hspace{2.8cm}\vdots
	\notag \\
	&\hspace{0.4cm} \xi^{\ell(\lambda) - 1} y,q\xi^{\ell(\lambda) - 1} y,\dots,q^{\lambda_{\ell(\lambda)} - 1}\xi^{\ell(\lambda) - 1} y). 
\end{align}
which is an element of $\bb{C}(\xi, y)$. 
\end{dfn}

\begin{rem}
In general, the map $\dualmap$ is not well-defined for every element $f(z_1,\dots,z_k)$ in $\bb{C}(z_1,\dots,z_k)$. However, as we will see, for the element $f(z_1,\dots,z_k) \in \bb{C}(z_1,\dots,z_k)$ considered in this paper, there will be a well-defined image under the map $\dualmap$. 
\end{rem}

\begin{thm}
\label{thm43-1458-25jul}
Let $\vec{c} = (3^N1^M2^L)$, $\vec{u} = (u_1,\dots,u_{N+M+L})$, and let $\lambda \in \operatorname{Par}(k)$. Then, 

\footnotesize
\begin{align}
	&
	\lim_{\xi \rightarrow t^{-1}}\,\,
	(
	\dualmap
	\comp 
	\bigg|_{
		\substack{
			q_1 = q, \\
			q_2 = q^{-1}t,\\
			q_3 = t^{-1} \\
		}
	}
	)
	\left(
	\cals{N}_{\lambda}(z_1,\dots,z_k )
	\times
	\prod_{1 \leq i < j \leq k}f^{\vec{c}}_{11}\left(\frac{z_j}{z_i} \right)
	\times
	\langle 0 |\widetilde{T}^{\vec{c},\vec{u}}_{1}(z_1 )\cdots \widetilde{T}^{\vec{c},\vec{u}}_{1}(z_k )|0\rangle
	\right)
	\\
	&= 
	\quantumcorner_\lambda\left(
	u_1,\dots,u_N, q^{-\frac{1}{2}}t^{-\frac{1}{2}}u_{N+1}, \dots, q^{-\frac{1}{2}}t^{-\frac{1}{2}}u_{N+M}, q^{\frac{1}{2}}t^{-1}u_{N+M+1}, \dots,q^{\frac{1}{2}}t^{-1}u_{N+M+L} ; q,t
	\right)
	\notag 
\end{align}
\normalsize
Here 
\begin{align}
	\cals{N}_{\lambda}(z_1,\dots,z_k )
	&:= 
	\prod_{1 \leq c < d \leq \ell(\lambda)}
	\prod_{
		\substack{
			i \in I^{(c)}
			\\
			j \in I^{(d)}
		}
	}
	\Delta\left(	q_3^{-\frac{1}{2}}\frac{z_j}{z_i}				\right)^{-1}, 
\end{align}
where $\displaystyle \Delta(z) := 
\frac{
	(1 - q_1q_3^{\frac{1}{2}}z)
	(1 - q^{-1}_1q_3^{-\frac{1}{2}}z)
}{
	(1 - q_3^{\frac{1}{2}}z)
	(1 - q_3^{-\frac{1}{2}}z)
}$ and 
\begin{align*}
	I^{(1)} &= \{1,\dots,\lambda_1\},
	\\
	I^{(2)} &= \{\lambda_1 + 1,\dots,\lambda_1 + \lambda_2\},
	\\
	&\vdots
	\\
	I^{(\ell(\lambda))} &= \{\sum_{j = 1}^{\ell(\lambda) - 1}\lambda_j+ 1, \cdots, \sum_{j = 1}^{\ell(\lambda)}\lambda_j\}. 
\end{align*}
\end{thm}

\begin{lem}
\label{lemm42-1252-25jul}
Let $\vec{c} = (3^N1^M2^L)$, $\vec{u} = (u_1,\dots,u_{N+M+L})$, and let $\lambda \in \operatorname{Par}(k)$. Then, 

\footnotesize
\begin{align}
	&
	\lim_{\xi \rightarrow t^{-1}}\,\,
	(
	\dualmap
	\comp 
	\bigg|_{
		\substack{
			q_1 = q, \\
			q_2 = q^{-1}t,\\
			q_3 = t^{-1} \\
		}
	}
	)
	\left(
	\cals{N}_{\lambda}(z_1,\dots,z_k )
	\times
	\prod_{1 \leq i < j \leq k}f^{\vec{c}}_{11}\left(\frac{z_j}{z_i} \right)
	\times
	\langle 0 |\widetilde{T}^{\vec{c},\vec{u}}_{1}(z_1 )\cdots \widetilde{T}^{\vec{c},\vec{u}}_{1}(z_k )|0\rangle
	\right)
	\\
	&= 
	\underbrace{				
		\sum_{i_1 = 1}^{N+M+L}
		\cdots
		\sum_{i_k = 1}^{N+M+L}
	}_{
		T(i_1,\dots,i_k) \in 
		\operatorname{RSSYTT}(N,M,L;\lambda)
	}
	\Biggl\{
	u_{i_1}\cdots u_{i_k}
	\times
	\left(
	\frac{
		q^{\frac{1}{2}} - q^{-\frac{1}{2}}
	}{
		t^{-\frac{1}{2}} - t^{\frac{1}{2}}
	}
	\right)^{|T_1|}
	\times
	\left(
	\frac{
		(q^{-1}t)^{\frac{1}{2}} - (qt^{-1})^{\frac{1}{2}}
	}{
		t^{-\frac{1}{2}} - t^{\frac{1}{2}}
	}
	\right)^{|T_2|}
	\notag 
	\\
	&\hspace{3.3cm}\times
	\lim_{\xi \rightarrow t^{-1}}\,\,
	(
	\dualmap
	\comp 
	\bigg|_{
		\substack{
			q_1 = q, \\
			q_2 = q^{-1}t,\\
			q_3 = t^{-1} \\
		}
	}
	)
	\bigg[
	\cals{N}_\lambda(z_1,\dots,z_k)
	\times
	\prod_{1 \leq a < b \leq k}
	\cals{D}^{(i_a,i_b)}\left(
	\frac{z_b}{z_a}
	; q, t
	\right)
	\bigg]
	\Biggr\}
	\notag 
	\\
	&= 
	\underbrace{				
		\sum_{i_1 = 1}^{N+M+L}
		\cdots
		\sum_{i_k = 1}^{N+M+L}
	}_{
		T(i_1,\dots,i_k) \in 
		\operatorname{RSSYTT}(N,M,L;\lambda)
	}
	\Biggl\{
	u_{i_1}\cdots u_{i_k}
	\times
	\left(
	\frac{
		q^{\frac{1}{2}} - q^{-\frac{1}{2}}
	}{
		t^{-\frac{1}{2}} - t^{\frac{1}{2}}
	}
	\right)^{|T_1|}
	\times
	\left(
	\frac{
		(q^{-1}t)^{\frac{1}{2}} - (qt^{-1})^{\frac{1}{2}}
	}{
		t^{-\frac{1}{2}} - t^{\frac{1}{2}}
	}
	\right)^{|T_2|}
	\notag 
	\\
	&\hspace{5.3cm}\times
	(
	\widetilde{\Psi}_{\lambda}^{(q,t^{-1})}
	\comp 
	\bigg|_{
		\substack{
			q_1 = q, \\
			q_2 = q^{-1}t,\\
			q_3 = t^{-1} \\
		}
	}
	)
	\bigg[
	\prod_{1 \leq a < b \leq k}
	\cals{C}^{(i_a,i_b)}\left(
	\frac{z_b}{z_a}
	; q, t
	\right)
	\bigg]
	\Biggr\}
	\notag 
\end{align}
\normalsize
where 
\begin{align}
	\cals{D}^{(i,j)}\left(
	\frac{z_b}{z_a}
	; q, t
	\right)
	:= 
	\begin{cases}
		\displaystyle 
		\frac{
			\left(1 - q_1^{-1}\frac{z_b}{z_a}\right)
			\left(1 - q_2^{-1}\frac{z_b}{z_a}\right)
		}{
			\left(1 - q_3\frac{z_b}{z_a}\right)
			\left(1 - \frac{z_b}{z_a}\right)
		}
		\hspace{0.3cm}
		&\text{ if } i < j
		\\
		\displaystyle 
		\frac{
			\left(1 - q_2^{-1}\frac{z_b}{z_a}\right)
			\left(1 - q_2\frac{z_b}{z_a}\right)
		}{
			\left(1 - q_3^{-1}\frac{z_b}{z_a}\right)
			\left(1 - q_3\frac{z_b}{z_a}\right)
		}
		\hspace{0.3cm}
		&\text{ if } i = j = \text{ hyper number }
		\\
		\displaystyle 
		\frac{
			\left(1 - q_1^{-1}\frac{z_b}{z_a}\right)
			\left(1 - q_1\frac{z_b}{z_a}\right)
		}{
			\left(1 - q_3^{-1}\frac{z_b}{z_a}\right)
			\left(1 - q_3\frac{z_b}{z_a}\right)
		}
		\hspace{0.3cm}
		&\text{ if } i = j = \text{ super number }
		\\
		1
		\hspace{0.3cm}
		&\text{ if } i = j = \text{ ordinary number }
		\\
		\displaystyle
		\frac{
			\left(1 - q_1\frac{z_b}{z_a}\right)
			\left(1 - q_2\frac{z_b}{z_a}\right)
		}{
			\left(1 - q_3^{-1}\frac{z_b}{z_a}\right)
			\left(1 - \frac{z_b}{z_a}\right)
		}
		\hspace{0.3cm}
		&\text{ if } i > j 
	\end{cases}
\end{align}
and 
\begin{align}
	\cals{C}^{(i_a,i_b)}\left(
	\frac{z_b}{z_a}
	; q, t
	\right)
	:= 
	\begin{cases}
		\cals{D}^{(i_a,i_b)}\left(
		\frac{z_b}{z_a}
		; q, t
		\right)
		\hspace{3cm}
		\text{ if }
		\operatorname{row}(a) = \operatorname{row}(b)
		\\
		\cals{D}^{(i_a,i_b)}\left(
		\frac{z_b}{z_a}
		; q, t
		\right)
		\times
		\Delta\left(	q_3^{-\frac{1}{2}}\frac{z_b}{z_a}				\right)^{-1} 
		\hspace{0.1cm}
		\text{ if }
		\operatorname{row}(a) < \operatorname{row}(b)
	\end{cases}
\label{eqn36-1808-31jul}
\end{align}
\end{lem}
\begin{proof}
The proof of this lemma is relegated to Appendix \ref{appA-1254}. 
\end{proof}

From \textbf{Lemma \ref{lemm42-1252-25jul}}
and \textbf{Definition \ref{defn314-1234-25jul}}, it is sufficient to prove \textbf{Lemma \ref{lemm45-1458-25jul}} below to establish the statement in \textbf{Theorem \ref{thm43-1458-25jul}}. 

\begin{lem}
\label{lemm45-1458-25jul}
For each $T = T(i_1,\dots,i_k;\lambda) \in \operatorname{RSSYTT}(N,M,L;\lambda)$, we have 

\footnotesize
\begin{align}
	&(
	\widetilde{\Psi}_{\lambda}^{(q,t^{-1})}
	\comp 
	\bigg|_{
		\substack{
			q_1 = q, \\
			q_2 = q^{-1}t,\\
			q_3 = t^{-1} \\
		}
	}
	)
	\bigg[
	\prod_{1 \leq a < b \leq k}
	\cals{C}^{(i_a,i_b)}\left(
	\frac{z_b}{z_a}
	; q, t
	\right)
	\bigg]
	\\
	&=  
	\cals{A}_2(T;q,t) \times 
	\psi_{T^{\prime}_{1}}(t,q) \times \psi_{T_0}(q,t)
	\times
	\frac{
		H(\operatorname{sh}(T_1 \sqcup T_2),q,t^{-1})
	}{
		H(\operatorname{sh}(T_1 \sqcup T_2)^\prime,t^{-1},q)
	}
	\times
	\left(	\frac{t^{-1} - 1}{q-1}		\right)^{|T_1|}
	\notag
	\times 
	\prod_{\square \in \operatorname{sh}(T_2)}
	\frac{
		(t^{-1})^{\ell_{\operatorname{sh}(T_2)}(\square) + 1} - q^{a_{\operatorname{sh}(T_2)}(\square)}
	}{
		q^{a_{\operatorname{sh}(T_2)}(\square) + 1} - (t^{-1})^{\ell_{\operatorname{sh}(T_2)}(\square)}
	}
	\notag 
\end{align}
\normalsize
\end{lem}

First note that 
\footnotesize
\begin{align}
	&(
	\widetilde{\Psi}_{\lambda}^{(q,t^{-1})}
	\comp 
	\bigg|_{
		\substack{
			q_1 = q, \\
			q_2 = q^{-1}t,\\
			q_3 = t^{-1} \\
		}
	}
	)
	\bigg[
	\prod_{1 \leq a < b \leq k}
	\cals{C}^{(i_a,i_b)}\left(
	\frac{z_b}{z_a}
	; q, t
	\right)
	\bigg]
	\label{eqn36-1108-15jul}
	\\
	&=
	(
	\widetilde{\Psi}_{\lambda}^{(q,t^{-1})}
	\comp 
	\bigg|_{
		\substack{
			q_1 = q, \\
			q_2 = q^{-1}t,\\
			q_3 = t^{-1} \\
		}
	}
	)
	\bigg[
	\underbrace{	\prod_{1 \leq a < b \leq k}				}_{\text{ $a$ and $b$ are boxes in $T_2$}}
	\cals{C}^{(i_a,i_b)}\left(
	\frac{z_b}{z_a}
	; q, t
	\right)
	\bigg]
	\label{eqn48-1512-25jul}
	\\
	&\hspace{0.3cm}\times 
	(
	\widetilde{\Psi}_{\lambda}^{(q,t^{-1})}
	\comp 
	\bigg|_{
		\substack{
			q_1 = q, \\
			q_2 = q^{-1}t,\\
			q_3 = t^{-1} \\
		}
	}
	)
	\bigg[
	\underbrace{	\prod_{1 \leq a < b \leq k}				}_{\text{ $a$ and $b$ are boxes in $T_0 \sqcup T_1$}}
	\cals{C}^{(i_a,i_b)}\left(
	\frac{z_b}{z_a}
	; q, t
	\right)
	\bigg]
	\label{eqn49-1512-25jul}
	\\
	&\hspace{0.3cm}\times 
	(
	\widetilde{\Psi}_{\lambda}^{(q,t^{-1})}
	\comp 
	\bigg|_{
		\substack{
			q_1 = q, \\
			q_2 = q^{-1}t,\\
			q_3 = t^{-1} \\
		}
	}
	)
	\bigg[
	\underbrace{	\prod_{1 \leq a < b \leq k}				}_{
		\substack{
			(1) \,\, \text{ $a$ is a box in $T_0 \sqcup T_1$}
			\\
			(2) \,\, \text{ $b$ is a box in $T_2$}
		}
	}
	\cals{C}^{(i_a,i_b)}\left(
	\frac{z_b}{z_a}
	; q, t
	\right)
	\bigg]
	\label{eqn310-1137-27jul}
	\\
	&\hspace{0.3cm}\times 
	(
	\widetilde{\Psi}_{\lambda}^{(q,t^{-1})}
	\comp 
	\bigg|_{
		\substack{
			q_1 = q, \\
			q_2 = q^{-1}t,\\
			q_3 = t^{-1} \\
		}
	}
	)
	\bigg[
	\underbrace{	\prod_{1 \leq a < b \leq k}				}_{
		\substack{
			(1) \,\, \text{ $a$ is a box in $T_2$}
			\\
			(2) \,\, \text{ $b$ is a box in $T_0 \sqcup T_1$}
		}
	}
	\cals{C}^{(i_a,i_b)}\left(
	\frac{z_b}{z_a}
	; q, t
	\right)
	\bigg]
	\label{eqn311-1137-27jul}
\end{align}
\normalsize

\subsection{Analysis on the Quantity in 
\eqref{eqn48-1512-25jul}}

Since $T = T(i_1,\dots,i_k;\lambda) \in \operatorname{RSSYTT}(N,M,L;\lambda)$, the hyper numbers are $N+M+1,\dots,N+M+L$. Thus, we obtain that 

\footnotesize
\begin{align}
	&(
	\widetilde{\Psi}_{\lambda}^{(q,t^{-1})}
	\comp 
	\bigg|_{
		\substack{
			q_1 = q, \\
			q_2 = q^{-1}t,\\
			q_3 = t^{-1} \\
		}
	}
	)
	\bigg[
	\underbrace{	\prod_{1 \leq a < b \leq k}				}_{\text{ $a$ and $b$ are boxes in $T_2$}}
	\cals{C}^{(i_a,i_b)}\left(
	\frac{z_b}{z_a}
	; q, t
	\right)
	\bigg]
	\label{eqn312-1501-1aug}
	\\
	&= 
	\prod_{\zeta = 1}^{\ell(T_2)}
	\prod_{d = N+M+1}^{N+M+L}
	(
	\widetilde{\Psi}_{\lambda}^{(q,t^{-1})}
	\comp 
	\bigg|_{
		\substack{
			q_1 = q, \\
			q_2 = q^{-1}t,\\
			q_3 = t^{-1} \\
		}
	}
	)
	\bigg[
	\underbrace{	\prod_{1 \leq a < b \leq k}				}_{
		\substack{
			(1) \,\, \text{ $a$ and $b$ are boxes in $T_2$}
			\\
			(2) \,\, \operatorname{row}(a) = \zeta
			\\
			(3) \,\, i_a = d 
		}	
	}
	\cals{C}^{(i_a,i_b)}\left(
	\frac{z_b}{z_a}
	; q, t
	\right)
	\bigg]. 
	\notag 
\end{align}
\normalsize
It is clear that 

\footnotesize
\begin{align}
&(
\widetilde{\Psi}_{\lambda}^{(q,t^{-1})}
\comp 
\bigg|_{
	\substack{
		q_1 = q, \\
		q_2 = q^{-1}t,\\
		q_3 = t^{-1} \\
	}
}
)
\bigg[
\underbrace{	\prod_{1 \leq a < b \leq k}				}_{
	\substack{
		(1) \,\, \text{ $a$ and $b$ are boxes in $T_2$}
		\\
		(2) \,\, \operatorname{row}(a) = \zeta
		\\
		(3) \,\, i_a = d 
	}	
}
\cals{C}^{(i_a,i_b)}\left(
\frac{z_b}{z_a}
; q, t
\right)
\bigg]
\label{eqn313-1501-1aug}
\\
&= 
(
\widetilde{\Psi}_{\lambda}^{(q,t^{-1})}
\comp 
\bigg|_{
	\substack{
		q_1 = q, \\
		q_2 = q^{-1}t,\\
		q_3 = t^{-1} \\
	}
}
)
\bigg[
\underbrace{	\prod_{1 \leq a < b \leq k}				}_{
	\substack{
		(1) \,\, \operatorname{row}(a) = \operatorname{row}(b) = \zeta, 
		\\
		(2) \,\, i_a = i_b = d
	}	
}
\cals{C}^{(i_a,i_b)}\left(
\frac{z_b}{z_a}
; q, t
\right)
\bigg]
\notag 
\\
&\hspace{0.5cm}\times 
\prod_{\alpha = N+M+1}^{d - 1} 
(
\widetilde{\Psi}_{\lambda}^{(q,t^{-1})}
\comp 
\bigg|_{
	\substack{
		q_1 = q, \\
		q_2 = q^{-1}t,\\
		q_3 = t^{-1} \\
	}
}
)
\bigg[
\underbrace{	\prod_{1 \leq a < b \leq k}				}_{
	\substack{
		(1) \,\, \operatorname{row}(a) = \operatorname{row}(b) = \zeta, 
		\\
		(2) \,\, i_a =  d
		\\
		(3) \,\, i_b = \alpha
	}	
}
\cals{C}^{(i_a,i_b)}\left(
\frac{z_b}{z_a}
; q, t
\right)
\bigg]
\notag 
\\
&\hspace{0.5cm}\times 
\prod_{\tau = \zeta + 1}^{\ell(T_2)} 
(
\widetilde{\Psi}_{\lambda}^{(q,t^{-1})}
\comp 
\bigg|_{
	\substack{
		q_1 = q, \\
		q_2 = q^{-1}t,\\
		q_3 = t^{-1} \\
	}
}
)
\bigg[
\underbrace{	\prod_{1 \leq a < b \leq k}				}_{
	\substack{
		(1) \,\, \operatorname{row}(a) = \zeta, 
		\\
		(2) \,\,  \operatorname{row}(b) = \tau, 
		\\
		(3) \,\, i_a = i_b = d
	}	
}
\cals{C}^{(i_a,i_b)}\left(
\frac{z_b}{z_a}
; q, t
\right)
\bigg]
\notag 
\\
&\hspace{0.5cm}\times 
\prod_{\tau = \zeta + 1}^{\ell(T_2)}
\prod_{p = d+1}^{N+M+L}
(
\widetilde{\Psi}_{\lambda}^{(q,t^{-1})}
\comp 
\bigg|_{
	\substack{
		q_1 = q, \\
		q_2 = q^{-1}t,\\
		q_3 = t^{-1} \\
	}
}
)
\bigg[
\underbrace{	\prod_{1 \leq a < b \leq k}				}_{
	\substack{
		(1) \,\, \operatorname{row}(a) = \zeta, 
		\\
		(2) \,\,  \operatorname{row}(b) = \tau, 
		\\
		(3) \,\, i_a = d
		\\
		(4) \,\, i_b = p
	}	
}
\cals{C}^{(i_a,i_b)}\left(
\frac{z_b}{z_a}
; q, t
\right)
\bigg]. 
\notag 
\end{align}
\normalsize

\begin{lem}
\label{lemm36-1329}
Suppose that $T = T(i_1,\dots,i_k;\lambda) \in \operatorname{RSSYTT}(N,M,L;\lambda)$. Let $\zeta \in \{1,\dots,\ell_{T_2}\}$ and let $d \in \{N+M+1,\dots,N+M+L\}$. Then, we have 

\footnotesize
\begin{align}
	&(
	\widetilde{\Psi}_{\lambda}^{(q,t^{-1})}
	\comp 
	\bigg|_{
		\substack{
			q_1 = q, \\
			q_2 = q^{-1}t,\\
			q_3 = t^{-1} \\
		}
	}
	)
	\bigg[
	\underbrace{	\prod_{1 \leq a < b \leq k}				}_{
		\substack{
			(1) \,\, \operatorname{row}(a) = \operatorname{row}(b) = \zeta, 
			\\
			(2) \,\, i_a = i_b = d
		}	
	}
	\cals{C}^{(i_a,i_b)}\left(
	\frac{z_b}{z_a}
	; q, t
	\right)
	\bigg]
	= 
	\prod_{j = 1}^{\theta_{\zeta,d} - 1}
	\frac{
		(1 - qt^{-1}q^{\theta_{\zeta,d} - j})(1 - t)
	}{
		(1 - t^{-1}q)(1 - tq^{\theta_{\zeta,d} - j})
	}
\end{align}
\normalsize
\end{lem}

\begin{lem}
Suppose that $T = T(i_1,\dots,i_k;\lambda) \in \operatorname{RSSYTT}(N,M,L;\lambda)$. Let $\zeta \in \{1,\dots,\ell_{T_2}\}$ and let $d, \alpha \in \{N+M+1,\dots,N+M+L\}$ where $d > \alpha$. Then, we have 

\footnotesize
\begin{align}
	&(
	\widetilde{\Psi}_{\lambda}^{(q,t^{-1})}
	\comp 
	\bigg|_{
		\substack{
			q_1 = q, \\
			q_2 = q^{-1}t,\\
			q_3 = t^{-1} \\
		}
	}
	)
	\bigg[
	\underbrace{	\prod_{1 \leq a < b \leq k}				}_{
		\substack{
			(1) \,\, \operatorname{row}(a) = \operatorname{row}(b) = \zeta, 
			\\
			(2) \,\, i_a =  d
			\\
			(3) \,\, i_b = \alpha
		}	
	}
	\cals{C}^{(i_a,i_b)}\left(
	\frac{z_b}{z_a}
	; q, t
	\right)
	\bigg]
	\\
	&= 
	\prod_{\omega = 1}^{\theta_{\zeta,\alpha}}
	\frac{
		\left(	1 - q^{	\sum_{\gamma = \alpha + 1}^{N + M + L} \theta_{\zeta,\gamma}	+ \omega - \sum_{\gamma = d+1}^{N + M + L}\theta_{\zeta,\gamma} 		}			\right)
		\left(	1 - tq^{	\sum_{\gamma = \alpha + 1}^{N + M + L} \theta_{\zeta,\gamma}	+ \omega - (\sum_{\gamma = d}^{N + M + L}\theta_{\zeta,\gamma} 	+ 1 )	}		\right)
	}{
		\left(	1 - q^{	\sum_{\gamma = \alpha + 1}^{N + M + L} \theta_{\zeta,\gamma}	+ \omega - \sum_{\gamma = d}^{N + M + L}\theta_{\zeta,\gamma} 		}			\right)
		\left(	1 - 	tq^{	\sum_{\gamma = \alpha + 1}^{N + M + L} \theta_{\zeta,\gamma}	+ \omega - (\sum_{\gamma = d + 1}^{N + M + L}\theta_{\zeta,\gamma} 	+ 1 )	}			\right)
	}
	\notag 
\end{align}
\normalsize
\end{lem}

\begin{lem}
Suppose that $T = T(i_1,\dots,i_k;\lambda) \in \operatorname{RSSYTT}(N,M,L;\lambda)$. Let $\zeta, \tau \in \{1,\dots,\ell_{T_2}\}$ where $\tau > \zeta$, 
and let $d \in \{N+M+1,\dots,N+M+L\}$. Then, we have 

\footnotesize
\begin{align}
	&(
	\widetilde{\Psi}_{\lambda}^{(q,t^{-1})}
	\comp 
	\bigg|_{
		\substack{
			q_1 = q, \\
			q_2 = q^{-1}t,\\
			q_3 = t^{-1} \\
		}
	}
	)
	\bigg[
	\underbrace{	\prod_{1 \leq a < b \leq k}				}_{
		\substack{
			(1) \,\, \operatorname{row}(a) = \zeta, 
			\\
			(2) \,\,  \operatorname{row}(b) = \tau, 
			\\
			(3) \,\, i_a = i_b = d
		}	
	}
	\cals{C}^{(i_a,i_b)}\left(
	\frac{z_b}{z_a}
	; q, t
	\right)
	\bigg]
	\\
	&= 
	\prod_{\omega = 1}^{\theta_{\tau,d}}
	\frac{
		\left(		1 - t^{-1} (t^{-1})^{\tau - \zeta}	q^{	\sum_{\gamma = d+1}^{N + M + L}\theta_{\tau,\gamma}	+ \omega - 	\sum_{\gamma = d+ 1}^{N + M + L} \theta_{\zeta,\gamma}				}			\right)
		\left(
		1 - (t^{-1})^{\tau - \zeta}	q^{	\sum_{\gamma = d+1}^{N + M + L}\theta_{\tau,\gamma}	+ \omega - 	\sum_{\gamma = d}^{N + M + L} \theta_{\zeta,\gamma}				}		
		\right)
	}{
		\left(		1 - t^{-1} (t^{-1})^{\tau - \zeta}	q^{	\sum_{\gamma = d+1}^{N + M + L}\theta_{\tau,\gamma}	+ \omega - 	\sum_{\gamma = d}^{N + M + L} \theta_{\zeta,\gamma}				}			\right)
		\left(
		1 - (t^{-1})^{\tau - \zeta}	q^{	\sum_{\gamma = d+1}^{N + M + L}\theta_{\tau,\gamma}	+ \omega - 	\sum_{\gamma = d+1}^{N + M + L} \theta_{\zeta,\gamma}				}		
		\right)
	}
	\notag 
\end{align}
\normalsize
\end{lem}

\begin{lem}
\label{lemm39-1329-1aug}
Suppose that $T = T(i_1,\dots,i_k;\lambda) \in \operatorname{RSSYTT}(N,M,L;\lambda)$. Let $\zeta, \tau \in \{1,\dots,\ell_{T_2}\}$ where $\tau > \zeta$, 
and let $d, p \in \{N+M+1,\dots,N+M+L\}$ where $d < p$. Then, we have 

\footnotesize
\begin{align}
	&(
	\widetilde{\Psi}_{\lambda}^{(q,t^{-1})}
	\comp 
	\bigg|_{
		\substack{
			q_1 = q, \\
			q_2 = q^{-1}t,\\
			q_3 = t^{-1} \\
		}
	}
	)
	\bigg[
	\underbrace{	\prod_{1 \leq a < b \leq k}				}_{
		\substack{
			(1) \,\, \operatorname{row}(a) = \zeta, 
			\\
			(2) \,\,  \operatorname{row}(b) = \tau, 
			\\
			(3) \,\, i_a = d
			\\
			(4) \,\, i_b = p
		}	
	}
	\cals{C}^{(i_a,i_b)}\left(
	\frac{z_b}{z_a}
	; q, t
	\right)
	\bigg]
	\\
	&= 
	\prod_{\Xi = 1}^{\theta_{\zeta,d}}
	\prod_{\omega = 1}^{\theta_{\tau,p}}
	\Biggl\{
	\notag 
	\\
	&\frac{
		\left(	1 - q^{-1}					
		(t^{-1})^{\tau - \zeta} q^{	\sum_{\gamma = p + 1}^{N + M + L} \theta_{\tau,\gamma}	+ \omega	- (\sum_{\gamma = d + 1}^{N + M + L} \theta_{\zeta,\gamma} + \Xi)		}	
		\right)
		\left(	1 - qt^{-1}			
		(t^{-1})^{\tau - \zeta} q^{	\sum_{\gamma = p + 1}^{N + M + L} \theta_{\tau,\gamma}	+ \omega	- (\sum_{\gamma = d + 1}^{N + M + L} \theta_{\zeta,\gamma} + \Xi)		}
		\right)
	}{
		\left(	1 - q			
		(t^{-1})^{\tau - \zeta} q^{	\sum_{\gamma = p + 1}^{N + M + L} \theta_{\tau,\gamma}	+ \omega	- (\sum_{\gamma = d + 1}^{N + M + L} \theta_{\zeta,\gamma} + \Xi)		}
		\right)
		\left(	1 - q^{-1}t			
		(t^{-1})^{\tau - \zeta} q^{	\sum_{\gamma = p + 1}^{N + M + L} \theta_{\tau,\gamma}	+ \omega	- (\sum_{\gamma = d + 1}^{N + M + L} \theta_{\zeta,\gamma} + \Xi)		}
		\right)
	}
	\notag 
	\\
	&\times
	\frac{
		\left(	1 - t		
		(t^{-1})^{\tau - \zeta} q^{	\sum_{\gamma = p + 1}^{N + M + L} \theta_{\tau,\gamma}	+ \omega	- (\sum_{\gamma = d + 1}^{N + M + L} \theta_{\zeta,\gamma} + \Xi)		}
		\right)
	}{
		\left(	1 - t^{-1}			
		(t^{-1})^{\tau - \zeta} q^{	\sum_{\gamma = p + 1}^{N + M + L} \theta_{\tau,\gamma}	+ \omega	- (\sum_{\gamma = d + 1}^{N + M + L} \theta_{\zeta,\gamma} + \Xi)		}
		\right)
	}
	\Biggr\}
	\notag 
\end{align}
\normalsize
\end{lem}

\begin{lem}
\label{lemm36-1211-27jul}
For each $T = T(i_1,\dots,i_k;\lambda) \in \operatorname{RSSYTT}(N,M,L;\lambda)$, we have 
\begin{align}
	(
	\widetilde{\Psi}_{\lambda}^{(q,t^{-1})}
	\comp 
	\bigg|_{
		\substack{
			q_1 = q, \\
			q_2 = q^{-1}t,\\
			q_3 = t^{-1} \\
		}
	}
	)
	\bigg[
	\underbrace{	\prod_{1 \leq a < b \leq k}				}_{\text{ $a$ and $b$ are boxes in $T_2$}}
	\cals{C}^{(i_a,i_b)}\left(
	\frac{z_b}{z_a}
	; q, t
	\right)
	\bigg]
	= 
	\cals{A}_2(T;q,t)
\end{align}
\end{lem}
\begin{proof}
From equations \eqref{eqn312-1501-1aug}, \eqref{eqn313-1501-1aug} and \textbf{Lemmas \ref{lemm36-1329}} - \textbf{\ref{lemm39-1329-1aug}}, we can see that the formula for 
\begin{align}
(
\widetilde{\Psi}_{\lambda}^{(q,t^{-1})}
\comp 
\bigg|_{
	\substack{
		q_1 = q, \\
		q_2 = q^{-1}t,\\
		q_3 = t^{-1} \\
	}
}
)
\bigg[
\underbrace{	\prod_{1 \leq a < b \leq k}				}_{\text{ $a$ and $b$ are boxes in $T_2$}}
\cals{C}^{(i_a,i_b)}\left(
\frac{z_b}{z_a}
; q, t
\right)
\bigg]
\end{align}
is the same as the formula for $\cals{A}_2(T;q,t)$ as given in \textbf{Lemma \ref{prp317-2237-16aug}}. 
\end{proof}

\subsection{Analysis on the Quantity in  \eqref{eqn49-1512-25jul}}

The goal of this subsection is to find a formula for 
\footnotesize
\begin{align}
(
\widetilde{\Psi}_{\lambda}^{(q,t^{-1})}
\comp 
\bigg|_{
	\substack{
		q_1 = q, \\
		q_2 = q^{-1}t,\\
		q_3 = t^{-1} \\
	}
}
)
\bigg[
\underbrace{	\prod_{1 \leq a < b \leq k}				}_{\text{ $a$ and $b$ are boxes in $T_0 \sqcup T_1$}}
\cals{C}^{(i_a,i_b)}\left(
\frac{z_b}{z_a}
; q, t
\right)
\bigg]. 
\end{align}
\normalsize
We will explain the idea for finding this formula through an example of a reverse SSYTT below. Consider 

\footnotesize
\begin{align}
	T = 
	\begin{ytableau}
		\textcolor{red}{8}  & \textcolor{red}{8} & \textcolor{red}{8} & \textcolor{red}{6} & \textcolor{red}{6} & 3 \\
		\textcolor{red}{8}  & \textcolor{red}{7} & \textcolor{red}{7} & \textcolor{red}{6} & \textcolor{blue}{4} & 2 \\
		\textcolor{red}{7}  & \textcolor{red}{7} & \textcolor{red}{6} & \textcolor{blue}{5} \\ 
		\textcolor{red}{6} & \textcolor{blue}{5} & \textcolor{blue}{4} & 3 \\ 
		\textcolor{red}{6} & \textcolor{blue}{5} 
	\end{ytableau}
\label{eqn410-1632-25jul}
\end{align}
\normalsize
Note that in this paper, we follow the convention of \cite{CSW2025} to denote hyper numbers, super numbers, and ordinary numbers in reverse SSYTT with red, blue, and black colors, respectively. 
For \eqref{eqn410-1632-25jul}, we have 

\footnotesize
\begin{align}
	T_0 \sqcup T_1 = 
	\begin{ytableau}
		\none  & \none & \none & \none & \none & 3 \\
		\none  & \none & \none & \none & \textcolor{blue}{4} & 2 \\
		\none  & \none & \none & \textcolor{blue}{5} \\ 
		\none & \textcolor{blue}{5} & \textcolor{blue}{4} & 3 \\ 
		\none & \textcolor{blue}{5} 
	\end{ytableau}
	\label{eqn310-1513-14jul}
\end{align}
\normalsize
We observe that in the general case of $T \in \operatorname{RSSYTT}(N,M,L;\lambda)$, $T_0 \sqcup T_1$ is a reverse SSYBT with a skew shape. In our previous paper \cite{CSW2025}, we calculated the contribution from reverse SSYBT with a full shape. Therefore, we can not directly use the results from \cite{CSW2025}. To utilize the results from the previous paper, we need to construct a full-shape reverse SSYBT from this skew-shape reverse SSYBT. The question is, how do we do that ? 

The method for constructing a full-shape reverse SSYBT from a skew-shape reverse SSYBT is as follows:
\begin{enumerate}[(1)]
\item Determine the number of rows in $T_2$. We denote the number of rows of $T_2$ by $\ell_{\operatorname{sh}(T_2)}$
\item Add super numbers $\textcolor{blue}{N + M+1}, \dots, \textcolor{blue}{N + M+\ell_{\operatorname{sh}(T_2)}}$ to $T_0 \sqcup T_1$, proceeding from right to left columns, and from smallest to largest numbers. The number added in the same column must be identical. 
\end{enumerate}

If we apply this method to $T_0 \sqcup T_1$ in equation \eqref{eqn310-1513-14jul}, we obtain 

\footnotesize
\begin{align}
	\begin{ytableau}
		\none  & \none & \none & \none & \none & 3 \\
		\none  & \none & \none & \none & \textcolor{blue}{4} & 2 \\
		\none  & \none & \none & \textcolor{blue}{5} \\ 
		\none & \textcolor{blue}{5} & \textcolor{blue}{4} & 3 \\ 
		\none & \textcolor{blue}{5} 
	\end{ytableau}
	\Longrightarrow
	\begin{ytableau}
		\textcolor{blue}{10}  & \textcolor{blue}{9} & \textcolor{blue}{8} & \textcolor{blue}{7} & \textcolor{blue}{6} & 3 \\
		\textcolor{blue}{10}  & \textcolor{blue}{9} & \textcolor{blue}{8} & \textcolor{blue}{7} & \textcolor{blue}{4} & 2 \\
		\textcolor{blue}{10}  & \textcolor{blue}{9} & \textcolor{blue}{8} & \textcolor{blue}{5} \\ 
		\textcolor{blue}{10} & \textcolor{blue}{5} & \textcolor{blue}{4} & 3 \\ 
		\textcolor{blue}{10} & \textcolor{blue}{5} 
	\end{ytableau}
\end{align}
\normalsize
That is, initially we had super numbers $\textcolor{blue}{4}, \textcolor{blue}{5}$, and since $\ell_{\operatorname{sh}(T_2)} = 5$, we need to add the super numbers $\textcolor{blue}{6}, \textcolor{blue}{7}, \textcolor{blue}{8}, \textcolor{blue}{9}, \textcolor{blue}{10}$. 

With this example in mind, we are ready to analyze the general case. First, let us define some notation that we will use. 

\begin{nota}\mbox{}
\begin{enumerate}[(1)]
\item We denote the $T_0 \sqcup T_1$ with added boxes by the notation $(T_0 \sqcup T_1)^{\star\star}$
\item We denote the added part by the symbol $(T_0 \sqcup T_1)^{\operatorname{add}}$
\end{enumerate}
\end{nota}

For example, for $T$ as in \eqref{eqn410-1632-25jul}, we have 
\footnotesize
\begin{align}
	T_0 \sqcup T_1
	&= 
	\begin{ytableau}
		\none  & \none & \none & \none & \none & 3 \\
		\none  & \none & \none & \none & \textcolor{blue}{4} & 2 \\
		\none  & \none & \none & \textcolor{blue}{5} \\ 
		\none & \textcolor{blue}{5} & \textcolor{blue}{4} & 3 \\ 
		\none & \textcolor{blue}{5} 
	\end{ytableau}
	\\
	(T_0 \sqcup T_1)^{\star\star}
	&= 
	\begin{ytableau}
		\textcolor{blue}{10}  & \textcolor{blue}{9} & \textcolor{blue}{8} & \textcolor{blue}{7} & \textcolor{blue}{6} & 3 \\
		\textcolor{blue}{10}  & \textcolor{blue}{9} & \textcolor{blue}{8} & \textcolor{blue}{7} & \textcolor{blue}{4} & 2 \\
		\textcolor{blue}{10}  & \textcolor{blue}{9} & \textcolor{blue}{8} & \textcolor{blue}{5} \\ 
		\textcolor{blue}{10} & \textcolor{blue}{5} & \textcolor{blue}{4} & 3 \\ 
		\textcolor{blue}{10} & \textcolor{blue}{5} 
	\end{ytableau}
	\\
	(T_0 \sqcup T_1)^{\operatorname{add}}
	&= 
	\begin{ytableau}
		\textcolor{blue}{10}  & \textcolor{blue}{9} & \textcolor{blue}{8} & \textcolor{blue}{7} & \textcolor{blue}{6}  \\
		\textcolor{blue}{10}  & \textcolor{blue}{9} & \textcolor{blue}{8} & \textcolor{blue}{7}  \\
		\textcolor{blue}{10}  & \textcolor{blue}{9} & \textcolor{blue}{8} \\ 
		\textcolor{blue}{10} \\ 
		\textcolor{blue}{10} 
	\end{ytableau}
\end{align}
\normalsize

\begin{prop}
\label{prop48-1657-25Jul}
Let $(N,M,L) \in (\bb{Z}^{\geq 0} \times \bb{Z}^{\geq 0} \times \bb{Z}^{\geq 0})\backslash\{(0,0,0)\}$, and $T \in \operatorname{RSSYTT}(N,M,L;\lambda)$. Then, the following statements hold. 
\begin{enumerate}[(1)]
\item $(T_0 \sqcup T_1)^{\star\star} \in \operatorname{RSSYTT}(N,M + \ell_{\operatorname{sh}(T_2)},0;\lambda)$
\item $(T_0 \sqcup T_1)^{\operatorname{add}} \in \operatorname{RSSYTT}(N,M+\ell_{\operatorname{sh}(T_2)},0;\operatorname{sh}(T_2))$ and $(T_0 \sqcup T_1)^{\operatorname{add}}$ contains only super numbers $N+M+1,\dots,N+M+\ell_{\operatorname{sh}(T_2)}$. 
\end{enumerate} 
\end{prop}
\begin{proof}
Obvious. 
\end{proof}

It is easy to show that 
\footnotesize
\begin{align}
	&(
	\widetilde{\Psi}_{\lambda}^{(q,t^{-1})}
	\comp 
	\bigg|_{
		\substack{
			q_1 = q, \\
			q_2 = q^{-1}t,\\
			q_3 = t^{-1} \\
		}
	}
	)
	\bigg[
	\underbrace{	\prod_{1 \leq a < b \leq k}				}_{\text{ $a$ and $b$ are boxes in $(T_0 \sqcup T_1)$}}
	\cals{C}^{(i_a,i_b)}\left(
	\frac{z_b}{z_a}
	; q, t
	\right)
	\bigg]
	\label{eqn316-1129-27jul}
	\\
	&= 
	(
	\widetilde{\Psi}_{\lambda}^{(q,t^{-1})}
	\comp 
	\bigg|_{
		\substack{
			q_1 = q, \\
			q_2 = q^{-1}t,\\
			q_3 = t^{-1} \\
		}
	}
	)
	\bigg[
	\underbrace{	\prod_{1 \leq a < b \leq k}				}_{\text{ $a$ and $b$ are boxes in $(T_0 \sqcup T_1)^{\star\star}$}}
	\cals{C}^{(i_a,i_b)}\left(
	\frac{z_b}{z_a}
	; q, t
	\right)
	\bigg]
\notag 
	\\
	&\hspace{0.3cm}\times
	\left(
	(
	\widetilde{\Psi}_{\lambda}^{(q,t^{-1})}
	\comp 
	\bigg|_{
		\substack{
			q_1 = q, \\
			q_2 = q^{-1}t,\\
			q_3 = t^{-1} \\
		}
	}
	)
	\bigg[
	\underbrace{	\prod_{1 \leq a < b \leq k}				}_{\text{ $a$ and $b$ are boxes in $(T_0 \sqcup T_1)^{\operatorname{add}}$}}
	\cals{C}^{(i_a,i_b)}\left(
	\frac{z_b}{z_a}
	; q, t
	\right)
	\bigg]
	\right)^{-1}
\notag 
	\\
	&\hspace{0.3cm}\times
	\left(
	(
	\widetilde{\Psi}_{\lambda}^{(q,t^{-1})}
	\comp 
	\bigg|_{
		\substack{
			q_1 = q, \\
			q_2 = q^{-1}t,\\
			q_3 = t^{-1} \\
		}
	}
	)
	\bigg[
	\underbrace{	\prod_{1 \leq a < b \leq k}				}_{
		\substack{
			(1) \,\, \text{ $a$ is a box in $(T_0 \sqcup T_1)^{\operatorname{add}}$}
			\\
			(2) \,\, \text{ $b$ is a box in $T_0 \sqcup T_1$}
		}
	}
	\cals{C}^{(i_a,i_b)}\left(
	\frac{z_b}{z_a}
	; q, t
	\right)
	\bigg]
	\right)^{-1}
	\notag 
	\\
	&\hspace{0.3cm}\times
	\left(
	(
	\widetilde{\Psi}_{\lambda}^{(q,t^{-1})}
	\comp 
	\bigg|_{
		\substack{
			q_1 = q, \\
			q_2 = q^{-1}t,\\
			q_3 = t^{-1} \\
		}
	}
	)
	\bigg[
	\underbrace{	\prod_{1 \leq a < b \leq k}				}_{
		\substack{
			(1) \,\, \text{ $a$ is a box in $T_0 \sqcup T_1$}
			\\
			(2) \,\, \text{ $b$ is a box in $(T_0 \sqcup T_1)^{\operatorname{add}}$}
		}
	}
	\cals{C}^{(i_a,i_b)}\left(
	\frac{z_b}{z_a}
	; q, t
	\right)
	\bigg]
	\right)^{-1}.
	\notag
\end{align}
\normalsize


\begin{prop}\mbox{}
\label{prp49-1703-25jul}

\footnotesize
\begin{align}
	&(
	\widetilde{\Psi}_{\lambda}^{(q,t^{-1})}
	\comp 
	\bigg|_{
		\substack{
			q_1 = q, \\
			q_2 = q^{-1}t,\\
			q_3 = t^{-1} \\
		}
	}
	)
	\bigg[
	\underbrace{	\prod_{1 \leq a < b \leq k}				}_{\text{ $a$ and $b$ are boxes in $(T_0 \sqcup T_1)^{\star\star}$}}
	\cals{C}^{(i_a,i_b)}\left(
	\frac{z_b}{z_a}
	; q, t
	\right)
	\bigg]
	\\
	&= 
	\prod_{\xi = N + 1}^{N + M + \ell_{\operatorname{sh}(T_2)} - 1}
	\psi_{\left(	(T_0 \sqcup T_1)^{\star\star}_{\operatorname{super part}}			\right)^{\prime \, (\xi)}/\left(	(T_0 \sqcup T_1)^{\star\star}_{\operatorname{super part}}			\right)^{\prime \, (\xi + 1)}}(t,q)
	\times 
	\psi_{
		T_0
	}(q,t)
	\times 
	\frac{
		H(\operatorname{sh}(T_1 \sqcup T_2),q,t^{-1})
	}{
		H(\operatorname{sh}(T_1 \sqcup T_2)^\prime,t^{-1},q)
	}
	\times 
	\left(
	\frac{t^{-1} - 1}{q - 1}
	\right)^{
		|T_1| + |T_2|
	}
	\notag 
\end{align}
\normalsize
\end{prop}
\begin{proof}
We know from \textbf{Proposition \ref{prop48-1657-25Jul}} that $(T_0 \sqcup T_1)^{\star\star} \in \operatorname{RSSYTT}(N,M + \ell_{\operatorname{sh}(T_2)},0;\lambda)$. Then, from \textbf{Lemma 4.5} of \cite{CSW2025}, we know that 

\footnotesize
\begin{align}
	&(
	\widetilde{\Psi}_{\lambda}^{(q,t^{-1})}
	\comp 
	\bigg|_{
		\substack{
			q_1 = q, \\
			q_2 = q^{-1}t,\\
			q_3 = t^{-1} \\
		}
	}
	)
	\bigg[
	\underbrace{	\prod_{1 \leq a < b \leq k}				}_{\text{ $a$ and $b$ are boxes in $(T_0 \sqcup T_1)^{\star\star}$}}
	\cals{C}^{(i_a,i_b)}\left(
	\frac{z_b}{z_a}
	; q, t
	\right)
	\bigg]
	\\
	&= 
	\psi_{
		\left(	(T_0 \sqcup T_1)^{\star\star}_{\operatorname{super part}}			\right)^\prime
	}(t,q)
	\times
	\psi_{
		\left(	(T_0 \sqcup T_1)^{\star\star}_{\operatorname{ordinary part}}			\right)
	}(q,t)
	\times 
	\frac{
		H(\operatorname{sh}(T_1 \sqcup T_2),q,t^{-1})
	}{
		H(\operatorname{sh}(T_1 \sqcup T_2)^\prime,t^{-1},q)
	}
	\times 
	\left(
	\frac{t^{-1} - 1}{q - 1}
	\right)^{
		|T_1| + |T_2|
	}.
	\notag 
\end{align}
\normalsize
From the construction of $(T_0 \sqcup T_1)^{\star\star}$ we have $(T_0 \sqcup T_1)^{\star\star}_{\operatorname{ordinary part}} = T_0$, and from \textbf{Definition \ref{dfn214-2113-26jul}}, we know that 

\footnotesize
\begin{align}
	&\psi_{
		\left(	(T_0 \sqcup T_1)^{\star\star}_{\operatorname{super part}}			\right)^\prime
	}(t,q)
	= 
	\prod_{\xi = N + 1}^{N + M + \ell_{\operatorname{sh}(T_2)} - 1}
	\psi_{\left(	(T_0 \sqcup T_1)^{\star\star}_{\operatorname{super part}}			\right)^{\prime \, (\xi)}/\left(	(T_0 \sqcup T_1)^{\star\star}_{\operatorname{super part}}			\right)^{\prime \, (\xi + 1)}}(t,q)
\end{align}
\normalsize

Therefore, we obtain that 
\footnotesize
\begin{align}
	&(
	\widetilde{\Psi}_{\lambda}^{(q,t^{-1})}
	\comp 
	\bigg|_{
		\substack{
			q_1 = q, \\
			q_2 = q^{-1}t,\\
			q_3 = t^{-1} \\
		}
	}
	)
	\bigg[
	\underbrace{	\prod_{1 \leq a < b \leq k}				}_{\text{ $a$ and $b$ are boxes in $(T_0 \sqcup T_1)^{\star\star}$}}
	\cals{C}^{(i_a,i_b)}\left(
	\frac{z_b}{z_a}
	; q, t
	\right)
	\bigg]
	\\
	&= 
	\prod_{\xi = N + 1}^{N + M + \ell_{\operatorname{sh}(T_2)} - 1}
	\psi_{\left(	(T_0 \sqcup T_1)^{\star\star}_{\operatorname{super part}}			\right)^{\prime \, (\xi)}/\left(	(T_0 \sqcup T_1)^{\star\star}_{\operatorname{super part}}			\right)^{\prime \, (\xi + 1)}}(t,q)
	\times 
	\psi_{
		T_0
	}(q,t)
	\times 
	\frac{
		H(\operatorname{sh}(T_1 \sqcup T_2),q,t^{-1})
	}{
		H(\operatorname{sh}(T_1 \sqcup T_2)^\prime,t^{-1},q)
	}
	\times 
	\left(
	\frac{t^{-1} - 1}{q - 1}
	\right)^{
		|T_1| + |T_2|
	}
	\notag 
\end{align}
\normalsize
\end{proof}

\begin{prop}\mbox{}
\label{prp410-1703-25jul}
\footnotesize
\begin{align}
	&(
	\widetilde{\Psi}_{\lambda}^{(q,t^{-1})}
	\comp 
	\bigg|_{
		\substack{
			q_1 = q, \\
			q_2 = q^{-1}t,\\
			q_3 = t^{-1} \\
		}
	}
	)
	\bigg[
	\underbrace{	\prod_{1 \leq a < b \leq k}				}_{\text{ $a$ and $b$ are boxes in $(T_0 \sqcup T_1)^{\operatorname{add}}$}}
	\cals{C}^{(i_a,i_b)}\left(
	\frac{z_b}{z_a}
	; q, t
	\right)
	\bigg]
	\label{eqn323-2211-26jul}
	\\
	&= \prod_{\xi = N + M + 1}^{N + M + \ell_{\operatorname{sh}(T_2)} - 1}
	\psi_{\left(	(T_0 \sqcup T_1)^{\star\star}_{\operatorname{super part}}			\right)^{\prime \, (\xi)}/\left(	(T_0 \sqcup T_1)^{\star\star}_{\operatorname{super part}}			\right)^{\prime \, (\xi + 1)}}(t,q)
	\times 
	\frac{
		H(\operatorname{sh}(T_2),q,t^{-1})
	}{
		H(\operatorname{sh}(T_2)^\prime,t^{-1},q)
	}
	\times 
	\left(		\frac{t^{-1} - 1}{q - 1}			\right)^{|T_2|}
	\notag 
\end{align}
\normalsize
\end{prop}
\begin{proof}
We know from \textbf{Proposition \ref{prop48-1657-25Jul}} that $(T_0 \sqcup T_1)^{\operatorname{add}} \in \operatorname{RSSYTT}(N,M+\ell_{\operatorname{sh}(T_2)},0;\operatorname{sh}(T_2))$. Thus, by using \textbf{Lemma 4.5} of \cite{CSW2025}, we can show that 

\footnotesize
\begin{align}
	&(
	\widetilde{\Psi}_{\lambda}^{(q,t^{-1})}
	\comp 
	\bigg|_{
		\substack{
			q_1 = q, \\
			q_2 = q^{-1}t,\\
			q_3 = t^{-1} \\
		}
	}
	)
	\bigg[
	\underbrace{	\prod_{1 \leq a < b \leq k}				}_{\text{ $a$ and $b$ are boxes in $(T_0 \sqcup T_1)^{\operatorname{add}}$}}
	\cals{C}^{(i_a,i_b)}\left(
	\frac{z_b}{z_a}
	; q, t
	\right)
	\bigg]
	\label{eqn324-2210-26jul}
	\\
	&= 
	\psi_{ 	\left(	(T_0 \sqcup T_1)^{\operatorname{add}}_{\operatorname{super part}}			\right)^\prime				}(t,q)
	\times 
	\frac{
		H(\operatorname{sh}(T_2),q,t^{-1})
	}{
		H(\operatorname{sh}(T_2)^\prime,t^{-1},q)
	}
	\times 
	\left(		\frac{t^{-1} - 1}{q - 1}			\right)^{|T_2|}. 
	\notag 
\end{align}
\normalsize
Since $(T_0 \sqcup T_1)^{\operatorname{add}}$ contains only super numbers $N+M+1,\dots,N+M+\ell_{\operatorname{sh}(T_2)}$, we get that 

\begin{align}
	&\psi_{ 	\left(	(T_0 \sqcup T_1)^{\operatorname{add}}			\right)^\prime				}(t,q)
	= 
	\prod_{\xi = N + M + 1}^{N + M + \ell_{\operatorname{sh}(T_2)} - 1}
	\psi_{\left(	(T_0 \sqcup T_1)^{\operatorname{add}}			\right)^{\prime \, (\xi)}/\left(	(T_0 \sqcup T_1)^{\operatorname{add}}			\right)^{\prime \, (\xi + 1)}	}(t/q). 
\end{align}
From the construction of $(T_0 \sqcup T_1)^{\operatorname{add}}$, it is clear that for each 
$i \in \{N+M+1,\dots,N+M+\ell_{T_2}\}$, we have $\left((T_0 \sqcup T_1)^{\operatorname{add}}			\right)^{\prime \, (i)} = 
\left(	(T_0 \sqcup T_1)^{\star\star}_{\operatorname{super part}}			\right)^{\prime \, (i)}.$
Therefore, 
\begin{align}
	&\psi_{ 	\left(	(T_0 \sqcup T_1)^{\operatorname{add}}			\right)^\prime				}(t,q)
	= 
	\prod_{\xi = N + M + 1}^{N + M + \ell_{\operatorname{sh}(T_2)} - 1}
	\psi_{\left(	(T_0 \sqcup T_1)^{\star\star}_{\operatorname{super part}}			\right)^{\prime \, (\xi)}/\left(	(T_0 \sqcup T_1)^{\star\star}_{\operatorname{super part}}			\right)^{\prime \, (\xi + 1)}}(t,q). 
\label{eqn326-2210-26jul}
\end{align}
Substituting \eqref{eqn326-2210-26jul} into \eqref{eqn324-2210-26jul}, we immediately obtain \eqref{eqn323-2211-26jul}. 
\end{proof}

\begin{cor}\mbox{}
\label{cor311-1128-27jul}
\footnotesize
\begin{align}
	&
	(
	\widetilde{\Psi}_{\lambda}^{(q,t^{-1})}
	\comp 
	\bigg|_{
		\substack{
			q_1 = q, \\
			q_2 = q^{-1}t,\\
			q_3 = t^{-1} \\
		}
	}
	)
	\bigg[
	\underbrace{	\prod_{1 \leq a < b \leq k}				}_{\text{ $a$ and $b$ are boxes in $(T_0 \sqcup T_1)^{\star\star}$}}
	\cals{C}^{(i_a,i_b)}\left(
	\frac{z_b}{z_a}
	; q, t
	\right)
	\bigg]
	\\
	&\hspace{0.3cm}\times
	\left(
	(
	\widetilde{\Psi}_{\lambda}^{(q,t^{-1})}
	\comp 
	\bigg|_{
		\substack{
			q_1 = q, \\
			q_2 = q^{-1}t,\\
			q_3 = t^{-1} \\
		}
	}
	)
	\bigg[
	\underbrace{	\prod_{1 \leq a < b \leq k}				}_{\text{ $a$ and $b$ are boxes in $(T_0 \sqcup T_1)^{\operatorname{add}}$}}
	\cals{C}^{(i_a,i_b)}\left(
	\frac{z_b}{z_a}
	; q, t
	\right)
	\bigg]
	\right)^{-1}
	\notag 
	\\
	&=  
	\psi_{T^{\prime}_{1}}(t,q) \times \psi_{T_0}(q,t)
	\times
	\frac{
		H(\operatorname{sh}(T_1 \sqcup T_2),q,t^{-1})
	}{
		H(\operatorname{sh}(T_1 \sqcup T_2)^\prime,t^{-1},q)
	}
	\times
	\left(	\frac{t^{-1} - 1}{q-1}		\right)^{|T_1|}
	\notag 
	\times 
	\prod_{\square \in T_2}
	\frac{
		(t^{-1})^{\ell_{\operatorname{sh}(T_2)}(\square) + 1} - q^{a_{\operatorname{sh}(T_2)}(\square)}
	}{
		q^{a_{\operatorname{sh}(T_2)}(\square) + 1} - (t^{-1})^{\ell_{\operatorname{sh}(T_2)}(\square)}
	}
	\notag 
\end{align}
\normalsize
\end{cor}
\begin{proof}
From \textbf{Propositions \ref{prp49-1703-25jul}} and \textbf{\ref{prp410-1703-25jul}}, we can show that 
\footnotesize
\begin{align}
&
(
\widetilde{\Psi}_{\lambda}^{(q,t^{-1})}
\comp 
\bigg|_{
	\substack{
		q_1 = q, \\
		q_2 = q^{-1}t,\\
		q_3 = t^{-1} \\
	}
}
)
\bigg[
\underbrace{	\prod_{1 \leq a < b \leq k}				}_{\text{ $a$ and $b$ are boxes in $(T_0 \sqcup T_1)^{\star\star}$}}
\cals{C}^{(i_a,i_b)}\left(
\frac{z_b}{z_a}
; q, t
\right)
\bigg]
\label{eqn327-1113-27jul}
\\
&\hspace{0.3cm}\times
\left(
(
\widetilde{\Psi}_{\lambda}^{(q,t^{-1})}
\comp 
\bigg|_{
	\substack{
		q_1 = q, \\
		q_2 = q^{-1}t,\\
		q_3 = t^{-1} \\
	}
}
)
\bigg[
\underbrace{	\prod_{1 \leq a < b \leq k}				}_{\text{ $a$ and $b$ are boxes in $(T_0 \sqcup T_1)^{\operatorname{add}}$}}
\cals{C}^{(i_a,i_b)}\left(
\frac{z_b}{z_a}
; q, t
\right)
\bigg]
\right)^{-1}
\notag 
\\
&= 
\prod_{\xi = N + 1}^{N + M}
\psi_{\left(	(T_0 \sqcup T_1)^{\star\star}_{\operatorname{super part}}			\right)^{\prime \, (\xi)}/\left(	(T_0 \sqcup T_1)^{\star\star}_{\operatorname{super part}}			\right)^{\prime \, (\xi + 1)}}(t,q)
\times
\notag 
\\
&\hspace{0.3cm}\times
\psi_{
	T_0
}(q,t)
\times 
\frac{
	H(\operatorname{sh}(T_1 \sqcup T_2),q,t^{-1})
}{
	H(\operatorname{sh}(T_1 \sqcup T_2)^\prime,t^{-1},q)
}
\times 
\left(
\frac{t^{-1} - 1}{q - 1}
\right)^{
	|T_1| 
}
\times
\Big(
\frac{
	H(\operatorname{sh}(T_2),q,t^{-1})
}{
	H(\operatorname{sh}(T_2)^\prime,t^{-1},q)
}
\Big)^{-1}. 
\notag 
\end{align}
\normalsize
It is clear from the construction of $(T_0 \sqcup T_1)^{\star\star}$ that 
\begin{align}
\psi_{T_1^\prime}(t,q) = 
\prod_{\xi = N + 1}^{N + M}
\psi_{\left(	(T_0 \sqcup T_1)^{\star\star}_{\operatorname{super part}}			\right)^{\prime \, (\xi)}/\left(	(T_0 \sqcup T_1)^{\star\star}_{\operatorname{super part}}			\right)^{\prime \, (\xi + 1)}}(t,q). 
\label{eqn328-1112-27jul}
\end{align}
Moreover, we can show from the definition of $H(\mu,q,t)$ that 
\begin{align}
	\frac{
		H(\operatorname{sh}(T_2)^\prime,t^{-1},q)
	}{
		H(\operatorname{sh}(T_2),q,t^{-1})
	}
	=
	\prod_{\square \in T_2}
	\frac{
		t^{-(\ell_{T_2}(\square) + 1)} - q^{a_{T_2}(\square)}
	}{
		q^{a_{T_2}(\square) + 1} - t^{-\ell_{T_2}(\square)}
	} 
\label{eqn329-1112-27jul}
\end{align}
By using \eqref{eqn328-1112-27jul} and \eqref{eqn329-1112-27jul}, we are able to rewrite equation \eqref{eqn327-1113-27jul} as 
\footnotesize
\begin{align}
	&
	(
	\widetilde{\Psi}_{\lambda}^{(q,t^{-1})}
	\comp 
	\bigg|_{
		\substack{
			q_1 = q, \\
			q_2 = q^{-1}t,\\
			q_3 = t^{-1} \\
		}
	}
	)
	\bigg[
	\underbrace{	\prod_{1 \leq a < b \leq k}				}_{\text{ $a$ and $b$ are boxes in $(T_0 \sqcup T_1)^{\star\star}$}}
	\cals{C}^{(i_a,i_b)}\left(
	\frac{z_b}{z_a}
	; q, t
	\right)
	\bigg]
	\\
	&\hspace{0.3cm}\times
	\left(
	(
	\widetilde{\Psi}_{\lambda}^{(q,t^{-1})}
	\comp 
	\bigg|_{
		\substack{
			q_1 = q, \\
			q_2 = q^{-1}t,\\
			q_3 = t^{-1} \\
		}
	}
	)
	\bigg[
	\underbrace{	\prod_{1 \leq a < b \leq k}				}_{\text{ $a$ and $b$ are boxes in $(T_0 \sqcup T_1)^{\operatorname{add}}$}}
	\cals{C}^{(i_a,i_b)}\left(
	\frac{z_b}{z_a}
	; q, t
	\right)
	\bigg]
	\right)^{-1}
	\notag 
	\\
	&=  
	\psi_{T^{\prime}_{1}}(t,q) \times \psi_{T_0}(q,t)
	\times
	\frac{
		H(\operatorname{sh}(T_1 \sqcup T_2),q,t^{-1})
	}{
		H(\operatorname{sh}(T_1 \sqcup T_2)^\prime,t^{-1},q)
	}
	\times
	\left(	\frac{t^{-1} - 1}{q-1}		\right)^{|T_1|}
	\notag 
	\times 
	\prod_{\square \in T_2}
	\frac{
		(t^{-1})^{\ell_{\operatorname{sh}(T_2)}(\square) + 1} - q^{a_{\operatorname{sh}(T_2)}(\square)}
	}{
		q^{a_{\operatorname{sh}(T_2)}(\square) + 1} - (t^{-1})^{\ell_{\operatorname{sh}(T_2)}(\square)}
	}
	\notag 
\end{align}
\normalsize
\end{proof}

Using \textbf{Corollary \ref{cor311-1128-27jul}} and equation \eqref{eqn316-1129-27jul}, we immediately obtain the following corollary. 

\begin{cor}\mbox{}
\label{coro312-1211-27jul}
\footnotesize
\begin{align}
	&(
	\widetilde{\Psi}_{\lambda}^{(q,t^{-1})}
	\comp 
	\bigg|_{
		\substack{
			q_1 = q, \\
			q_2 = q^{-1}t,\\
			q_3 = t^{-1} \\
		}
	}
	)
	\bigg[
	\underbrace{	\prod_{1 \leq a < b \leq k}				}_{\text{ $a$ and $b$ are boxes in $(T_0 \sqcup T_1)$}}
	\cals{C}^{(i_a,i_b)}\left(
	\frac{z_b}{z_a}
	; q, t
	\right)
	\bigg]
	\\
	&= 
	\psi_{T^{\prime}_{1}}(t,q) \times \psi_{T_0}(q,t)
	\times
	\frac{
		H(\operatorname{sh}(T_1 \sqcup T_2),q,t^{-1})
	}{
		H(\operatorname{sh}(T_1 \sqcup T_2)^\prime,t^{-1},q)
	}
	\times
	\left(	\frac{t^{-1} - 1}{q-1}		\right)^{|T_1|}
	\notag 
	\times 
	\prod_{\square \in T_2}
	\frac{
		(t^{-1})^{\ell_{\operatorname{sh}(T_2)}(\square) + 1} - q^{a_{\operatorname{sh}(T_2)}(\square)}
	}{
		q^{a_{\operatorname{sh}(T_2)}(\square) + 1} - (t^{-1})^{\ell_{\operatorname{sh}(T_2)}(\square)}
	}
	\notag 
	\\
	&\hspace{0.3cm}\times
	\left(
	(
	\widetilde{\Psi}_{\lambda}^{(q,t^{-1})}
	\comp 
	\bigg|_{
		\substack{
			q_1 = q, \\
			q_2 = q^{-1}t,\\
			q_3 = t^{-1} \\
		}
	}
	)
	\bigg[
	\underbrace{	\prod_{1 \leq a < b \leq k}				}_{
		\substack{
			(1) \,\, \text{ $a$ is a box in $(T_0 \sqcup T_1)^{\operatorname{add}}$}
			\\
			(2) \,\, \text{ $b$ is a box in $T_0 \sqcup T_1$}
		}
	}
	\cals{C}^{(i_a,i_b)}\left(
	\frac{z_b}{z_a}
	; q, t
	\right)
	\bigg]
	\right)^{-1}
	\notag 
	\\
	&\hspace{0.3cm}\times
	\left(
	(
	\widetilde{\Psi}_{\lambda}^{(q,t^{-1})}
	\comp 
	\bigg|_{
		\substack{
			q_1 = q, \\
			q_2 = q^{-1}t,\\
			q_3 = t^{-1} \\
		}
	}
	)
	\bigg[
	\underbrace{	\prod_{1 \leq a < b \leq k}				}_{
		\substack{
			(1) \,\, \text{ $a$ is a box in $T_0 \sqcup T_1$}
			\\
			(2) \,\, \text{ $b$ is a box in $(T_0 \sqcup T_1)^{\operatorname{add}}$}
		}
	}
	\cals{C}^{(i_a,i_b)}\left(
	\frac{z_b}{z_a}
	; q, t
	\right)
	\bigg]
	\right)^{-1}.
	\notag
\end{align}
\normalsize
\end{cor}

\subsection{Analysis on the Quantities in  \eqref{eqn310-1137-27jul} and \eqref{eqn311-1137-27jul}}

\begin{prop}\mbox{}
\label{prp313-1201-27jul}
\footnotesize
\begin{align}
&(
\widetilde{\Psi}_{\lambda}^{(q,t^{-1})}
\comp 
\bigg|_{
	\substack{
		q_1 = q, \\
		q_2 = q^{-1}t,\\
		q_3 = t^{-1} \\
	}
}
)
\bigg[
\underbrace{	\prod_{1 \leq a < b \leq k}				}_{
	\substack{
		(1) \,\, \text{ $a$ is a box in $T_2$}
		\\
		(2) \,\, \text{ $b$ is a box in $T_0 \sqcup T_1$}
	}
}
\cals{C}^{(i_a,i_b)}\left(
\frac{z_b}{z_a}
; q, t
\right)
\bigg]
\\
&= 
(
\widetilde{\Psi}_{\lambda}^{(q,t^{-1})}
\comp 
\bigg|_{
	\substack{
		q_1 = q, \\
		q_2 = q^{-1}t,\\
		q_3 = t^{-1} \\
	}
}
)
\bigg[
\underbrace{	\prod_{1 \leq a < b \leq k}				}_{
	\substack{
		(1) \,\, \text{ $a$ is a box in $(T_0 \sqcup T_1)^{\operatorname{add}}$}
		\\
		(2) \,\, \text{ $b$ is a box in $T_0 \sqcup T_1$}
	}
}
\cals{C}^{(i_a,i_b)}\left(
\frac{z_b}{z_a}
; q, t
\right)
\bigg]
\notag 
\end{align}
\normalsize
\end{prop}
\begin{proof}
From the construction of $(T_0 \sqcup T_1)^{\operatorname{add}}$, we know that the numbers appearing in $(T_0 \sqcup T_1)^{\operatorname{add}}$ are all larger than the numbers appearing in $T_0 \sqcup T_1$. Also, since $T_2$ is the part of hyper number, we also get that all of the numbers in $T_2$ are larger than the numbers in $T_0 \sqcup T_1$. Furthermore, we know that $\operatorname{sh}((T_0 \sqcup T_1)^{\operatorname{add}}) = \operatorname{sh}(T_2)$. 

From this, it is clear that if $a \in (T_0 \sqcup T_1)^{\operatorname{add}}, a^\prime \in T_2$ are located in the same position, then for any $b \in T_0 \sqcup T_1$, we have 

\footnotesize
\begin{align}
	(
	\widetilde{\Psi}_{\lambda}^{(q,t^{-1})}
	\comp 
	\bigg|_{
		\substack{
			q_1 = q, \\
			q_2 = q^{-1}t,\\
			q_3 = t^{-1} \\
		}
	}
	)
	\Big[
	\cals{C}^{(i_a,i_b)}\left(
	\frac{z_b}{z_a}
	; q, t
	\right)
	\Big]
	=
	(
	\widetilde{\Psi}_{\lambda}^{(q,t^{-1})}
	\comp 
	\bigg|_{
		\substack{
			q_1 = q, \\
			q_2 = q^{-1}t,\\
			q_3 = t^{-1} \\
		}
	}
	)
	\Big[
	\cals{C}^{(i_{a^\prime},i_b)}\left(
	\frac{z_b}{z_{a^\prime}}
	; q, t
	\right)
	\Big].
\end{align}
\normalsize
From this we can conclude that 
\footnotesize
\begin{align}
	&(
	\widetilde{\Psi}_{\lambda}^{(q,t^{-1})}
	\comp 
	\bigg|_{
		\substack{
			q_1 = q, \\
			q_2 = q^{-1}t,\\
			q_3 = t^{-1} \\
		}
	}
	)
	\bigg[
	\underbrace{	\prod_{1 \leq a < b \leq k}				}_{
		\substack{
			(1) \,\, \text{ $a$ is a box in $T_2$}
			\\
			(2) \,\, \text{ $b$ is a box in $T_0 \sqcup T_1$}
		}
	}
	\cals{C}^{(i_a,i_b)}\left(
	\frac{z_b}{z_a}
	; q, t
	\right)
	\bigg]
	\\
	&= 
	(
	\widetilde{\Psi}_{\lambda}^{(q,t^{-1})}
	\comp 
	\bigg|_{
		\substack{
			q_1 = q, \\
			q_2 = q^{-1}t,\\
			q_3 = t^{-1} \\
		}
	}
	)
	\bigg[
	\underbrace{	\prod_{1 \leq a < b \leq k}				}_{
		\substack{
			(1) \,\, \text{ $a$ is a box in $(T_0 \sqcup T_1)^{\operatorname{add}}$}
			\\
			(2) \,\, \text{ $b$ is a box in $T_0 \sqcup T_1$}
		}
	}
	\cals{C}^{(i_a,i_b)}\left(
	\frac{z_b}{z_a}
	; q, t
	\right)
	\bigg]
	\notag 
\end{align}
\normalsize
\end{proof}

\begin{prop}\mbox{}
\label{prp314-1219-27jul}
	\footnotesize
	\begin{align}
		&(
		\widetilde{\Psi}_{\lambda}^{(q,t^{-1})}
		\comp 
		\bigg|_{
			\substack{
				q_1 = q, \\
				q_2 = q^{-1}t,\\
				q_3 = t^{-1} \\
			}
		}
		)
		\bigg[
		\underbrace{	\prod_{1 \leq a < b \leq k}				}_{
			\substack{
				(1) \,\, \text{ $a$ is a box in $T_0 \sqcup T_1$}
				\\
				(2) \,\, \text{ $b$ is a box in $T_2$}
			}
		}
		\cals{C}^{(i_a,i_b)}\left(
		\frac{z_b}{z_a}
		; q, t
		\right)
		\bigg]
		\\
		&=
		(
		\widetilde{\Psi}_{\lambda}^{(q,t^{-1})}
		\comp 
		\bigg|_{
			\substack{
				q_1 = q, \\
				q_2 = q^{-1}t,\\
				q_3 = t^{-1} \\
			}
		}
		)
		\bigg[
		\underbrace{	\prod_{1 \leq a < b \leq k}				}_{
			\substack{
				(1) \,\, \text{ $a$ is a box in $T_0 \sqcup T_1$}
				\\
				(2) \,\, \text{ $b$ is a box in $(T_0 \sqcup T_1)^{\operatorname{add}}$}
			}
		}
		\cals{C}^{(i_a,i_b)}\left(
		\frac{z_b}{z_a}
		; q, t
		\right)
		\bigg] 
		\notag 
	\end{align}
	\normalsize
\end{prop}
\begin{proof}
This proposition can be proved by using the same line of argument in \textbf{Proposition \ref{prp313-1201-27jul}}. 
\end{proof}

\subsection{Proof of Lemma \ref{lemm45-1458-25jul}}

According to equation \eqref{eqn36-1108-15jul}, \textbf{Lemma \ref{lemm36-1211-27jul}}, and \textbf{Corollary \ref{coro312-1211-27jul}}, we obtain that 

\footnotesize
\begin{align}
	&(
	\widetilde{\Psi}_{\lambda}^{(q,t^{-1})}
	\comp 
	\bigg|_{
		\substack{
			q_1 = q, \\
			q_2 = q^{-1}t,\\
			q_3 = t^{-1} \\
		}
	}
	)
	\bigg[
	\prod_{1 \leq a < b \leq k}
	\cals{C}^{(i_a,i_b)}\left(
	\frac{z_b}{z_a}
	; q, t
	\right)
	\bigg]
	\\
	&=  
	\cals{A}_2(T;q,t) \times 
	\psi_{T^{\prime}_{1}}(t,q) \times \psi_{T_0}(q,t)
	\times
	\frac{
		H(\operatorname{sh}(T_1 \sqcup T_2),q,t^{-1})
	}{
		H(\operatorname{sh}(T_1 \sqcup T_2)^\prime,t^{-1},q)
	}
	\times
	\left(	\frac{t^{-1} - 1}{q-1}		\right)^{|T_1|}
	\notag
	\times 
	\prod_{\square \in \operatorname{sh}(T_2)}
	\frac{
		(t^{-1})^{\ell_{\operatorname{sh}(T_2)}(\square) + 1} - q^{a_{\operatorname{sh}(T_2)}(\square)}
	}{
		q^{a_{\operatorname{sh}(T_2)}(\square) + 1} - (t^{-1})^{\ell_{\operatorname{sh}(T_2)}(\square)}
	}
	\notag 
	\\
	&\hspace{0.3cm}\times
	\left(
	(
	\widetilde{\Psi}_{\lambda}^{(q,t^{-1})}
	\comp 
	\bigg|_{
		\substack{
			q_1 = q, \\
			q_2 = q^{-1}t,\\
			q_3 = t^{-1} \\
		}
	}
	)
	\bigg[
	\underbrace{	\prod_{1 \leq a < b \leq k}				}_{
		\substack{
			(1) \,\, \text{ $a$ is a box in $(T_0 \sqcup T_1)^{\operatorname{add}}$}
			\\
			(2) \,\, \text{ $b$ is a box in $T_0 \sqcup T_1$}
		}
	}
	\cals{C}^{(i_a,i_b)}\left(
	\frac{z_b}{z_a}
	; q, t
	\right)
	\bigg]
	\right)^{-1}
	\notag 
	\\
	&\hspace{0.3cm}\times
	\left(
	(
	\widetilde{\Psi}_{\lambda}^{(q,t^{-1})}
	\comp 
	\bigg|_{
		\substack{
			q_1 = q, \\
			q_2 = q^{-1}t,\\
			q_3 = t^{-1} \\
		}
	}
	)
	\bigg[
	\underbrace{	\prod_{1 \leq a < b \leq k}				}_{
		\substack{
			(1) \,\, \text{ $a$ is a box in $T_0 \sqcup T_1$}
			\\
			(2) \,\, \text{ $b$ is a box in $(T_0 \sqcup T_1)^{\operatorname{add}}$}
		}
	}
	\cals{C}^{(i_a,i_b)}\left(
	\frac{z_b}{z_a}
	; q, t
	\right)
	\bigg]
	\right)^{-1}.
	\notag
	\\
	&\hspace{0.3cm}\times
	(
	\widetilde{\Psi}_{\lambda}^{(q,t^{-1})}
	\comp 
	\bigg|_{
		\substack{
			q_1 = q, \\
			q_2 = q^{-1}t,\\
			q_3 = t^{-1} \\
		}
	}
	)
	\bigg[
	\underbrace{	\prod_{1 \leq a < b \leq k}				}_{
		\substack{
			(1) \,\, \text{ $a$ is a box in $T_2$}
			\\
			(2) \,\, \text{ $b$ is a box in $T_0 \sqcup T_1$}
		}
	}
	\cals{C}^{(i_a,i_b)}\left(
	\frac{z_b}{z_a}
	; q, t
	\right)
	\bigg]
	\notag 
	\\
	&\hspace{0.3cm}\times
	(
	\widetilde{\Psi}_{\lambda}^{(q,t^{-1})}
	\comp 
	\bigg|_{
		\substack{
			q_1 = q, \\
			q_2 = q^{-1}t,\\
			q_3 = t^{-1} \\
		}
	}
	)
	\bigg[
	\underbrace{	\prod_{1 \leq a < b \leq k}				}_{
		\substack{
			(1) \,\, \text{ $a$ is a box in $T_0 \sqcup T_1$}
			\\
			(2) \,\, \text{ $b$ is a box in $T_2$}
		}
	}
	\cals{C}^{(i_a,i_b)}\left(
	\frac{z_b}{z_a}
	; q, t
	\right)
	\bigg]
	\notag 
\end{align}
\normalsize
Applying \textbf{Propositions \ref{prp313-1201-27jul}} and \textbf{\ref{prp314-1219-27jul}}, we get that 

\footnotesize
\begin{align}
	&(
	\widetilde{\Psi}_{\lambda}^{(q,t^{-1})}
	\comp 
	\bigg|_{
		\substack{
			q_1 = q, \\
			q_2 = q^{-1}t,\\
			q_3 = t^{-1} \\
		}
	}
	)
	\bigg[
	\prod_{1 \leq a < b \leq k}
	\cals{C}^{(i_a,i_b)}\left(
	\frac{z_b}{z_a}
	; q, t
	\right)
	\bigg]
	\\
	&=  
	\cals{A}_2(T;q,t) \times 
	\psi_{T^{\prime}_{1}}(t,q) \times \psi_{T_0}(q,t)
	\times
	\frac{
		H(\operatorname{sh}(T_1 \sqcup T_2),q,t^{-1})
	}{
		H(\operatorname{sh}(T_1 \sqcup T_2)^\prime,t^{-1},q)
	}
	\times
	\left(	\frac{t^{-1} - 1}{q-1}		\right)^{|T_1|}
	\notag
	\times 
	\prod_{\square \in \operatorname{sh}(T_2)}
	\frac{
		(t^{-1})^{\ell_{\operatorname{sh}(T_2)}(\square) + 1} - q^{a_{\operatorname{sh}(T_2)}(\square)}
	}{
		q^{a_{\operatorname{sh}(T_2)}(\square) + 1} - (t^{-1})^{\ell_{\operatorname{sh}(T_2)}(\square)}
	}. 
	\notag 
\end{align}
\normalsize
Thus, we have proved the \textbf{Lemma \ref{lemm45-1458-25jul}}. 

\section{Partially Symmetricity of the Quantum Corner Polynomial}
\label{sec5-0016-7aug}

In this section, we will prove that the quantum corner polynomial, as defined in \textbf{Definition \ref{defn314-1234-25jul}}, is a partially symmetric polynomial. This is the another main result of this paper. 
We begin by recalling the definition of a partially symmetric polynomial.

\begin{dfn}[Definition 6.10 of \cite{CSw2024}]
Let $g(x_1,\dots,x_n) := \sum_{(i_1,\dots,i_n)}c_{i_1,\dots,i_n}x_1^{i_1}\cdots x_n^{i_n}$ be a formal power series and let $I_1,\dots,I_\ell$ be a collection of disjoint subsets of $\{1,\dots,n\}$ such that 
\begin{align}
I_1 \cup \cdots \cup I_\ell = \{1,\dots,n\}
\end{align}
We say that $g(x_1,\dots,x_n)$ is a \textbf{partially symmetric polynomial} with respect to the index sets $I_1,\dots,I_\ell$ if for each $i \in \{1,\dots,\ell\}$, the formal power series $g(x_1,\dots,x_n)$ is symmetric with respect to the variables $\{x_j\}_{j \in I_i}$. That is, for any permutation $\sigma$ of $\{1,\dots,n\}$ such that for each $i \in \{1,\dots,\ell\}$ $\sigma(I_i) = I_i$, we have 
\begin{align}
g(x_1,\dots,x_n) = g(x_{\sigma(1)},\dots,x_{\sigma(n)}) 
\end{align}
\end{dfn}

We now state the main theorem of this section. 

\begin{thm}
\label{thm42-main-1526}
The quantum corner polynomials
\begin{align}
\quantumcorner_\lambda(x_1,\dots,x_N;x_{N+1},\dots,x_{N+M};x_{N+M+1},\dots,x_{N+M+L};
q,t)
\end{align}
are partially symmetric with respect to the index sets $I_1 = \{1,\dots,N\}$, $I_2 = \{N + 1,\dots, N + M\}$, $I_3 = \{N+M+1,\dots,N+M+L\}$. 
\end{thm}

To prove this theorem, we will need to use a key fact related to star product $\star$. We will therefore review the necessary facts about the star product below. 

\begin{dfn}[\cite{FHHSY-comm} \cite{FHSSY-ker}]
For each $j \in \bb{Z}^{\geq 0}$, let $\cals{R}^j$ be the set of symmetric rational functions of $j$ variables with coefficients in the field $\bb{Q}(q,t)$. We define the \textbf{star product} as a map $\star : \cals{R}^m \otimes \cals{R}^n \rightarrow \cals{R}^{m+n}$ that sends each $f \in \cals{R}^m$ and $g \in \cals{R}^n$ to 

\footnotesize
\begin{align}
&(f \star g)(x_1,\dots,x_{m+n})
\\
&:= 
\underset{x_1,\dots,x_{m+n}}{\operatorname{Sym}}
\Bigg[
f(x_1,\dots,x_m)
g(x_{m+1},\dots,x_{m+n})
\prod_{
\substack{
1 \leq \alpha \leq m
\\
m+1 \leq \beta \leq m+n
}
}
\frac{
	\left(1 - t	\frac{z_\beta}{z_\alpha}		\right)
	\left(1 - q^{-1}		\frac{z_\beta}{z_\alpha}					\right)
	\left(1 - qt^{-1}			\frac{z_\beta}{z_\alpha}				\right)
}{
	\left(1 - 	\frac{z_\beta}{z_\alpha}				\right)^3
}
\Bigg]
\in 
\cals{R}^{m+n}.
\notag 
\end{align}
\normalsize 
Here $\underset{x_1,\dots,x_{j}}{\operatorname{Sym}}$ means 
\begin{align}
\underset{x_1,\dots,x_{j}}{\operatorname{Sym}}\Bigg[
h(x_1,\dots,x_j)
\Bigg]
:= 
\frac{1}{j!}\sum_{\sigma \in S_j}
h(x_{\sigma(1)},\dots,x_{\sigma(j)}),
\end{align}
where $S_j$ is the set of all bijection from $\{1,\dots,j\}$ to $\{1,\dots,j\}$. 
\end{dfn}

\begin{dfn}
For each $c \in \{1,2,3\}$ and 
$n \in \bb{Z}^{\geq 0}$, we define 
\begin{align}
	\epsilon^{(c)}_n(z_1,\dots,z_n) := 
	\begin{cases}
		\displaystyle 
		\prod_{
			1 \leq i < j \leq n 
		}
		\frac{
			\left(1 - qt^{-1}\frac{z_j}{z_i}\right)
			\left(1 - q^{-1}t\frac{z_j}{z_i}\right)
		}{
			\left(1 - \frac{z_j}{z_i}\right)^2
		}
		&\text{ if } c = 2
		\\
		\displaystyle 
		\prod_{
			1 \leq i < j \leq n 
		}
		\frac{
			\left(1 - q\frac{z_j}{z_i}\right)
			\left(1 - q^{-1}\frac{z_j}{z_i}\right)
		}{
			\left(1 - \frac{z_j}{z_i}\right)^2
		}
		&\text{ if } c = 1
		\\
		\displaystyle 
		\prod_{
			1 \leq i < j \leq n 
		}
		\frac{
			\left(1 - t\frac{z_j}{z_i}\right)
			\left(1 - t^{-1}\frac{z_j}{z_i}\right)
		}{
			\left(1 - \frac{z_j}{z_i}\right)^2
		}
		&\text{ if } c = 3 
	\end{cases}
\end{align}
\end{dfn}

\begin{rem}
It is clear that for each $c \in \{1,2,3\}$, $\epsilon^{(c)}_n \in \cals{R}^n$. 
\end{rem}

\begin{prop}
\label{prop46-1604-5aug}
For each $c \in \{1,2,3\}$ and for each $n, m \in \bb{Z}^{\geq 0}$, we have $\epsilon^{(c)}_{n} \star \epsilon^{(c)}_{m} = \epsilon^{(c)}_{m} \star \epsilon^{(c)}_{n}$. 
\end{prop}
\begin{proof}
See Theorem 1.5 of \cite{FHHSY-comm}. 
\end{proof}

Now we are ready to prove \textbf{Theorem \ref{thm42-main-1526}}. 

\begin{proof}[Proof of \textbf{Theorem \ref{thm42-main-1526}}]
First, for convenience, we introduce a shorthand notation: for each $i \in \{1,\dots,N+M+L\}$, let 
\begin{align}
	y_i := 
	\frac{
		q_{c_i}^{\frac{1}{2}} - q_{c_i}^{-\frac{1}{2}}
	}{
		q_{3}^{\frac{1}{2}} - q_{3}^{-\frac{1}{2}}
	}. 
\end{align}
Note that 

\footnotesize
\begin{align}
	&
	\bigg|_{
		\substack{
			q_1 = q, \\
			q_2 = q^{-1}t,\\
			q_3 = t^{-1} \\
		}
	}
	\left[
	\prod_{1 \leq i < j \leq k}f^{\vec{c}}_{11}\left(\frac{z_j}{z_i} \right)
	\times
	\langle 0 |\widetilde{T}^{\vec{c},\vec{u}}_{1}(z_1 )\cdots \widetilde{T}^{\vec{c},\vec{u}}_{1}(z_k )|0\rangle
	\right]
	\label{eqn46-1544-5aug}
	\\
	&= 
	\sum_{i_1 = 1}^{N+M+L}
	\cdots
	\sum_{i_k = 1}^{N+M+L}
	\bigg[
	\bigg|_{
		\substack{
			q_1 = q, \\
			q_2 = q^{-1}t,\\
			q_3 = t^{-1} \\
		}
	}
	y_{i_1}\cdots y_{i_k}
	u_{i_1}\cdots u_{i_k}
	\times
	\prod_{1 \leq a < b \leq k}
	\bigg|_{
		\substack{
			q_1 = q, \\
			q_2 = q^{-1}t,\\
			q_3 = t^{-1} \\
		}
	}
	\cals{D}^{(i_a,i_b)}\left(
	\frac{z_b}{z_a}
	; q, t
	\right)
	\bigg]
	\notag 
	\\
	&=
	\prod_{1 \leq a < b \leq k}
	\frac{
		\left(1 - \frac{z_b}{z_a}\right)^2
	}{
		\left(1 - t\frac{z_b}{z_a}\right)
		\left(1 - t^{-1}\frac{z_b}{z_a}\right)
	}
	\times 
	\sum_{i_1 = 1}^{N+M+L}
	\cdots
	\sum_{i_k = 1}^{N+M+L}
	\bigg[
	\bigg|_{
		\substack{
			q_1 = q, \\
			q_2 = q^{-1}t,\\
			q_3 = t^{-1} \\
		}
	}
	y_{i_1}\cdots y_{i_k}
	u_{i_1}\cdots u_{i_k}
	\times
	\prod_{1 \leq a < b \leq k}
	\gamma^{(i_a,i_b)}(\frac{z_b}{z_a},q,t)
	\bigg]
	\notag 
\end{align}
\normalsize
where 
\begin{align}
	\gamma^{(i,j)}(z,q,t)
	:= 
	\begin{cases}
		\displaystyle 
		\frac{
			\left(1 - tz\right)
			\left(1 - q^{-1}		z				\right)
			\left(1 - qt^{-1}			z			\right)
		}{
			\left(1 - 	z			\right)^3
		}
		\hspace{0.3cm}
		&\text{ if } i < j
		\\
		\displaystyle 
		\frac{
			\left(1 - qt^{-1}		z		\right)
			\left(1 - q^{-1}t		z		\right)
		}{
			\left(1 - z\right)^2
		}
		\hspace{0.3cm}
		&\text{ if } i = j = \text{ hyper number }
		\\
		\displaystyle 
		\frac{
			\left(1 - q^{-1}			z			\right)
			\left(1 - q			z	\right)
		}{
			\left(1 - z\right)^2
		}
		\hspace{0.3cm}
		&\text{ if } i = j = \text{ super number }
		\\
		\displaystyle 
		\frac{
			\left(1 - tz\right)
			\left(1 - t^{-1}z\right)
		}{
			\left(1 - z\right)^2
		}
		\hspace{0.3cm}
		&\text{ if } i = j = \text{ ordinary number }
		\\
		\displaystyle
		\frac{
			\left(1 - t^{-1}z\right)
			\left(1 - q		z		\right)
			\left(1 - q^{-1}t		z	\right)
		}{
			\left(1 - z\right)^3
		}
		\hspace{0.3cm}
		&\text{ if } i > j 
	\end{cases}
\end{align}

Define
\begin{align}
	\cals{J}(a_1,\dots,a_n) = 
	\left\{
	(I_1,\dots,I_n)
	\;\middle\vert\;
	\begin{array}{@{}l@{}}
		(1) \,\, I_1 \sqcup \cdots \sqcup I_n = 
		\{1,\dots,m\}
		\\
		(2) \,\, |I_k| = a_k \,\, ^\forall k = 1,\dots,n
	\end{array}
	\right\}. 
\end{align}
Then, one can show that 
\vspace{0.2cm}

\footnotesize
\begin{align}
	&\sum_{i_1 = 1}^{N+M+L}
	\cdots
	\sum_{i_k = 1}^{N+M+L}
	\bigg[
	\bigg|_{
		\substack{
			q_1 = q, \\
			q_2 = q^{-1}t,\\
			q_3 = t^{-1} \\
		}
	}
	y_{i_1}\cdots y_{i_k}
	u_{i_1}\cdots u_{i_k}
	\times
	\prod_{1 \leq a < b \leq k}
	\gamma^{(i_a,i_b)}(\frac{z_b}{z_a},q,t)
	\bigg]
	\label{eqn48-1544-5aug}
	\\
	&= 
	\underbrace{				
		\sum_{
			(a_1,\dots,a_{N+M+L}) \in 
			\left(
			\bb{Z}^{\geq 0}
			\right)^{N+M+L}
		}
	}_{
		a_1 + \dots + a_{N+M+L} = k 
	}
	\biggl\{
	\bigg|_{
		\substack{
			q_1 = q, \\
			q_2 = q^{-1}t,\\
			q_3 = t^{-1} \\
		}
	}
	\bigg(
	\frac{q_1^{\frac{1}{2}} - q_1^{-\frac{1}{2}}}{
		q_3^{\frac{1}{2}} - q_3^{-\frac{1}{2}}
	}
	\bigg)^{a_{N+1} + \cdots + a_{N+M}}
	\bigg(
	\frac{q_2^{\frac{1}{2}} - q_2^{-\frac{1}{2}}}{
		q_3^{\frac{1}{2}} - q_3^{-\frac{1}{2}}
	}
	\bigg)^{a_{N+M+1} + \cdots + a_{N+M+L}}
	\times
	u_1^{a_1}\cdots u_{N+M+L}^{a_{N+M+L}}
	\times 
	\notag 
	\\
	&\times
	\sum_{
		(I_1,\dots,I_{N+M+L}) \in 
		\cals{J}(a_1,\dots,a_{N+M+L})
	}
	\Big[
	\prod_{k = 1}^{N}\epsilon^{(3)}_{a_k}(z_{I_k} ; q,t)
	\times
	\prod_{k = N + 1}^{N + M}\epsilon^{(1)}_{a_k}(z_{I_k} ; q,t)
	\times
	\prod_{k = N+M+1}^{N+M+L}\epsilon^{(2)}_{a_k}(z_{I_k} ; q,t)
	\times 
	\notag 
	\\
	&\times 
	\prod_{1 \leq i < j \leq N+M+L}
	\prod_{\alpha \in I_i, \beta \in I_j}
	\frac{
		\left(1 - t	\frac{z_\beta}{z_\alpha}		\right)
		\left(1 - q^{-1}		\frac{z_\beta}{z_\alpha}					\right)
		\left(1 - qt^{-1}			\frac{z_\beta}{z_\alpha}				\right)
	}{
		\left(1 - 	\frac{z_\beta}{z_\alpha}				\right)^3
	}
	\Big]
	\biggr\}
	\notag 
	\\
	&= 
	\underbrace{				
		\sum_{
			(a_1,\dots,a_{N+M+L}) \in 
			\left(
			\bb{Z}^{\geq 0}
			\right)^{N+M+L}
		}
	}_{
		a_1 + \dots + a_{N+M+L} = k 
	}
	\biggl\{
	\bigg|_{
		\substack{
			q_1 = q, \\
			q_2 = q^{-1}t,\\
			q_3 = t^{-1} \\
		}
	}
	\bigg(
	\frac{q_1^{\frac{1}{2}} - q_1^{-\frac{1}{2}}}{
		q_3^{\frac{1}{2}} - q_3^{-\frac{1}{2}}
	}
	\bigg)^{a_{N+1} + \cdots + a_{N+M}}
	\bigg(
	\frac{q_2^{\frac{1}{2}} - q_2^{-\frac{1}{2}}}{
		q_3^{\frac{1}{2}} - q_3^{-\frac{1}{2}}
	}
	\bigg)^{a_{N+M+1} + \cdots + a_{N+M+L}}
	\times
	u_1^{a_1}\cdots u_{N+M+L}^{a_{N+M+L}}
	\times 
	\notag 
	\\
	&\times
	\frac{
		\left(	a_1 + \cdots + a_{N+M+L}	 \right) !
	}{
		a_1! \cdots a_{N+M+L} !
	} 
	\times 
	\Bigg[
	\epsilon^{(3)}_{a_1}(z_{I_1} ; q,t) \star \cdots \star \epsilon^{(3)}_{a_N}(z_{I_N} ; q,t)
	\star \epsilon^{(1)}_{a_{N+1}}(z_{I_{N+1}} ; q,t) \star \cdots \star \epsilon^{(1)}_{a_{N+M}}(z_{I_{N+M}} ; q,t)
	\notag 
	\\
	&\star \epsilon^{(2)}_{a_{N+M+1}}(z_{I_{N+M+1}} ; q,t) \star \cdots \star \epsilon^{(2)}_{a_{N+M+L}}(z_{I_{N+M+L}} ; q,t)
	\Bigg]
	\biggr\}
	\notag 
\end{align}
\normalsize
From equations \eqref{eqn46-1544-5aug} and \eqref{eqn48-1544-5aug}, 
we obtain that for $(a_1,\dots,a_{N+M+L}) \in 
\left(
\bb{Z}^{\geq 0}
\right)^{N+M+L}$ such that $a_1 + \dots + a_{N+M+L} = k$, the coefficient in front of $u_1^{a_1}\cdots u_{N+M+L}^{a_{N+M+L}}$ in
 
\footnotesize
\begin{align}
	\lim_{\xi \rightarrow t^{-1}}\,\,
	(
	\dualmap
	\comp 
	\bigg|_{
		\substack{
			q_1 = q, \\
			q_2 = q^{-1}t,\\
			q_3 = t^{-1} \\
		}
	}
	)
	\left(
	\cals{N}_{\lambda}(z_1,\dots,z_k )
	\times
	\prod_{1 \leq i < j \leq k}f^{\vec{c}}_{11}\left(\frac{z_j}{z_i} \right)
	\times
	\langle 0 |\widetilde{T}^{\vec{c},\vec{u}}_{1}(z_1 )\cdots \widetilde{T}^{\vec{c},\vec{u}}_{1}(z_k )|0\rangle
	\right)
\end{align}
\normalsize
is equal to 

\footnotesize
\begin{align}
	&
	\bigg|_{
		\substack{
			q_1 = q, \\
			q_2 = q^{-1}t,\\
			q_3 = t^{-1} \\
		}
	}
	\bigg(
	\frac{q_1^{\frac{1}{2}} - q_1^{-\frac{1}{2}}}{
		q_3^{\frac{1}{2}} - q_3^{-\frac{1}{2}}
	}
	\bigg)^{a_{N+1} + \cdots + a_{N+M}}
	\bigg(
	\frac{q_2^{\frac{1}{2}} - q_2^{-\frac{1}{2}}}{
		q_3^{\frac{1}{2}} - q_3^{-\frac{1}{2}}
	}
	\bigg)^{a_{N+M+1} + \cdots + a_{N+M+L}}
	\times 
	\frac{
		\left(	a_1 + \cdots + a_{N+M+L}	 \right) !
	}{
		a_1! \cdots a_{N+M+L} !
	} 
	\times 
	\label{eqn414-1208-5aug}
	\\
	&\times
	(\lim_{\xi \rightarrow t^{-1}} \comp \dualmap)
	\biggl\{
	\bigg|_{
		\substack{
			q_1 = q, \\
			q_2 = q^{-1}t,\\
			q_3 = t^{-1} \\
		}
	}
	\cals{N}_{\lambda}(z_1,\dots,z_k )
	\times 
	\prod_{1 \leq a < b \leq k}
	\frac{
		\left(1 - \frac{z_b}{z_a}\right)^2
	}{
		\left(1 - t\frac{z_b}{z_a}\right)
		\left(1 - t^{-1}\frac{z_b}{z_a}\right)
	}
	\times 
	\notag 
	\\
	&\times\Bigg[
	\epsilon^{(3)}_{a_1}(z_{I_1} ; q,t) \star \cdots \star \epsilon^{(3)}_{a_N}(z_{I_N} ; q,t)
	\star \epsilon^{(1)}_{a_{N+1}}(z_{I_{N+1}} ; q,t) \star \cdots \star \epsilon^{(1)}_{a_{N+M}}(z_{I_{N+M}} ; q,t)
	\notag 
	\\
	&\star \epsilon^{(2)}_{a_{N+M+1}}(z_{I_{N+M+1}} ; q,t) \star \cdots \star \epsilon^{(2)}_{a_{N+M+L}}(z_{I_{N+M+L}} ; q,t)
	\Bigg]
	\biggr\}
	\notag 
\end{align}
\normalsize

By utilizing \textbf{Theorem \ref{thm43-1458-25jul}} and Equation \eqref{eqn414-1208-5aug}, one can see that to prove \textbf{Theorem \ref{thm42-main-1526}}, it is sufficient to show that for each $c \in \{1,2,3\}$ and for each non-negative integers $n, m \in \bb{Z}^{\geq 0}$, $\epsilon^{(c)}_{n} \star \epsilon^{(c)}_{m} = \epsilon^{(c)}_{m} \star \epsilon^{(c)}_{n}$. 

However, this is precisely the statement of \textbf{Proposition \ref{prop46-1604-5aug}}. Therefore, we have proved \textbf{Theorem \ref{thm42-main-1526}}. 
\end{proof}

\appendix

\section{Proof of Lemma \ref{prp317-2237-16aug}}
\label{app-A-2246-16aug}

\begin{prop}\mbox{}
\label{prpa1-17aug-sun-1329}
\footnotesize
\begin{align}
	&\prod_{\zeta = 1}^{\ell(T_2)}
	\prod_{d = N+M+1}^{N+M+L}
	\prod_{\alpha = N+M+1}^{d - 1} 
	\prod_{\omega = 1}^{\theta_{\zeta,\alpha}}
	\frac{
		\left(	1 - q^{	\sum_{\gamma = \alpha + 1}^{N + M + L} \theta_{\zeta,\gamma}	+ \omega - \sum_{\gamma = d+1}^{N + M + L}\theta_{\zeta,\gamma} 		}			\right)
		\left(	1 - tq^{	\sum_{\gamma = \alpha + 1}^{N + M + L} \theta_{\zeta,\gamma}	+ \omega - (\sum_{\gamma = d}^{N + M + L}\theta_{\zeta,\gamma} 	+ 1 )	}		\right)
	}{
		\left(	1 - q^{	\sum_{\gamma = \alpha + 1}^{N + M + L} \theta_{\zeta,\gamma}	+ \omega - \sum_{\gamma = d}^{N + M + L}\theta_{\zeta,\gamma} 		}			\right)
		\left(	1 - 	tq^{	\sum_{\gamma = \alpha + 1}^{N + M + L} \theta_{\zeta,\gamma}	+ \omega - (\sum_{\gamma = d + 1}^{N + M + L}\theta_{\zeta,\gamma} 	+ 1 )	}			\right)
	}
	\notag 
	\\
	&=
	\prod_{\zeta = 1}^{\ell(T_2)}
	\left[
	\frac{
		f_{q,t}(q^{T^{(N+M+1)}_{\zeta} 		})
	}{
		f_{q,t}(1	)
	}
	\right]
	\times
	\prod_{d = N+M+1}^{N+M+L}
	\prod_{\zeta = 1}^{\ell(T_2)}
	\left[
	\frac{
		f_{q,t}(1)
	}{
		f_{q,t}(q^{T^{(d)}_{\zeta} - T^{(d+1)}_{\zeta}})
	}
	\right]
\end{align}
\normalsize
\end{prop}
\begin{proof}
First note that for any $\beta \in \{1,\dots,N+M+L\}$, we get that $\sum_{\gamma = \beta}^{N+M+L}\theta_{\zeta,\gamma} = T^{(\beta)}_{\zeta}$. Thus, 
\footnotesize
\begin{align}
	&\prod_{\zeta = 1}^{\ell(T_2)}
	\prod_{d = N+M+1}^{N+M+L}
	\prod_{\alpha = N+M+1}^{d - 1} 
	\prod_{\omega = 1}^{\theta_{\zeta,\alpha}}
	\frac{
		\left(	1 - q^{	\sum_{\gamma = \alpha + 1}^{N + M + L} \theta_{\zeta,\gamma}	+ \omega - \sum_{\gamma = d+1}^{N + M + L}\theta_{\zeta,\gamma} 		}			\right)
		\left(	1 - tq^{	\sum_{\gamma = \alpha + 1}^{N + M + L} \theta_{\zeta,\gamma}	+ \omega - (\sum_{\gamma = d}^{N + M + L}\theta_{\zeta,\gamma} 	+ 1 )	}		\right)
	}{
		\left(	1 - q^{	\sum_{\gamma = \alpha + 1}^{N + M + L} \theta_{\zeta,\gamma}	+ \omega - \sum_{\gamma = d}^{N + M + L}\theta_{\zeta,\gamma} 		}			\right)
		\left(	1 - 	tq^{	\sum_{\gamma = \alpha + 1}^{N + M + L} \theta_{\zeta,\gamma}	+ \omega - (\sum_{\gamma = d + 1}^{N + M + L}\theta_{\zeta,\gamma} 	+ 1 )	}			\right)
	}
	\notag 
	\\
	&=
	\prod_{d = N+M+1}^{N+M+L}
	\prod_{\zeta = 1}^{\ell(T_2)}
	\prod_{\alpha = N+M+1}^{d - 1} 
	\prod_{\omega = 1}^{\theta_{\zeta,\alpha}}
	\frac{
		\left(	1 - q^{		T^{(\alpha + 1)}_{\zeta}		+ \omega - T^{(d+1)}_{\zeta}	}			\right)
		\left(	1 - tq^{	T^{(\alpha + 1)}_{\zeta}		+ \omega - (	T^{(d)}_{\zeta}	+ 1 )	}		\right)
	}{
		\left(	1 - q^{	T^{(\alpha + 1)}_{\zeta}		+ \omega - T^{(d)}_{\zeta}		}			\right)
		\left(	1 - 	tq^{	T^{(\alpha + 1)}_{\zeta}		+ \omega - (	T^{(d+1)}_{\zeta}	+ 1 )	}			\right)
	}
\end{align}
\normalsize
One can easily show that 
\footnotesize
\begin{align}
\prod_{\alpha = N+M+1}^{d - 1} 
\prod_{\omega = 1}^{\theta_{\zeta,\alpha}}
\left(	1 - q^{		T^{(\alpha + 1)}_{\zeta}		+ \omega - T^{(d+1)}_{\zeta}	}			\right)
&= 
\frac{
	\left(	qq^{T^{(d)}_{\zeta} - T^{(d+1)}_{\zeta}}	; q	\right)_\infty
}{
	\left(	qq^{T^{(N+M+1)}_{\zeta} - T^{(d+1)}_{\zeta}}	; q	\right)_\infty
}
\\
\prod_{\alpha = N+M+1}^{d - 1} 
\prod_{\omega = 1}^{\theta_{\zeta,\alpha}}
\frac{
	1
}{
	\left(	1 - 	tq^{	T^{(\alpha + 1)}_{\zeta}		+ \omega - (	T^{(d+1)}_{\zeta}	+ 1 )	}			\right)
}
&=
\frac{
	\left(	tq^{T^{(N+M+1)}_{\zeta} - T^{(d+1)}_{\zeta}}	; q	\right)_\infty
}{
	\left(	tq^{T^{(d)}_{\zeta} - T^{(d+1)}_{\zeta}}	; q	\right)_\infty
}
\\
\prod_{\alpha = N+M+1}^{d - 1} 
\prod_{\omega = 1}^{\theta_{\zeta,\alpha}}
\left(	1 - tq^{	T^{(\alpha + 1)}_{\zeta}		+ \omega - (	T^{(d)}_{\zeta}	+ 1 )	}		\right)
&= 
\frac{
	(t;q)_{\infty}
}{
	(tq^{T^{(N+M+1)}_{\zeta} - T^{(d)}_{\zeta}} ; q)_{\infty}
}
\\
\prod_{\alpha = N+M+1}^{d - 1} 
\prod_{\omega = 1}^{\theta_{\zeta,\alpha}}
\frac{
	1
}{
	\left(	1 - q^{	T^{(\alpha + 1)}_{\zeta}		+ \omega - T^{(d)}_{\zeta}		}			\right)
}
&= 
\frac{
	\left(	qq^{T^{(N+M+1)}_{\zeta} - T^{(d)}_{\zeta}}	; q	\right)_\infty
}{		
	\left(	q	; q	\right)_\infty
}
\end{align}
\normalsize
Thus, we obtain that 
\footnotesize
\begin{align}
	&\prod_{\zeta = 1}^{\ell(T_2)}
	\prod_{d = N+M+1}^{N+M+L}
	\prod_{\alpha = N+M+1}^{d - 1} 
	\prod_{\omega = 1}^{\theta_{\zeta,\alpha}}
	\frac{
		\left(	1 - q^{	\sum_{\gamma = \alpha + 1}^{N + M + L} \theta_{\zeta,\gamma}	+ \omega - \sum_{\gamma = d+1}^{N + M + L}\theta_{\zeta,\gamma} 		}			\right)
		\left(	1 - tq^{	\sum_{\gamma = \alpha + 1}^{N + M + L} \theta_{\zeta,\gamma}	+ \omega - (\sum_{\gamma = d}^{N + M + L}\theta_{\zeta,\gamma} 	+ 1 )	}		\right)
	}{
		\left(	1 - q^{	\sum_{\gamma = \alpha + 1}^{N + M + L} \theta_{\zeta,\gamma}	+ \omega - \sum_{\gamma = d}^{N + M + L}\theta_{\zeta,\gamma} 		}			\right)
		\left(	1 - 	tq^{	\sum_{\gamma = \alpha + 1}^{N + M + L} \theta_{\zeta,\gamma}	+ \omega - (\sum_{\gamma = d + 1}^{N + M + L}\theta_{\zeta,\gamma} 	+ 1 )	}			\right)
	}
	\notag 
	\\
	&= 
	\prod_{d = N+M+1}^{N+M+L}
	\prod_{\zeta = 1}^{\ell(T_2)}
	\Bigg[
	\frac{
		\left(	qq^{T^{(d)}_{\zeta} - T^{(d+1)}_{\zeta}}	; q	\right)_\infty
	}{
		\left(	qq^{T^{(N+M+1)}_{\zeta} - T^{(d+1)}_{\zeta}}	; q	\right)_\infty
	} 
	\times 
	\frac{
		\left(	tq^{T^{(N+M+1)}_{\zeta} - T^{(d+1)}_{\zeta}}	; q	\right)_\infty
	}{
		\left(	tq^{T^{(d)}_{\zeta} - T^{(d+1)}_{\zeta}}	; q	\right)_\infty
	} 
	\notag 
	\\
	&\hspace{6.8cm}\times 
	\frac{
		(t;q)_{\infty}
	}{
		(tq^{T^{(N+M+1)}_{\zeta} - T^{(d)}_{\zeta}} ; q)_{\infty}
	} 
	\times
	\frac{
		\left(	qq^{T^{(N+M+1)}_{\zeta} - T^{(d)}_{\zeta}}	; q	\right)_\infty
	}{		
		\left(	q	; q	\right)_\infty
	}
	\Bigg]
	\notag 
	\\
	&= 
	\prod_{d = N+M+1}^{N+M+L}
	\prod_{\zeta = 1}^{\ell(T_2)}
	\left[
	\frac{
		f_{q,t}(q^{T^{(N+M+1)}_{\zeta} - T^{(d+1)}_{\zeta}})f_{q,t}(1)
	}{
		f_{q,t}(q^{T^{(d)}_{\zeta} - T^{(d+1)}_{\zeta}})f_{q,t}(q^{T^{(N+M+1)}_{\zeta} - T^{(d)}_{\zeta}})
	}
	\right]
	\notag 
	\\
	&=
	\prod_{\zeta = 1}^{\ell(T_2)}
	\left[
	\frac{
		f_{q,t}(q^{T^{(N+M+1)}_{\zeta} 		})
	}{
		f_{q,t}(1	)
	}
	\right]
	\times
	\prod_{d = N+M+1}^{N+M+L}
	\prod_{\zeta = 1}^{\ell(T_2)}
	\left[
	\frac{
		f_{q,t}(1)
	}{
		f_{q,t}(q^{T^{(d)}_{\zeta} - T^{(d+1)}_{\zeta}})
	}
	\right]
\end{align}
\normalsize
\end{proof}

\begin{prop}
\label{prpa2-17aug-sun-1329}
\footnotesize
\begin{align}
	&\prod_{\zeta = 1}^{\ell(T_2)}
	\prod_{d = N+M+1}^{N+M+L} 
	\prod_{\tau = \zeta + 1}^{\ell(T_2)} 
	\prod_{\omega = 1}^{\theta_{\tau,d}}
	\frac{
		\left(		1 - t^{-1} (t^{-1})^{\tau - \zeta}	q^{	\sum_{\gamma = d+1}^{N + M + L}\theta_{\tau,\gamma}	+ \omega - 	\sum_{\gamma = d+ 1}^{N + M + L} \theta_{\zeta,\gamma}				}			\right)
		\left(
		1 - (t^{-1})^{\tau - \zeta}	q^{	\sum_{\gamma = d+1}^{N + M + L}\theta_{\tau,\gamma}	+ \omega - 	\sum_{\gamma = d}^{N + M + L} \theta_{\zeta,\gamma}				}		
		\right)
	}{
		\left(		1 - t^{-1} (t^{-1})^{\tau - \zeta}	q^{	\sum_{\gamma = d+1}^{N + M + L}\theta_{\tau,\gamma}	+ \omega - 	\sum_{\gamma = d}^{N + M + L} \theta_{\zeta,\gamma}				}			\right)
		\left(
		1 - (t^{-1})^{\tau - \zeta}	q^{	\sum_{\gamma = d+1}^{N + M + L}\theta_{\tau,\gamma}	+ \omega - 	\sum_{\gamma = d+1}^{N + M + L} \theta_{\zeta,\gamma}				}		
		\right)
	}
	\notag 
	\\
	&=
	\prod_{d = N+M+1}^{N+M+L} 
	\prod_{\zeta = 1}^{\ell(T_2)}
	\prod_{\tau = \zeta + 1}^{\ell(T_2)} 
	\prod_{\alpha = 1}^{\theta_{\zeta,d}}
	\frac{
		\left(		1 - t t^{\tau - \zeta} q^{ T^{(d+1)}_{\zeta} + \alpha	 -	(T^{(d)}_{\tau}	+ 1)		}					\right)
		\left(
		1 - t^{\tau - \zeta} q^{ T^{(d+1)}_{\zeta} + \alpha	-(	T^{(d+1)}_{\tau}	+ 1)		}	
		\right)
	}{
		\left(
		1 - t^{\tau - \zeta} q^{ T^{(d+1)}_{\zeta} + \alpha	 -	(T^{(d)}_{\tau} + 1)		}	
		\right)
		\left(		1 - t t^{\tau - \zeta} q^{	T^{(d+1)}_{\zeta} + \alpha	 - (T^{(d+1)}_{\tau}	+ 1)		}					\right)
	}
\end{align}
\normalsize
\end{prop}
\begin{proof}
First, one can show that 
\footnotesize
\begin{align}
	&\prod_{\omega = 1}^{\theta_{\tau,d}}
	\frac{
		\left(		1 - t^{-1} (t^{-1})^{\tau - \zeta}	q^{	\sum_{\gamma = d+1}^{N + M + L}\theta_{\tau,\gamma}	+ \omega - 	\sum_{\gamma = d+ 1}^{N + M + L} \theta_{\zeta,\gamma}				}			\right)
		\left(
		1 - (t^{-1})^{\tau - \zeta}	q^{	\sum_{\gamma = d+1}^{N + M + L}\theta_{\tau,\gamma}	+ \omega - 	\sum_{\gamma = d}^{N + M + L} \theta_{\zeta,\gamma}				}		
		\right)
	}{
		\left(		1 - t^{-1} (t^{-1})^{\tau - \zeta}	q^{	\sum_{\gamma = d+1}^{N + M + L}\theta_{\tau,\gamma}	+ \omega - 	\sum_{\gamma = d}^{N + M + L} \theta_{\zeta,\gamma}				}			\right)
		\left(
		1 - (t^{-1})^{\tau - \zeta}	q^{	\sum_{\gamma = d+1}^{N + M + L}\theta_{\tau,\gamma}	+ \omega - 	\sum_{\gamma = d+1}^{N + M + L} \theta_{\zeta,\gamma}				}		
		\right)
	}
	\\
	&= 
	\prod_{\alpha = 1}^{\theta_{\zeta,d}}
	\frac{
		\left(		1 - t^{-1} (t^{-1})^{\tau - \zeta} q^{	\sum_{\gamma = d}^{N+M+L} \theta_{\tau,\gamma}	+ 1	- (\sum_{\gamma = d + 1}^{N+M+L} \theta_{\zeta,\gamma} + \alpha)		}					\right)
		\left(
		1 - (t^{-1})^{\tau - \zeta} q^{	\sum_{\gamma = d + 1}^{N+M+L} \theta_{\tau,\gamma}	+ 1	- (\sum_{\gamma = d + 1}^{N+M+L} \theta_{\zeta,\gamma} + \alpha)		}	
		\right)
	}{
		\left(		1 - t^{-1} (t^{-1})^{\tau - \zeta} q^{	\sum_{\gamma = d + 1}^{N+M+L} \theta_{\tau,\gamma}	+ 1	- (\sum_{\gamma = d + 1}^{N+M+L} \theta_{\zeta,\gamma} + \alpha)		}					\right)
		\left(
		1 - (t^{-1})^{\tau - \zeta} q^{	\sum_{\gamma = d}^{N+M+L} \theta_{\tau,\gamma}	+ 1	- (\sum_{\gamma = d + 1}^{N+M+L} \theta_{\zeta,\gamma} + \alpha)		}	
		\right)
	}
	\notag 
\end{align}
\normalsize
Thus,
\footnotesize
\begin{align}
	&\prod_{\zeta = 1}^{\ell(T_2)}
	\prod_{d = N+M+1}^{N+M+L} 
	\prod_{\tau = \zeta + 1}^{\ell(T_2)} 
	\prod_{\omega = 1}^{\theta_{\tau,d}}
	\frac{
		\left(		1 - t^{-1} (t^{-1})^{\tau - \zeta}	q^{	\sum_{\gamma = d+1}^{N + M + L}\theta_{\tau,\gamma}	+ \omega - 	\sum_{\gamma = d+ 1}^{N + M + L} \theta_{\zeta,\gamma}				}			\right)
		\left(
		1 - (t^{-1})^{\tau - \zeta}	q^{	\sum_{\gamma = d+1}^{N + M + L}\theta_{\tau,\gamma}	+ \omega - 	\sum_{\gamma = d}^{N + M + L} \theta_{\zeta,\gamma}				}		
		\right)
	}{
		\left(		1 - t^{-1} (t^{-1})^{\tau - \zeta}	q^{	\sum_{\gamma = d+1}^{N + M + L}\theta_{\tau,\gamma}	+ \omega - 	\sum_{\gamma = d}^{N + M + L} \theta_{\zeta,\gamma}				}			\right)
		\left(
		1 - (t^{-1})^{\tau - \zeta}	q^{	\sum_{\gamma = d+1}^{N + M + L}\theta_{\tau,\gamma}	+ \omega - 	\sum_{\gamma = d+1}^{N + M + L} \theta_{\zeta,\gamma}				}		
		\right)
	}
	\notag 
	\\
	&=
	\prod_{\zeta = 1}^{\ell(T_2)}
	\prod_{d = N+M+1}^{N+M+L} 
	\prod_{\tau = \zeta + 1}^{\ell(T_2)} 
	\prod_{\alpha = 1}^{\theta_{\zeta,d}}
	\frac{
		\left(		1 - t t^{\tau - \zeta} q^{ \sum_{\gamma = d + 1}^{N+M+L} \theta_{\zeta,\gamma} + \alpha	 -	(\sum_{\gamma = d}^{N+M+L} \theta_{\tau,\gamma}	+ 1)		}					\right)
		\left(
		1 - t^{\tau - \zeta} q^{ \sum_{\gamma = d + 1}^{N+M+L} \theta_{\zeta,\gamma} + \alpha	-(	\sum_{\gamma = d + 1}^{N+M+L} \theta_{\tau,\gamma}	+ 1)		}	
		\right)
	}{
		\left(		1 - t t^{\tau - \zeta} q^{	\sum_{\gamma = d + 1}^{N+M+L} \theta_{\zeta,\gamma} + \alpha	 - (\sum_{\gamma = d + 1}^{N+M+L} \theta_{\tau,\gamma}	+ 1)		}					\right)
		\left(
		1 - t^{\tau - \zeta} q^{ \sum_{\gamma = d + 1}^{N+M+L} \theta_{\zeta,\gamma} + \alpha	 -	(\sum_{\gamma = d}^{N+M+L} \theta_{\tau,\gamma}	+ 1)		}	
		\right)
	}
	\notag 
	\\
	&=
	\prod_{d = N+M+1}^{N+M+L} 
	\prod_{\zeta = 1}^{\ell(T_2)}
	\prod_{\tau = \zeta + 1}^{\ell(T_2)} 
	\prod_{\alpha = 1}^{\theta_{\zeta,d}}
	\frac{
		\left(		1 - t t^{\tau - \zeta} q^{ T^{(d+1)}_{\zeta} + \alpha	 -	(T^{(d)}_{\tau}	+ 1)		}					\right)
		\left(
		1 - t^{\tau - \zeta} q^{ T^{(d+1)}_{\zeta} + \alpha	-(	T^{(d+1)}_{\tau}	+ 1)		}	
		\right)
	}{
		\left(		1 - t t^{\tau - \zeta} q^{	T^{(d+1)}_{\zeta} + \alpha	 - (T^{(d+1)}_{\tau}	+ 1)		}					\right)
		\left(
		1 - t^{\tau - \zeta} q^{ T^{(d+1)}_{\zeta} + \alpha	 -	(T^{(d)}_{\tau} + 1)		}	
		\right)
	}
\end{align}
\normalsize
Note that in the last equality we use the fact that for any $\beta \in \{1,\dots,N+M+L\}$, we get that $\sum_{\gamma = \beta}^{N+M+L}\theta_{\zeta,\gamma} = T^{(\beta)}_{\zeta}$. 
\end{proof}

\begin{prop}\mbox{}
\label{prpa3-17aug-sun-1329}
\footnotesize
\begin{align}
	&
	\prod_{\zeta = 1}^{\ell(T_2)}
	\prod_{d = N+M+1}^{N+M+L}
	\prod_{\tau = \zeta + 1}^{\ell(T_2)}
	\prod_{p = d+1}^{N+M+L}
	\prod_{\Xi = 1}^{\theta_{\zeta,d}}
	\prod_{\omega = 1}^{\theta_{\tau,p}}
	\Biggl\{
	\\
	&\frac{
		\left(	1 - q^{-1}					
		(t^{-1})^{\tau - \zeta} q^{	\sum_{\gamma = p + 1}^{N + M + L} \theta_{\tau,\gamma}	+ \omega	- (\sum_{\gamma = d + 1}^{N + M + L} \theta_{\zeta,\gamma} + \Xi)		}	
		\right)
		\left(	1 - qt^{-1}			
		(t^{-1})^{\tau - \zeta} q^{	\sum_{\gamma = p + 1}^{N + M + L} \theta_{\tau,\gamma}	+ \omega	- (\sum_{\gamma = d + 1}^{N + M + L} \theta_{\zeta,\gamma} + \Xi)		}
		\right)
	}{
		\left(	1 - q			
		(t^{-1})^{\tau - \zeta} q^{	\sum_{\gamma = p + 1}^{N + M + L} \theta_{\tau,\gamma}	+ \omega	- (\sum_{\gamma = d + 1}^{N + M + L} \theta_{\zeta,\gamma} + \Xi)		}
		\right)
		\left(	1 - q^{-1}t			
		(t^{-1})^{\tau - \zeta} q^{	\sum_{\gamma = p + 1}^{N + M + L} \theta_{\tau,\gamma}	+ \omega	- (\sum_{\gamma = d + 1}^{N + M + L} \theta_{\zeta,\gamma} + \Xi)		}
		\right)
	}
	\notag 
	\\
	&\times
	\frac{
		\left(	1 - t		
		(t^{-1})^{\tau - \zeta} q^{	\sum_{\gamma = p + 1}^{N + M + L} \theta_{\tau,\gamma}	+ \omega	- (\sum_{\gamma = d + 1}^{N + M + L} \theta_{\zeta,\gamma} + \Xi)		}
		\right)
	}{
		\left(	1 - t^{-1}			
		(t^{-1})^{\tau - \zeta} q^{	\sum_{\gamma = p + 1}^{N + M + L} \theta_{\tau,\gamma}	+ \omega	- (\sum_{\gamma = d + 1}^{N + M + L} \theta_{\zeta,\gamma} + \Xi)		}
		\right)
	}
	\Biggr\}
	\notag 
	\\
	&=
	\prod_{\zeta = 1}^{\ell(T_2)}
	\Biggl\{
	\frac{f_{q,t}\left(		t^{\ell(T_2) - \zeta} q^{ T^{(N+M+1)}_{\zeta}		}					\right)}{
		f_{q,t}\left(		t^{\ell(T_2) - \zeta} 					\right)
	}
	\times 
	\frac{
		f_{q,t}\left(	1					\right)
	}{
		f_{q,t}\left(q^{ T^{(N+M+1)}_{\zeta} 		}						\right)
	} 
	\Biggr\}
	\times
	\notag 
	\\
	&\hspace{0.3cm}\times 
	\prod_{\zeta = 1}^{\ell(T_2)}
	\prod_{d = N+M+1}^{N+M+L}
	\prod_{\tau = \zeta + 1}^{\ell(T_2)}
	\Biggl\{
	\frac{
		f_{q,t}\left(	t^{\tau - \zeta} q^{ T^{(d+1)}_{\zeta} 	- T^{(d+1)}_{\tau}	 		}						\right)
	}{
		f_{q,t}\left(		t^{-1}		
		t^{\tau - \zeta} q^{ T^{(d+1)}_{\zeta} 	- T^{(d+1)}_{\tau}			}					\right)
	}
	\times
	\frac{
		f_{q,t}\left(		t^{-1}		
		t^{\tau - \zeta} q^{ T^{(d)}_{\zeta} 	- T^{(d+1)}_{\tau} 		}					\right)
	}{
		f_{q,t}\left(	t^{\tau - \zeta} q^{ T^{(d)}_{\zeta} 	- T^{(d+1)}_{\tau}		}						\right)
	}
	\Biggr\} 
	\notag 
\end{align}
\normalsize
\end{prop}
\begin{proof}
Note that we can write 
\footnotesize
\begin{align}
	&\frac{
		\left(	1 - q^{-1}					
		(t^{-1})^{\tau - \zeta} q^{	\sum_{\gamma = p + 1}^{N + M + L} \theta_{\tau,\gamma}	+ \omega	- (\sum_{\gamma = d + 1}^{N + M + L} \theta_{\zeta,\gamma} + \Xi)		}	
		\right)
		\left(	1 - qt^{-1}			
		(t^{-1})^{\tau - \zeta} q^{	\sum_{\gamma = p + 1}^{N + M + L} \theta_{\tau,\gamma}	+ \omega	- (\sum_{\gamma = d + 1}^{N + M + L} \theta_{\zeta,\gamma} + \Xi)		}
		\right)
	}{
		\left(	1 - q			
		(t^{-1})^{\tau - \zeta} q^{	\sum_{\gamma = p + 1}^{N + M + L} \theta_{\tau,\gamma}	+ \omega	- (\sum_{\gamma = d + 1}^{N + M + L} \theta_{\zeta,\gamma} + \Xi)		}
		\right)
		\left(	1 - q^{-1}t			
		(t^{-1})^{\tau - \zeta} q^{	\sum_{\gamma = p + 1}^{N + M + L} \theta_{\tau,\gamma}	+ \omega	- (\sum_{\gamma = d + 1}^{N + M + L} \theta_{\zeta,\gamma} + \Xi)		}
		\right)
	}
	\\
	&\times
	\frac{
		\left(	1 - t		
		(t^{-1})^{\tau - \zeta} q^{	\sum_{\gamma = p + 1}^{N + M + L} \theta_{\tau,\gamma}	+ \omega	- (\sum_{\gamma = d + 1}^{N + M + L} \theta_{\zeta,\gamma} + \Xi)		}
		\right)
	}{
		\left(	1 - t^{-1}			
		(t^{-1})^{\tau - \zeta} q^{	\sum_{\gamma = p + 1}^{N + M + L} \theta_{\tau,\gamma}	+ \omega	- (\sum_{\gamma = d + 1}^{N + M + L} \theta_{\zeta,\gamma} + \Xi)		}
		\right)
	}
	\notag 
	\\
	&= 
	\frac{
		\left(	1 - 					
		(t^{-1})^{\tau - \zeta} q^{	\sum_{\gamma = p + 1}^{N + M + L} \theta_{\tau,\gamma}	+ \omega - 1	- (\sum_{\gamma = d + 1}^{N + M + L} \theta_{\zeta,\gamma} + \Xi)		}	
		\right)
		\left(	1 - 
		(t^{-1})^{\tau - \zeta} q^{	\sum_{\gamma = p + 1}^{N + M + L} \theta_{\tau,\gamma}	+ \omega	- (\sum_{\gamma = d + 1}^{N + M + L} \theta_{\zeta,\gamma} + \Xi)		}
		\right)
	}{
		\left(	1 - 			
		(t^{-1})^{\tau - \zeta} q^{	\sum_{\gamma = p + 1}^{N + M + L} \theta_{\tau,\gamma}	+ \omega	- (\sum_{\gamma = d + 1}^{N + M + L} \theta_{\zeta,\gamma} + \Xi)		}
		\right)
		\left(	1 - 			
		(t^{-1})^{\tau - \zeta} q^{	\sum_{\gamma = p + 1}^{N + M + L} \theta_{\tau,\gamma}	+ \omega + 1	- (\sum_{\gamma = d + 1}^{N + M + L} \theta_{\zeta,\gamma} + \Xi)		}
		\right)
	}
	\notag 
	\\
	&\times
	\frac{
		\left(	1 - t^{-1}			
		(t^{-1})^{\tau - \zeta} q^{	\sum_{\gamma = p + 1}^{N + M + L} \theta_{\tau,\gamma}	+ \omega + 1	- (\sum_{\gamma = d + 1}^{N + M + L} \theta_{\zeta,\gamma} + \Xi)		}
		\right)
		\left(	1 - t		
		(t^{-1})^{\tau - \zeta} q^{	\sum_{\gamma = p + 1}^{N + M + L} \theta_{\tau,\gamma}	+ \omega	- (\sum_{\gamma = d + 1}^{N + M + L} \theta_{\zeta,\gamma} + \Xi)		}
		\right)
	}{
		\left(	1 - t^{-1}			
		(t^{-1})^{\tau - \zeta} q^{	\sum_{\gamma = p + 1}^{N + M + L} \theta_{\tau,\gamma}	+ \omega	- (\sum_{\gamma = d + 1}^{N + M + L} \theta_{\zeta,\gamma} + \Xi)		}
		\right)
		\left(	1 - t			
		(t^{-1})^{\tau - \zeta} q^{	\sum_{\gamma = p + 1}^{N + M + L} \theta_{\tau,\gamma}	+ \omega - 1	- (\sum_{\gamma = d + 1}^{N + M + L} \theta_{\zeta,\gamma} + \Xi)		}
		\right)
	}
	\notag 
\end{align}
\normalsize
Thus,
\footnotesize
\begin{align}
	&
	\prod_{\omega = 1}^{\theta_{\tau,p}}
	\Biggl\{
	\frac{
		\left(	1 - q^{-1}					
		(t^{-1})^{\tau - \zeta} q^{	\sum_{\gamma = p + 1}^{N + M + L} \theta_{\tau,\gamma}	+ \omega	- (\sum_{\gamma = d + 1}^{N + M + L} \theta_{\zeta,\gamma} + \Xi)		}	
		\right)
		\left(	1 - qt^{-1}			
		(t^{-1})^{\tau - \zeta} q^{	\sum_{\gamma = p + 1}^{N + M + L} \theta_{\tau,\gamma}	+ \omega	- (\sum_{\gamma = d + 1}^{N + M + L} \theta_{\zeta,\gamma} + \Xi)		}
		\right)
	}{
		\left(	1 - q			
		(t^{-1})^{\tau - \zeta} q^{	\sum_{\gamma = p + 1}^{N + M + L} \theta_{\tau,\gamma}	+ \omega	- (\sum_{\gamma = d + 1}^{N + M + L} \theta_{\zeta,\gamma} + \Xi)		}
		\right)
		\left(	1 - q^{-1}t			
		(t^{-1})^{\tau - \zeta} q^{	\sum_{\gamma = p + 1}^{N + M + L} \theta_{\tau,\gamma}	+ \omega	- (\sum_{\gamma = d + 1}^{N + M + L} \theta_{\zeta,\gamma} + \Xi)		}
		\right)
	}
	\\
	&\times
	\frac{
		\left(	1 - t		
		(t^{-1})^{\tau - \zeta} q^{	\sum_{\gamma = p + 1}^{N + M + L} \theta_{\tau,\gamma}	+ \omega	- (\sum_{\gamma = d + 1}^{N + M + L} \theta_{\zeta,\gamma} + \Xi)		}
		\right)
	}{
		\left(	1 - t^{-1}			
		(t^{-1})^{\tau - \zeta} q^{	\sum_{\gamma = p + 1}^{N + M + L} \theta_{\tau,\gamma}	+ \omega	- (\sum_{\gamma = d + 1}^{N + M + L} \theta_{\zeta,\gamma} + \Xi)		}
		\right)
	}
	\Biggr\}
	\notag 
	\\
	&= 
	\frac{
		\left(	1 - 					
		t^{\tau - \zeta} q^{ T^{(d+1)}_{\zeta} + \Xi		- T^{(p+1)}_{\tau}			}	
		\right)
		\left(	1 - 
		t^{\tau - \zeta} q^{ T^{(d+1)}_{\zeta} + \Xi	-(T^{(p+1)}_{\tau}	+ 1)		}
		\right)
		\left(	1 - t
		t^{\tau - \zeta} q^{ T^{(d+1)}_{\zeta} + \Xi	-(T^{(p)}_{\tau}	 + 1)		}
		\right)
		\left(	1 - t^{-1}		
		t^{\tau - \zeta} q^{ T^{(d+1)}_{\zeta} + \Xi		- T^{(p)}_{\tau}			}
		\right)
	}{
		\left(	1 - t			
		t^{\tau - \zeta} q^{ T^{(d+1)}_{\zeta} + \Xi		-(T^{(p+1)}_{\tau}	+ 1)	}
		\right)
		\left(	1 - t^{-1}			
		t^{\tau - \zeta} q^{ T^{(d+1)}_{\zeta} + \Xi		- T^{(p+1)}_{\tau}	 	}
		\right)
		\left(	1 - 			
		t^{\tau - \zeta} q^{ T^{(d+1)}_{\zeta} + \Xi	- T^{(p)}_{\tau}			}
		\right)
		\left(	1 - 			
		t^{\tau - \zeta} q^{ T^{(d+1)}_{\zeta} + \Xi		-(T^{(p)}_{\tau}	 + 1)		}
		\right)
	}
	\notag 
\end{align}
\normalsize
Therefore,
\footnotesize 
\begin{align}
	&
	\prod_{\Xi = 1}^{\theta_{\zeta,d}}
	\prod_{\omega = 1}^{\theta_{\tau,p}}
	\Biggl\{
	\frac{
		\left(	1 - q^{-1}					
		(t^{-1})^{\tau - \zeta} q^{	\sum_{\gamma = p + 1}^{N + M + L} \theta_{\tau,\gamma}	+ \omega	- (\sum_{\gamma = d + 1}^{N + M + L} \theta_{\zeta,\gamma} + \Xi)		}	
		\right)
		\left(	1 - qt^{-1}			
		(t^{-1})^{\tau - \zeta} q^{	\sum_{\gamma = p + 1}^{N + M + L} \theta_{\tau,\gamma}	+ \omega	- (\sum_{\gamma = d + 1}^{N + M + L} \theta_{\zeta,\gamma} + \Xi)		}
		\right)
	}{
		\left(	1 - q			
		(t^{-1})^{\tau - \zeta} q^{	\sum_{\gamma = p + 1}^{N + M + L} \theta_{\tau,\gamma}	+ \omega	- (\sum_{\gamma = d + 1}^{N + M + L} \theta_{\zeta,\gamma} + \Xi)		}
		\right)
		\left(	1 - q^{-1}t			
		(t^{-1})^{\tau - \zeta} q^{	\sum_{\gamma = p + 1}^{N + M + L} \theta_{\tau,\gamma}	+ \omega	- (\sum_{\gamma = d + 1}^{N + M + L} \theta_{\zeta,\gamma} + \Xi)		}
		\right)
	}
	\\
	&\times
	\frac{
		\left(	1 - t		
		(t^{-1})^{\tau - \zeta} q^{	\sum_{\gamma = p + 1}^{N + M + L} \theta_{\tau,\gamma}	+ \omega	- (\sum_{\gamma = d + 1}^{N + M + L} \theta_{\zeta,\gamma} + \Xi)		}
		\right)
	}{
		\left(	1 - t^{-1}			
		(t^{-1})^{\tau - \zeta} q^{	\sum_{\gamma = p + 1}^{N + M + L} \theta_{\tau,\gamma}	+ \omega	- (\sum_{\gamma = d + 1}^{N + M + L} \theta_{\zeta,\gamma} + \Xi)		}
		\right)
	}
	\Biggr\}
	\notag 
	\\
	&= 
	\prod_{\Xi = 0}^{\theta_{\zeta,d} - 1}
	\frac{
		\left(	1 - 					
		t^{\tau - \zeta} q^{ T^{(d+1)}_{\zeta} + \Xi + 1		- T^{(p+1)}_{\tau}			}	
		\right)
	}{
		\left(	1 - t			
		t^{\tau - \zeta} q^{ T^{(d+1)}_{\zeta} + \Xi 	- T^{(p+1)}_{\tau}		}
		\right)
	}
	\times
	\prod_{\Xi = 0}^{\theta_{\zeta,d} - 1}
	\frac{
		\left(	1 - 
		t^{\tau - \zeta} q^{ T^{(d+1)}_{\zeta} + \Xi 	- T^{(p+1)}_{\tau}			}
		\right)
	}{
		\left(	1 - t^{-1}			
		t^{\tau - \zeta} q^{ T^{(d+1)}_{\zeta} + \Xi + 1		- T^{(p+1)}_{\tau}	 	}
		\right)
	}
	\notag 
	\\
	&\hspace{0.3cm}\times
	\prod_{\Xi = 0}^{\theta_{\zeta,d} - 1}
	\frac{
		\left(	1 - t
		t^{\tau - \zeta} q^{ T^{(d+1)}_{\zeta} + \Xi 	- T^{(p)}_{\tau}	 		}
		\right)
	}{
		\left(	1 - 			
		t^{\tau - \zeta} q^{ T^{(d+1)}_{\zeta} + \Xi + 1	- T^{(p)}_{\tau}			}
		\right)
	} 
	\times
	\prod_{\Xi = 0}^{\theta_{\zeta,d} - 1}
	\frac{
		\left(	1 - t^{-1}		
		t^{\tau - \zeta} q^{ T^{(d+1)}_{\zeta} + \Xi + 1		- T^{(p)}_{\tau}			}
		\right)
	}{
		\left(	1 - 			
		t^{\tau - \zeta} q^{ T^{(d+1)}_{\zeta} + \Xi 		- T^{(p)}_{\tau}			}
		\right)
	}
	\notag 
	\\
	&= \frac{f_{q,t}\left(		t^{\tau - \zeta} q^{ T^{(d)}_{\zeta}	- T^{(p+1)}_{\tau}		}					\right)}{
		f_{q,t}\left(	t^{\tau - \zeta} q^{ T^{(d)}_{\zeta} 	- T^{(p)}_{\tau}		}						\right)
	}
	\times 
	\frac{
		f_{q,t}\left(	t^{-1}			
		t^{\tau - \zeta} q^{ T^{(d+1)}_{\zeta} 	- T^{(p+1)}_{\tau}	 	}						\right)
	}{
		f_{q,t}\left(		t^{-1}		
		t^{\tau - \zeta} q^{ T^{(d+1)}_{\zeta} 	- T^{(p)}_{\tau}			}					\right)
	}
	\times 
	\frac{
		f_{q,t}\left(	t^{\tau - \zeta} q^{ T^{(d+1)}_{\zeta} 	- T^{(p)}_{\tau}	 		}						\right)
	}{
		f_{q,t}\left(		t^{\tau - \zeta} q^{ T^{(d+1)}_{\zeta}	- T^{(p+1)}_{\tau}		}					\right)
	}
	\times
	\frac{
		f_{q,t}\left(		t^{-1}		
		t^{\tau - \zeta} q^{ T^{(d)}_{\zeta} 	- T^{(p)}_{\tau} 		}					\right)
	}{
		f_{q,t}\left(	t^{-1}			
		t^{\tau - \zeta} q^{ T^{(d)}_{\zeta} 	- T^{(p+1)}_{\tau}		}						\right)
	}.
	\notag 
\end{align}
\normalsize
Consequently, 
\footnotesize
\begin{align}
	&\prod_{p = d+1}^{N+M+L}
	\prod_{\Xi = 1}^{\theta_{\zeta,d}}
	\prod_{\omega = 1}^{\theta_{\tau,p}}
	\Biggl\{
	\\
	&\frac{
		\left(	1 - q^{-1}					
		(t^{-1})^{\tau - \zeta} q^{	\sum_{\gamma = p + 1}^{N + M + L} \theta_{\tau,\gamma}	+ \omega	- (\sum_{\gamma = d + 1}^{N + M + L} \theta_{\zeta,\gamma} + \Xi)		}	
		\right)
		\left(	1 - qt^{-1}			
		(t^{-1})^{\tau - \zeta} q^{	\sum_{\gamma = p + 1}^{N + M + L} \theta_{\tau,\gamma}	+ \omega	- (\sum_{\gamma = d + 1}^{N + M + L} \theta_{\zeta,\gamma} + \Xi)		}
		\right)
	}{
		\left(	1 - q			
		(t^{-1})^{\tau - \zeta} q^{	\sum_{\gamma = p + 1}^{N + M + L} \theta_{\tau,\gamma}	+ \omega	- (\sum_{\gamma = d + 1}^{N + M + L} \theta_{\zeta,\gamma} + \Xi)		}
		\right)
		\left(	1 - q^{-1}t			
		(t^{-1})^{\tau - \zeta} q^{	\sum_{\gamma = p + 1}^{N + M + L} \theta_{\tau,\gamma}	+ \omega	- (\sum_{\gamma = d + 1}^{N + M + L} \theta_{\zeta,\gamma} + \Xi)		}
		\right)
	}
	\notag 
	\\
	&\times
	\frac{
		\left(	1 - t		
		(t^{-1})^{\tau - \zeta} q^{	\sum_{\gamma = p + 1}^{N + M + L} \theta_{\tau,\gamma}	+ \omega	- (\sum_{\gamma = d + 1}^{N + M + L} \theta_{\zeta,\gamma} + \Xi)		}
		\right)
	}{
		\left(	1 - t^{-1}			
		(t^{-1})^{\tau - \zeta} q^{	\sum_{\gamma = p + 1}^{N + M + L} \theta_{\tau,\gamma}	+ \omega	- (\sum_{\gamma = d + 1}^{N + M + L} \theta_{\zeta,\gamma} + \Xi)		}
		\right)
	}
	\Biggr\}
	\notag 
	\\
	&= 
	\frac{f_{q,t}\left(		t^{\tau - \zeta} q^{ T^{(d)}_{\zeta}		}					\right)}{
		f_{q,t}\left(	t^{-1}			
		t^{\tau - \zeta} q^{ T^{(d)}_{\zeta} 		}						\right)
	}
	\times 
	\frac{
		f_{q,t}\left(	t^{-1}			
		t^{\tau - \zeta} q^{ T^{(d+1)}_{\zeta} 		}						\right)
	}{
		f_{q,t}\left(		t^{\tau - \zeta} q^{ T^{(d+1)}_{\zeta}		}					\right)
	}
	\times 
	\frac{
		f_{q,t}\left(	t^{\tau - \zeta} q^{ T^{(d+1)}_{\zeta} 	- T^{(d+1)}_{\tau}	 		}						\right)
	}{
		f_{q,t}\left(		t^{-1}		
		t^{\tau - \zeta} q^{ T^{(d+1)}_{\zeta} 	- T^{(d+1)}_{\tau}			}					\right)
	}
	\times
	\frac{
		f_{q,t}\left(		t^{-1}		
		t^{\tau - \zeta} q^{ T^{(d)}_{\zeta} 	- T^{(d+1)}_{\tau} 		}					\right)
	}{
		f_{q,t}\left(	t^{\tau - \zeta} q^{ T^{(d)}_{\zeta} 	- T^{(d+1)}_{\tau}		}						\right)
	}. 
	\notag 
\end{align}
\normalsize
From this, we obtain that 
\footnotesize
\begin{align}
	&
	\prod_{\zeta = 1}^{\ell(T_2)}
	\prod_{d = N+M+1}^{N+M+L}
	\prod_{\tau = \zeta + 1}^{\ell(T_2)}
	\prod_{p = d+1}^{N+M+L}
	\prod_{\Xi = 1}^{\theta_{\zeta,d}}
	\prod_{\omega = 1}^{\theta_{\tau,p}}
	\Biggl\{
	\\
	&\frac{
		\left(	1 - q^{-1}					
		(t^{-1})^{\tau - \zeta} q^{	\sum_{\gamma = p + 1}^{N + M + L} \theta_{\tau,\gamma}	+ \omega	- (\sum_{\gamma = d + 1}^{N + M + L} \theta_{\zeta,\gamma} + \Xi)		}	
		\right)
		\left(	1 - qt^{-1}			
		(t^{-1})^{\tau - \zeta} q^{	\sum_{\gamma = p + 1}^{N + M + L} \theta_{\tau,\gamma}	+ \omega	- (\sum_{\gamma = d + 1}^{N + M + L} \theta_{\zeta,\gamma} + \Xi)		}
		\right)
	}{
		\left(	1 - q			
		(t^{-1})^{\tau - \zeta} q^{	\sum_{\gamma = p + 1}^{N + M + L} \theta_{\tau,\gamma}	+ \omega	- (\sum_{\gamma = d + 1}^{N + M + L} \theta_{\zeta,\gamma} + \Xi)		}
		\right)
		\left(	1 - q^{-1}t			
		(t^{-1})^{\tau - \zeta} q^{	\sum_{\gamma = p + 1}^{N + M + L} \theta_{\tau,\gamma}	+ \omega	- (\sum_{\gamma = d + 1}^{N + M + L} \theta_{\zeta,\gamma} + \Xi)		}
		\right)
	}
	\notag 
	\\
	&\times
	\frac{
		\left(	1 - t		
		(t^{-1})^{\tau - \zeta} q^{	\sum_{\gamma = p + 1}^{N + M + L} \theta_{\tau,\gamma}	+ \omega	- (\sum_{\gamma = d + 1}^{N + M + L} \theta_{\zeta,\gamma} + \Xi)		}
		\right)
	}{
		\left(	1 - t^{-1}			
		(t^{-1})^{\tau - \zeta} q^{	\sum_{\gamma = p + 1}^{N + M + L} \theta_{\tau,\gamma}	+ \omega	- (\sum_{\gamma = d + 1}^{N + M + L} \theta_{\zeta,\gamma} + \Xi)		}
		\right)
	}
	\Biggr\}
	\notag 
	\\
	&=
	\prod_{\zeta = 1}^{\ell(T_2)}
	\Biggl\{
	\frac{f_{q,t}\left(		t^{\ell(T_2) - \zeta} q^{ T^{(N+M+1)}_{\zeta}		}					\right)}{
		f_{q,t}\left(		t^{\ell(T_2) - \zeta} 					\right)
	}
	\times 
	\frac{
		f_{q,t}\left(	1					\right)
	}{
		f_{q,t}\left(q^{ T^{(N+M+1)}_{\zeta} 		}						\right)
	} 
	\Biggr\}
	\times
	\notag 
	\\
	&\hspace{0.3cm}\times 
	\prod_{\zeta = 1}^{\ell(T_2)}
	\prod_{d = N+M+1}^{N+M+L}
	\prod_{\tau = \zeta + 1}^{\ell(T_2)}
	\Biggl\{
	\frac{
		f_{q,t}\left(	t^{\tau - \zeta} q^{ T^{(d+1)}_{\zeta} 	- T^{(d+1)}_{\tau}	 		}						\right)
	}{
		f_{q,t}\left(		t^{-1}		
		t^{\tau - \zeta} q^{ T^{(d+1)}_{\zeta} 	- T^{(d+1)}_{\tau}			}					\right)
	}
	\times
	\frac{
		f_{q,t}\left(		t^{-1}		
		t^{\tau - \zeta} q^{ T^{(d)}_{\zeta} 	- T^{(d+1)}_{\tau} 		}					\right)
	}{
		f_{q,t}\left(	t^{\tau - \zeta} q^{ T^{(d)}_{\zeta} 	- T^{(d+1)}_{\tau}		}						\right)
	}
	\Biggr\}.
	\notag  
\end{align}
\normalsize
\end{proof}

\begin{cor}\mbox{}
\label{cor4-17aug-1339-sun}
\footnotesize
\begin{align}
&\prod_{\zeta = 1}^{\ell(T_2)}
\prod_{d = N+M+1}^{N+M+L}
\prod_{\alpha = N+M+1}^{d - 1} 
\prod_{\omega = 1}^{\theta_{\zeta,\alpha}}
\frac{
	\left(	1 - q^{	\sum_{\gamma = \alpha + 1}^{N + M + L} \theta_{\zeta,\gamma}	+ \omega - \sum_{\gamma = d+1}^{N + M + L}\theta_{\zeta,\gamma} 		}			\right)
	\left(	1 - tq^{	\sum_{\gamma = \alpha + 1}^{N + M + L} \theta_{\zeta,\gamma}	+ \omega - (\sum_{\gamma = d}^{N + M + L}\theta_{\zeta,\gamma} 	+ 1 )	}		\right)
}{
	\left(	1 - q^{	\sum_{\gamma = \alpha + 1}^{N + M + L} \theta_{\zeta,\gamma}	+ \omega - \sum_{\gamma = d}^{N + M + L}\theta_{\zeta,\gamma} 		}			\right)
	\left(	1 - 	tq^{	\sum_{\gamma = \alpha + 1}^{N + M + L} \theta_{\zeta,\gamma}	+ \omega - (\sum_{\gamma = d + 1}^{N + M + L}\theta_{\zeta,\gamma} 	+ 1 )	}			\right)
}
\notag 
\\
&\times
\prod_{\zeta = 1}^{\ell(T_2)}
\prod_{d = N+M+1}^{N+M+L}
\prod_{\tau = \zeta + 1}^{\ell(T_2)}
\prod_{p = d+1}^{N+M+L}
\prod_{\Xi = 1}^{\theta_{\zeta,d}}
\prod_{\omega = 1}^{\theta_{\tau,p}}
\Biggl\{
\notag 
\\
&\frac{
	\left(	1 - q^{-1}					
	(t^{-1})^{\tau - \zeta} q^{	\sum_{\gamma = p + 1}^{N + M + L} \theta_{\tau,\gamma}	+ \omega	- (\sum_{\gamma = d + 1}^{N + M + L} \theta_{\zeta,\gamma} + \Xi)		}	
	\right)
	\left(	1 - qt^{-1}			
	(t^{-1})^{\tau - \zeta} q^{	\sum_{\gamma = p + 1}^{N + M + L} \theta_{\tau,\gamma}	+ \omega	- (\sum_{\gamma = d + 1}^{N + M + L} \theta_{\zeta,\gamma} + \Xi)		}
	\right)
}{
	\left(	1 - q			
	(t^{-1})^{\tau - \zeta} q^{	\sum_{\gamma = p + 1}^{N + M + L} \theta_{\tau,\gamma}	+ \omega	- (\sum_{\gamma = d + 1}^{N + M + L} \theta_{\zeta,\gamma} + \Xi)		}
	\right)
	\left(	1 - q^{-1}t			
	(t^{-1})^{\tau - \zeta} q^{	\sum_{\gamma = p + 1}^{N + M + L} \theta_{\tau,\gamma}	+ \omega	- (\sum_{\gamma = d + 1}^{N + M + L} \theta_{\zeta,\gamma} + \Xi)		}
	\right)
}
\notag 
\\
&\times
\frac{
	\left(	1 - t		
	(t^{-1})^{\tau - \zeta} q^{	\sum_{\gamma = p + 1}^{N + M + L} \theta_{\tau,\gamma}	+ \omega	- (\sum_{\gamma = d + 1}^{N + M + L} \theta_{\zeta,\gamma} + \Xi)		}
	\right)
}{
	\left(	1 - t^{-1}			
	(t^{-1})^{\tau - \zeta} q^{	\sum_{\gamma = p + 1}^{N + M + L} \theta_{\tau,\gamma}	+ \omega	- (\sum_{\gamma = d + 1}^{N + M + L} \theta_{\zeta,\gamma} + \Xi)		}
	\right)
}
\Biggr\}
\label{eqna7-1343-17aug-sun}
\\
&= 
\prod_{d = N+M+1}^{N+M+L}
\prod_{\zeta = 1}^{\ell(T_2)}
\prod_{\tau = \zeta}^{\ell(T_2)}
\frac{
	f_{q,t}\left(	t^{\tau - \zeta} q^{ T^{(d+1)}_{\zeta} 	- T^{(d+1)}_{\tau}	 		}						\right)
}{
	f_{q,t}\left(	t^{\tau - \zeta} q^{ T^{(d)}_{\zeta} 	- T^{(d+1)}_{\tau}		}						\right)
}
\frac{
	f_{q,t}\left(		
	t^{\tau - \zeta} q^{ T^{(d)}_{\zeta} 	- T^{(d+1)}_{\tau + 1} 		}					\right)
}{
	f_{q,t}\left(		
	t^{\tau- \zeta} q^{ T^{(d+1)}_{\zeta} 	- T^{(d+1)}_{\tau + 1}			}					\right)
}
\label{eqn-a7-1329-17aug}
\\
&= \prod_{d = N+M+1}^{N+M+L}\fraks{C}_{T^{(d)}/T^{(d+1)}}(q,t)
\end{align}
\normalsize
\end{cor}
\begin{proof}
According to \textbf{Propositions} \textbf{\ref{prpa1-17aug-sun-1329}} and \textbf{\ref{prpa3-17aug-sun-1329}}, we obtain that the quantity in equation \eqref{eqna7-1343-17aug-sun} is equal to 
\footnotesize
\begin{align}
&= 
\prod_{d = N+M+1}^{N+M+L}
\prod_{\zeta = 1}^{\ell(T_2)}
\left[
\frac{
	f_{q,t}(1)
}{
	f_{q,t}(q^{T^{(d)}_{\zeta} - T^{(d+1)}_{\zeta}})
}
\right]
\times
\prod_{\zeta = 1}^{\ell(T_2)}
\Biggl\{
\frac{f_{q,t}\left(		t^{\ell(T_2) - \zeta} q^{ T^{(N+M+1)}_{\zeta}		}					\right)}{
	f_{q,t}\left(		t^{\ell(T_2) - \zeta} 					\right)
}
\Biggr\}
\times
\\
&\hspace{0.3cm}\times 
\prod_{\zeta = 1}^{\ell(T_2)}
\prod_{d = N+M+1}^{N+M+L}
\prod_{\tau = \zeta + 1}^{\ell(T_2)}
\Biggl\{
\frac{
	f_{q,t}\left(	t^{\tau - \zeta} q^{ T^{(d+1)}_{\zeta} 	- T^{(d+1)}_{\tau}	 		}						\right)
}{
	f_{q,t}\left(	t^{\tau - \zeta} q^{ T^{(d)}_{\zeta} 	- T^{(d+1)}_{\tau}		}						\right)
}
\Biggr\} 
\times
\prod_{\zeta = 1}^{\ell(T_2)}
\prod_{d = N+M+1}^{N+M+L}
\prod_{\tau = \zeta + 1}^{\ell(T_2)}
\Biggl\{
\frac{
	f_{q,t}\left(		t^{-1}		
	t^{\tau - \zeta} q^{ T^{(d)}_{\zeta} 	- T^{(d+1)}_{\tau} 		}					\right)
}{
	f_{q,t}\left(		t^{-1}		
	t^{\tau - \zeta} q^{ T^{(d+1)}_{\zeta} 	- T^{(d+1)}_{\tau}			}					\right)
}
\Biggr\} 
\notag 
\end{align}
\normalsize
Since 
\footnotesize
\begin{align}
	&\prod_{\zeta = 1}^{\ell(T_2)}
	\prod_{d = N+M+1}^{N+M+L}
	\prod_{\tau = \zeta + 1}^{\ell(T_2)}
	\Biggl\{
	\frac{
		f_{q,t}\left(		t^{-1}		
		t^{\tau - \zeta} q^{ T^{(d)}_{\zeta} 	- T^{(d+1)}_{\tau} 		}					\right)
	}{
		f_{q,t}\left(		t^{-1}		
		t^{\tau - \zeta} q^{ T^{(d+1)}_{\zeta} 	- T^{(d+1)}_{\tau}			}					\right)
	}
	\Biggr\} 
	\\
	&= 
	\prod_{\zeta = 1}^{\ell(T_2)}
	\prod_{d = N+M+1}^{N+M+L}
	\prod_{\tau = \zeta}^{\ell(T_2) - 1}
	\Biggl\{
	\frac{
		f_{q,t}\left(		
		t^{\tau - \zeta} q^{ T^{(d)}_{\zeta} 	- T^{(d+1)}_{\tau + 1} 		}					\right)
	}{
		f_{q,t}\left(		
		t^{\tau- \zeta} q^{ T^{(d+1)}_{\zeta} 	- T^{(d+1)}_{\tau + 1}			}					\right)
	}
	\Biggr\} 
	\notag 
	\\
	&=
	\prod_{\zeta = 1}^{\ell(T_2)}
	\prod_{d = N+M+1}^{N+M+L}
	\prod_{\tau = \zeta}^{\ell(T_2)}
	\frac{
		f_{q,t}\left(		
		t^{\tau - \zeta} q^{ T^{(d)}_{\zeta} 	- T^{(d+1)}_{\tau + 1} 		}					\right)
	}{
		f_{q,t}\left(		
		t^{\tau- \zeta} q^{ T^{(d+1)}_{\zeta} 	- T^{(d+1)}_{\tau + 1}			}					\right)
	}
	\times 
	\prod_{\zeta = 1}^{\ell(T_2)}
	\frac{
		f_{q,t}\left(		
		t^{\ell(T_2) - \zeta} 				\right)
	}{
		f_{q,t}\left(		
		t^{\ell(T_2) - \zeta} q^{ T^{(N+M+1)}_{\zeta} 		}					\right)
	},
	\notag 
\end{align}
\normalsize
we obtain that the LHS of equation \eqref{eqn-a7-1329-17aug} is equal to 
\footnotesize
\begin{align}
	&= 
	\prod_{d = N+M+1}^{N+M+L}
	\prod_{\zeta = 1}^{\ell(T_2)}
	\left[
	\frac{
		f_{q,t}(1)
	}{
		f_{q,t}(q^{T^{(d)}_{\zeta} - T^{(d+1)}_{\zeta}})
	}
	\right]
	\\
	&\hspace{0.3cm}\times 
	\prod_{\zeta = 1}^{\ell(T_2)}
	\prod_{d = N+M+1}^{N+M+L}
	\prod_{\tau = \zeta + 1}^{\ell(T_2)}
	\Biggl\{
	\frac{
		f_{q,t}\left(	t^{\tau - \zeta} q^{ T^{(d+1)}_{\zeta} 	- T^{(d+1)}_{\tau}	 		}						\right)
	}{
		f_{q,t}\left(	t^{\tau - \zeta} q^{ T^{(d)}_{\zeta} 	- T^{(d+1)}_{\tau}		}						\right)
	}
	\Biggr\} 
	\times
	\prod_{\zeta = 1}^{\ell(T_2)}
	\prod_{d = N+M+1}^{N+M+L}
	\prod_{\tau = \zeta}^{\ell(T_2)}
	\frac{
		f_{q,t}\left(		
		t^{\tau - \zeta} q^{ T^{(d)}_{\zeta} 	- T^{(d+1)}_{\tau + 1} 		}					\right)
	}{
		f_{q,t}\left(		
		t^{\tau- \zeta} q^{ T^{(d+1)}_{\zeta} 	- T^{(d+1)}_{\tau + 1}			}					\right)
	}
	\notag 
	\\
	&= 
	\prod_{d = N+M+1}^{N+M+L}
	\prod_{\zeta = 1}^{\ell(T_2)}
	\prod_{\tau = \zeta}^{\ell(T_2)}
	\frac{
		f_{q,t}\left(	t^{\tau - \zeta} q^{ T^{(d+1)}_{\zeta} 	- T^{(d+1)}_{\tau}	 		}						\right)
	}{
		f_{q,t}\left(	t^{\tau - \zeta} q^{ T^{(d)}_{\zeta} 	- T^{(d+1)}_{\tau}		}						\right)
	}
	\frac{
		f_{q,t}\left(		
		t^{\tau - \zeta} q^{ T^{(d)}_{\zeta} 	- T^{(d+1)}_{\tau + 1} 		}					\right)
	}{
		f_{q,t}\left(		
		t^{\tau- \zeta} q^{ T^{(d+1)}_{\zeta} 	- T^{(d+1)}_{\tau + 1}			}					\right)
	}
	\notag 
\end{align}
\normalsize

Next, observe that for $\tau \in \{1,\dots,\ell(T_2)\}$ such that $\tau > \ell(T^{(d+1)})$, 
\footnotesize
\begin{align}
	\frac{
		f_{q,t}\left(	t^{\tau - \zeta} q^{ T^{(d+1)}_{\zeta} 	- T^{(d+1)}_{\tau}	 		}						\right)
	}{
		f_{q,t}\left(	t^{\tau - \zeta} q^{ T^{(d)}_{\zeta} 	- T^{(d+1)}_{\tau}		}						\right)
	}
	\frac{
		f_{q,t}\left(		
		t^{\tau - \zeta} q^{ T^{(d)}_{\zeta} 	- T^{(d+1)}_{\tau + 1} 		}					\right)
	}{
		f_{q,t}\left(		
		t^{\tau- \zeta} q^{ T^{(d+1)}_{\zeta} 	- T^{(d+1)}_{\tau + 1}			}					\right)
	} 
	= 1
\end{align}
\normalsize
Therefore, 
\footnotesize
\begin{align}
	&\prod_{\zeta = 1}^{\ell(T_2)}
	\prod_{\tau = \zeta}^{\ell(T_2)}
	\frac{
		f_{q,t}\left(	t^{\tau - \zeta} q^{ T^{(d+1)}_{\zeta} 	- T^{(d+1)}_{\tau}	 		}						\right)
	}{
		f_{q,t}\left(	t^{\tau - \zeta} q^{ T^{(d)}_{\zeta} 	- T^{(d+1)}_{\tau}		}						\right)
	}
	\frac{
		f_{q,t}\left(		
		t^{\tau - \zeta} q^{ T^{(d)}_{\zeta} 	- T^{(d+1)}_{\tau + 1} 		}					\right)
	}{
		f_{q,t}\left(		
		t^{\tau- \zeta} q^{ T^{(d+1)}_{\zeta} 	- T^{(d+1)}_{\tau + 1}			}					\right)
	}
	\\ 
	&= 
	\prod_{\zeta = 1}^{\ell(T_2)}
	\prod_{\tau = \zeta}^{\ell(T^{(d+1)})}
	\frac{
		f_{q,t}\left(	t^{\tau - \zeta} q^{ T^{(d+1)}_{\zeta} 	- T^{(d+1)}_{\tau}	 		}						\right)
	}{
		f_{q,t}\left(	t^{\tau - \zeta} q^{ T^{(d)}_{\zeta} 	- T^{(d+1)}_{\tau}		}						\right)
	}
	\frac{
		f_{q,t}\left(		
		t^{\tau - \zeta} q^{ T^{(d)}_{\zeta} 	- T^{(d+1)}_{\tau + 1} 		}					\right)
	}{
		f_{q,t}\left(		
		t^{\tau- \zeta} q^{ T^{(d+1)}_{\zeta} 	- T^{(d+1)}_{\tau + 1}			}					\right)
	} 
	\notag 
	\\
	&= 
	\prod_{1 \leq \zeta \leq \tau \leq \ell(T^{(d+1)})}
	\frac{
		f_{q,t}\left(	t^{\tau - \zeta} q^{ T^{(d+1)}_{\zeta} 	- T^{(d+1)}_{\tau}	 		}						\right)
	}{
		f_{q,t}\left(	t^{\tau - \zeta} q^{ T^{(d)}_{\zeta} 	- T^{(d+1)}_{\tau}		}						\right)
	}
	\frac{
		f_{q,t}\left(		
		t^{\tau - \zeta} q^{ T^{(d)}_{\zeta} 	- T^{(d+1)}_{\tau + 1} 		}					\right)
	}{
		f_{q,t}\left(		
		t^{\tau- \zeta} q^{ T^{(d+1)}_{\zeta} 	- T^{(d+1)}_{\tau + 1}			}					\right)
	} 
	\notag 
	\\
	&= \fraks{C}_{T^{(d)}/T^{(d+1)}}(q,t).
	\notag 
\end{align}
\normalsize
Thus, we have proved the \textbf{Corollary \ref{cor4-17aug-1339-sun}}. 
\end{proof}

From \textbf{Proposition \ref{prpa2-17aug-sun-1329}} and \textbf{Corollary \ref{cor4-17aug-1339-sun}}, we get that the RHS of equation \eqref{eqn35-1146-17aug} is equal to 
\begin{align}
	& 
	\prod_{d = N+M+1}^{N+M+L}\fraks{C}_{T^{(d)}/T^{(d+1)}}(q,t)
	\times 
	\prod_{\zeta = 1}^{\ell(T_2)}
	\prod_{d = N+M+1}^{N+M+L}
	\prod_{j = 1}^{\theta_{\zeta,d} - 1}
	\frac{
		(1 - qt^{-1}q^{ j})(1 - t)
	}{
		(1 - qt^{-1})(1 - tq^{ j})
	}
	\\
	&\times
	\prod_{d = N+M+1}^{N+M+L} 
	\prod_{\zeta = 1}^{\ell(T_2)}
	\prod_{\tau = \zeta + 1}^{\ell(T_2)} 
	\prod_{\alpha = 0}^{\theta_{\zeta,d} - 1}
	\frac{
		\left(		1 - t t^{\tau - \zeta} q^{ T^{(d+1)}_{\zeta} + \alpha	 -	T^{(d)}_{\tau}			}					\right)
		\left(
		1 - t^{\tau - \zeta} q^{ T^{(d+1)}_{\zeta} + \alpha	-	T^{(d+1)}_{\tau}			}	
		\right)
	}{
		\left(
		1 - t^{\tau - \zeta} q^{ T^{(d+1)}_{\zeta} + \alpha	 -	T^{(d)}_{\tau} 		}	
		\right)
		\left(		1 - t t^{\tau - \zeta} q^{	T^{(d+1)}_{\zeta} + \alpha	 - T^{(d+1)}_{\tau}			}					\right)
	}.
	\notag 
\end{align}
However, this is nothing but the definition of $\cals{A}_2(T;q,t)$ (\textbf{Definition \ref{dfn213-1508-1aug}}). Thus, we have proved the \textbf{Lemma \ref{prp317-2237-16aug}}. 

\section{Proof of Lemma \ref{lemm42-1252-25jul}}
\label{appA-1254}

The goal of this appendix is to prove \textbf{Lemma \ref{lemm42-1252-25jul}}. The idea of the proof is similar to the case of the quantum corner VOA with $\vec{c} = (3^N1^M)$ as discussed in \textbf{Appendix A} of paper \cite{CSW2025}. However, since this paper focuses on the quantum corner VOA with $\vec{c} = (3^N1^M2^L)$, it is necessary to generalize the explanations provided in \textbf{Appendix A} of paper \cite{CSW2025}. Here, we will provide detailed proofs only for statements that have changed. In cases where the statement has not changed, we will refer to paper \cite{CSW2025} for proof details. 

From equation \eqref{eqn21-1326-27jul}, we know that 
\begin{align}
	\widetilde{T}^{\vec{c},\vec{u}}_1(z)
	&= 
	\sum_{i = 1}^{n}
	\frac{
		q_{c_i}^{\frac{1}{2}} - q_{c_i}^{-\frac{1}{2}}
	}{
		q_{3}^{\frac{1}{2}} - q_{3}^{-\frac{1}{2}}
	}
	\normord{
		\widetilde{		\Lambda		}^{\vec{c},\vec{u}}_i(z)
	}
\end{align}
Thus, 
\footnotesize
\begin{align}
	&
	(
	\dualmap
	\comp 
	\bigg|_{
		\substack{
			q_1 = q, \\
			q_2 = q^{-1}t,\\
			q_3 = t^{-1} \\
		}
	}
	)
	\left(
	\cals{N}_{\lambda}(z_1,\dots,z_k )
	\times
	\prod_{1 \leq i < j \leq k}f^{\vec{c}}_{11}\left(\frac{z_j}{z_i} \right)
	\times
	\langle 0 |\widetilde{T}^{\vec{c},\vec{u}}_{1}(z_1 )\cdots \widetilde{T}^{\vec{c},\vec{u}}_{1}(z_k )|0\rangle
	\right)
	\\
	&= 
	\sum_{i_1 = 1}^{N+M+L}
	\cdots
	\sum_{i_k = 1}^{N+M+L}
	(
	\dualmap
	\comp 
	\bigg|_{
		\substack{
			q_1 = q, \\
			q_2 = q^{-1}t,\\
			q_3 = t^{-1} \\
		}
	}
	)
	\bigg[
	y_{i_1}\cdots y_{i_k}
	u_{i_1}\cdots u_{i_k}
	\times
	\cals{N}_\lambda(z_1,\dots,z_k)
	\times
	\prod_{1 \leq a < b \leq k}
	\cals{D}^{(i_a,i_b)}\left(
	\frac{z_b}{z_a}
	; q, t
	\right)
	\bigg]
	\notag 
\end{align}
\normalsize
where 
$\displaystyle y_j :=
\frac{
	q_{c_j}^{\frac{1}{2}} - q_{c_j}^{-\frac{1}{2}}
}{
	q_{3}^{\frac{1}{2}} - q_{3}^{-\frac{1}{2}}
}
\,\, (j = 1,\dots,N+M+L)$
and 
\footnotesize 
\begin{align}
	\cals{D}^{(i,j)}\left(
	\frac{z_b}{z_a}
	; q, t
	\right)
	:= 
	\begin{cases}
		\displaystyle 
		\frac{
			\left(1 - q_1^{-1}\frac{z_b}{z_a}\right)
			\left(1 - q_2^{-1}\frac{z_b}{z_a}\right)
		}{
			\left(1 - q_3\frac{z_b}{z_a}\right)
			\left(1 - \frac{z_b}{z_a}\right)
		}
		\hspace{0.3cm}
		&\text{ if } i < j
		\\
		\displaystyle 
		\frac{
			\left(1 - q_2^{-1}\frac{z_b}{z_a}\right)
			\left(1 - q_2\frac{z_b}{z_a}\right)
		}{
			\left(1 - q_3^{-1}\frac{z_b}{z_a}\right)
			\left(1 - q_3\frac{z_b}{z_a}\right)
		}
		\hspace{0.3cm}
		&\text{ if } i = j = \text{ hyper-number }
		\\
		\displaystyle 
		\frac{
			\left(1 - q_1^{-1}\frac{z_b}{z_a}\right)
			\left(1 - q_1\frac{z_b}{z_a}\right)
		}{
			\left(1 - q_3^{-1}\frac{z_b}{z_a}\right)
			\left(1 - q_3\frac{z_b}{z_a}\right)
		}
		\hspace{0.3cm}
		&\text{ if } i = j = \text{ super-number }
		\\
		1
		\hspace{0.3cm}
		&\text{ if } i = j = \text{ ordinary-number }
		\\
		\displaystyle
		\frac{
			\left(1 - q_1\frac{z_b}{z_a}\right)
			\left(1 - q_2\frac{z_b}{z_a}\right)
		}{
			\left(1 - q_3^{-1}\frac{z_b}{z_a}\right)
			\left(1 - \frac{z_b}{z_a}\right)
		}
		\hspace{0.3cm}
		&\text{ if } i > j 
	\end{cases}
\end{align}
\normalsize

\begin{prop}[\cite{CSW2025}]\mbox{}
\label{prp-a1-1339-27jul}
Let $\lambda \in \operatorname{Par}(k)$. Then, the following statements hold:
\begin{enumerate}[(1)]
\item 
$\cals{N}_{\lambda}(z_1,\dots,z_k )$ contains exactly the factor $(1-t\xi)^{\ell(\lambda) - 1}$.
\item
$\cals{N}_{\lambda}(z_1,\dots,z_k )$ can be written as 

\footnotesize
\begin{align}
	&\cals{N}_{\lambda}(z_1,\dots,z_k )
	= 
	\Delta\left(	q_3^{-\frac{1}{2}}\frac{z_{\lambda_1 + 1}}{z_1}				\right)^{-1}
	\Delta\left(	q_3^{-\frac{1}{2}}\frac{z_{\lambda_1 + \lambda_2 + 1}}{z_{\lambda_1 + 1}}				\right)^{-1}
	\times 
	\cdots 
	\times 
	\Delta\left(	q_3^{-\frac{1}{2}}\frac{z_{\sum_{j = 1}^{\ell(\lambda) - 1}\lambda_j+1}}{z_{\sum_{j = 1}^{\ell(\lambda) - 2}\lambda_j+1}}				\right)^{-1}
	\times
	\cals{A}
	\label{d7-1950-21jan}
\end{align}
\normalsize
where the limit 
\begin{align}
	\lim_{\xi \rightarrow t^{-1}}\,\,
	(
	\dualmap
	\comp 
	\bigg|_{
		\substack{
			q_1 = q, \\
			q_2 = q^{-1}t,\\
			q_3 = t^{-1} \\
		}
	}
	)
	\left(
	\cals{A}
	\right)
\end{align}
exists and be nonzero. 
\end{enumerate}
\end{prop}

From \textbf{Proposition \ref{prp-a1-1339-27jul}}, we obtain that 
\footnotesize
\begin{align}
	&
	\lim_{\xi \rightarrow t^{-1}}\,\,
	(
	\dualmap
	\comp 
	\bigg|_{
		\substack{
			q_1 = q, \\
			q_2 = q^{-1}t,\\
			q_3 = t^{-1} \\
		}
	}
	)
	\left(
	\cals{N}_{\lambda}(z_1,\dots,z_k )
	\times
	\prod_{1 \leq i < j \leq k}f^{\vec{c}}_{11}\left(\frac{z_j}{z_i} \right)
	\times
	\langle 0 |\widetilde{T}^{\vec{c},\vec{u}}_{1}(z_1 )\cdots \widetilde{T}^{\vec{c},\vec{u}}_{1}(z_k )|0\rangle
	\right)
	\label{d7-eqn-1830-21jan}
	\\
	&= 
	\lim_{\xi \rightarrow t^{-1}}\,\,
	(
	\dualmap
	\comp 
	\bigg|_{
		\substack{
			q_1 = q, \\
			q_2 = q^{-1}t,\\
			q_3 = t^{-1} \\
		}
	}
	)
	\left(
	\cals{A}
	\right)
	\times
	\notag 
	\\
	&\hspace{0.3cm} \times \sum_{i_1 = 1}^{N+M}
	\cdots
	\sum_{i_k = 1}^{N+M}
	\lim_{\xi \rightarrow t^{-1}}\,\,
	(
	\dualmap
	\comp 
	\bigg|_{
		\substack{
			q_1 = q, \\
			q_2 = q^{-1}t,\\
			q_3 = t^{-1} \\
		}
	}
	)
	\Bigg[
	\Delta\left(	q_3^{-\frac{1}{2}}\frac{z_{\lambda_1 + 1}}{z_1}				\right)^{-1}
	\times 
	\cdots 
	\times 
	\Delta\left(	q_3^{-\frac{1}{2}}\frac{z_{\sum_{j = 1}^{\ell(\lambda) - 1}\lambda_j+1}}{z_{\sum_{j = 1}^{\ell(\lambda) - 2}\lambda_j+1}}				\right)^{-1}
	\notag 
	\\
	&\hspace{5.1cm}\times
	\prod_{1 \leq i < j \leq k}f^{\vec{c}}_{11}\left(\frac{z_j}{z_i} \right)
	\times
	y_{i_1} \cdots y_{i_k}
	\times 
	\langle 0 |			\normord{
		\widetilde{\Lambda}^{\vec{c},\vec{u}}_{i_1}(z_1)
	}
	\times
	\cdots
	\times
	\normord{
		\widetilde{\Lambda}^{\vec{c},\vec{u}}_{i_k}(z_k)
	}				
	|0\rangle
	\Bigg].
	\notag 
	\\
	&= 
	\lim_{\xi \rightarrow t^{-1}}\,\,
	(
	\dualmap
	\comp 
	\bigg|_{
		\substack{
			q_1 = q, \\
			q_2 = q^{-1}t,\\
			q_3 = t^{-1} \\
		}
	}
	)
	\left(
	\cals{A}
	\right)
	\times
	\notag 
	\\
	&\hspace{0.3cm} \times \sum_{i_1 = 1}^{N+M}
	\cdots
	\sum_{i_k = 1}^{N+M}
	\lim_{\xi \rightarrow t^{-1}}\,\,
	(
	\dualmap
	\comp 
	\bigg|_{
		\substack{
			q_1 = q, \\
			q_2 = q^{-1}t,\\
			q_3 = t^{-1} \\
		}
	}
	)
	\Bigg[
	y_{i_1}\cdots y_{i_k}
	u_{i_1}\cdots u_{i_k}
	\notag 
	\\
	&\hspace{4.7cm}\times
	\Delta\left(	q_3^{-\frac{1}{2}}\frac{z_{\lambda_1 + 1}}{z_1}				\right)^{-1}
	\times 
	\cdots 
	\times 
	\Delta\left(	q_3^{-\frac{1}{2}}\frac{z_{\sum_{j = 1}^{\ell(\lambda) - 1}\lambda_j+1}}{z_{\sum_{j = 1}^{\ell(\lambda) - 2}\lambda_j+1}}				\right)^{-1}
	\times
	\prod_{1 \leq a < b \leq k}
	\cals{D}^{(i_a,i_b)}\left(
	\frac{z_b}{z_a}
	; q, t
	\right)
	\Bigg]
	\notag
\end{align}
\normalsize

\begin{lem}[\cite{CSW2025}]
\label{lemma2-1359-27jul}
Let $\lambda \in \operatorname{Par}(k)$. 
Suppose that $(i_1,\dots,i_k) \in \{1,\dots,N+M+L\}^k$ such that $T(i_1,\dots,i_k;\lambda)$ is a tableau with an adjacent pair of boxes in a row that breaks the reverse SSYTT conditions. Then, we have 

\footnotesize
\begin{align}
	&\lim_{\xi \rightarrow t^{-1}}\,\,
	(
	\dualmap
	\comp 
	\bigg|_{
		\substack{
			q_1 = q, \\
			q_2 = q^{-1}t,\\
			q_3 = t^{-1} \\
		}
	}
	)
	\Bigg[
	y_{i_1}\cdots y_{i_k}
	u_{i_1}\cdots u_{i_k}
	\\
	&\hspace{2.9cm}\times
	\Delta\left(	q_3^{-\frac{1}{2}}\frac{z_{\lambda_1 + 1}}{z_1}				\right)^{-1}
	\times 
	\cdots 
	\times 
	\Delta\left(	q_3^{-\frac{1}{2}}\frac{z_{\sum_{j = 1}^{\ell(\lambda) - 1}\lambda_j+1}}{z_{\sum_{j = 1}^{\ell(\lambda) - 2}\lambda_j+1}}				\right)^{-1}
	\times
	\prod_{1 \leq a < b \leq k}
	\cals{D}^{(i_a,i_b)}\left(
	\frac{z_b}{z_a}
	; q, t
	\right)
	\Bigg]
	\notag
	\\
	&= 0. 
	\notag 
\end{align}
\normalsize 
\end{lem}

From the \textbf{Lemma \ref{lemma2-1359-27jul}} above, we have shown that if the reverse SSYTT condition is broken in the row direction, then the contribution from it is $0$. Therefore, in order to prove \textbf{Lemma \ref{lemm42-1252-25jul}}, it is sufficient to show that if $T(i_1,\dots,i_k;\lambda)$ is a tableau where every adjacent pair of boxes in a row satisfies the reverse SSYTT conditions, but there is an adjacent pair of boxes in a column that breaks the reverse SSYTT condition, then

\footnotesize
\begin{align}
	&\lim_{\xi \rightarrow t^{-1}}\,\,
	(
	\dualmap
	\comp 
	\bigg|_{
		\substack{
			q_1 = q, \\
			q_2 = q^{-1}t,\\
			q_3 = t^{-1} \\
		}
	}
	)
	\Bigg[
	y_{i_1}\cdots y_{i_k}
	u_{i_1}\cdots u_{i_k}
	\\
	&\hspace{2.9cm}\times
	\Delta\left(	q_3^{-\frac{1}{2}}\frac{z_{\lambda_1 + 1}}{z_1}				\right)^{-1}
	\times 
	\cdots 
	\times 
	\Delta\left(	q_3^{-\frac{1}{2}}\frac{z_{\sum_{j = 1}^{\ell(\lambda) - 1}\lambda_j+1}}{z_{\sum_{j = 1}^{\ell(\lambda) - 2}\lambda_j+1}}				\right)^{-1}
	\times
	\prod_{1 \leq a < b \leq k}
	\cals{D}^{(i_a,i_b)}\left(
	\frac{z_b}{z_a}
	; q, t
	\right)
	\Bigg]
	\notag
	\\
	&= 0. 
	\notag 
\end{align}
\normalsize 
To prove this statement, we need the concept of triangle-from cancellation and the concepts of breaking pair, breaking triangle, and breaking band, which are explained in detail in the following subsections.

\subsection{Triangle-form Cancellation}

\begin{prop}
\label{prpa3-1521-27jul}
Suppose that 
\begin{align}
	\begin{ytableau}
		i_{c}  &   i_{c+1} \\
		\none & i_{d}
	\end{ytableau}
	\label{D22-1725-17jan}
\end{align}
are boxes in a reverse SSYTT $T(i_1,\dots,i_k;\lambda)$. Then, there will be no factor $(1 - t\xi)$ appearing in the numerator or denominator of the quantity 
\begin{align}
	(
	\dualmap
	\comp 
	\bigg|_{
		\substack{
			q_1 = q, \\
			q_2 = q^{-1}t,\\
			q_3 = t^{-1} \\
		}
	}
	)
	\Bigg[
	\underbrace{	\prod_{1 \leq a < b \leq k}				}_{a,b \in \{c,c+1,d\}}
	\cals{D}^{(i_a,i_b)}\left(
	\frac{z_b}{z_a}
	; q, t
	\right)
	\Bigg]. 
	\label{A22-eqn-1555-23jan}
\end{align}
\end{prop}
\begin{proof}
We will divide the proof into 3 cases: 
\begin{enumerate}
	\item $i_c$ is a hyper number
	\item $i_c$ is a super number
	\item $i_c$ is an ordinary number
\end{enumerate}
In the cases where $i_c$ is an ordinary number and $i_c$ is a super number, we have already proved this in \cite{CSW2025}. Therefore, in this proof, we will only prove the case where $i_c$ is a hyper number, which can be further divided into subcases as follows:
\begin{itemize}
	\item $i_c, i_{c+1}, i_d$ are hyper numbers such that $i_c > i_{c+1} > i_d$. In this case, we have 
	\footnotesize
	\begin{align}
		&(
		\dualmap
		\comp 
		\bigg|_{
			\substack{
				q_1 = q, \\
				q_2 = q^{-1}t,\\
				q_3 = t^{-1} \\
			}
		}
		)
		\Bigg[
		\underbrace{	\prod_{1 \leq a < b \leq k}				}_{a,b \in \{c,c+1,d\}}
		\cals{D}^{(i_a,i_b)}\left(
		\frac{z_b}{z_a}
		; q, t
		\right)
		\Bigg]
		= 
		\frac{
			(1 - q^2)(1 - t)(1 - q^2\xi)(1 - q^{-1}t\xi)
		}{
			(1 - tq)(1 - q)(1 - tq\xi)(1 - \xi)
		}
	\end{align}
	\normalsize
	\item $i_c, i_{c+1}, i_d$ are hyper numbers such that $i_c > i_{c+1} = i_d$. In this case, we have 
	\footnotesize
	\begin{align}
		&(
		\dualmap
		\comp 
		\bigg|_{
			\substack{
				q_1 = q, \\
				q_2 = q^{-1}t,\\
				q_3 = t^{-1} \\
			}
		}
		)
		\Bigg[
		\underbrace{	\prod_{1 \leq a < b \leq k}				}_{a,b \in \{c,c+1,d\}}
		\cals{D}^{(i_a,i_b)}\left(
		\frac{z_b}{z_a}
		; q, t
		\right)
		\Bigg]
		= 
		\frac{
			(1 - q^2)(1 - t)(1 - q^2\xi)(1 - qt^{-1}\xi)(1 - q^{-1}t\xi)
		}{
			(1 - tq)(1 - q)(1 - tq\xi)(1 - q\xi)(1 - t^{-1}\xi)
		}
	\end{align}
	\normalsize
	\item $i_c, i_{c+1}, i_d$ are hyper-numbers such that $i_c = i_{c+1} > i_d$. In this case, we have 
	\footnotesize
	\begin{align}
		&(
		\dualmap
		\comp 
		\bigg|_{
			\substack{
				q_1 = q, \\
				q_2 = q^{-1}t,\\
				q_3 = t^{-1} \\
			}
		}
		)
		\Bigg[
		\underbrace{	\prod_{1 \leq a < b \leq k}				}_{a,b \in \{c,c+1,d\}}
		\cals{D}^{(i_a,i_b)}\left(
		\frac{z_b}{z_a}
		; q, t
		\right)
		\Bigg]
		= 
		\frac{
			(1 - t^{-1}q^2)(1 - t)(1 - q^2\xi)(1 - q^{-1}t\xi)
		}{
			(1 - tq)(1 - t^{-1}q)(1 - tq\xi)(1 - \xi)
		}
	\end{align}
	\normalsize
	\item $i_c, i_{c+1}, i_d$ are hyper-numbers such that $i_c = i_{c+1} = i_d$. In this case, we have 
	\footnotesize
	\begin{align}
		&(
		\dualmap
		\comp 
		\bigg|_{
			\substack{
				q_1 = q, \\
				q_2 = q^{-1}t,\\
				q_3 = t^{-1} \\
			}
		}
		)
		\Bigg[
		\underbrace{	\prod_{1 \leq a < b \leq k}				}_{a,b \in \{c,c+1,d\}}
		\cals{D}^{(i_a,i_b)}\left(
		\frac{z_b}{z_a}
		; q, t
		\right)
		\Bigg]
		= 
		\frac{
			(1 - t^{-1}q^2)(1 - t)(1 - t^{-1}q^2\xi)(1 - q^{-1}t\xi)
		}{
			(1 - tq)(1 - t^{-1}q)(1 - tq\xi)(1 - t^{-1}\xi)
		}
	\end{align}
	\normalsize
	\item $i_c$ is a hyper number, $i_{c+1}$ is a super number, $i_d$ is a super number such that $i_{c+1} = i_d$.
	In this case, we have $i_c > i_{c+1} = i_d$. Thus, 
	\footnotesize
	\begin{align}
		&(
		\dualmap
		\comp 
		\bigg|_{
			\substack{
				q_1 = q, \\
				q_2 = q^{-1}t,\\
				q_3 = t^{-1} \\
			}
		}
		)
		\Bigg[
		\underbrace{	\prod_{1 \leq a < b \leq k}				}_{a,b \in \{c,c+1,d\}}
		\cals{D}^{(i_a,i_b)}\left(
		\frac{z_b}{z_a}
		; q, t
		\right)
		\Bigg]
		= 
		\frac{
			(1 - q^2)(1 - t)(1 - q^2\xi)(1 - qt^{-1}\xi)(1 - q^{-1}t\xi)
		}{
			(1 - tq)(1 - q)(1 - tq\xi)(1 - q\xi)(1 - t^{-1}\xi)
		}
	\end{align}
	\normalsize
	\item $i_c$ is a hyper number, $i_{c+1}$ is a super number, $i_d$ is a super number such that $i_{c+1} > i_d$. 
	In this case, we have $i_c > i_{c+1} > i_d$. Thus, 
	\footnotesize
	\begin{align}
		&(
		\dualmap
		\comp 
		\bigg|_{
			\substack{
				q_1 = q, \\
				q_2 = q^{-1}t,\\
				q_3 = t^{-1} \\
			}
		}
		)
		\Bigg[
		\underbrace{	\prod_{1 \leq a < b \leq k}				}_{a,b \in \{c,c+1,d\}}
		\cals{D}^{(i_a,i_b)}\left(
		\frac{z_b}{z_a}
		; q, t
		\right)
		\Bigg]
		= 
		\frac{
			(1 - q^2)(1 - t)(1 - q^2\xi)(1 - q^{-1}t\xi)
		}{
			(1 - tq)(1 - q)(1 - tq\xi)(1 - \xi)
		}
	\end{align}
	\normalsize
	\item $i_c$ is a hyper number, $i_{c+1}$ is a super number, $i_d$ is an ordinary number. In this case, we have 
	$i_c > i_{c+1} > i_d$. Thus, 
	\footnotesize
	\begin{align}
		&(
		\dualmap
		\comp 
		\bigg|_{
			\substack{
				q_1 = q, \\
				q_2 = q^{-1}t,\\
				q_3 = t^{-1} \\
			}
		}
		)
		\Bigg[
		\underbrace{	\prod_{1 \leq a < b \leq k}				}_{a,b \in \{c,c+1,d\}}
		\cals{D}^{(i_a,i_b)}\left(
		\frac{z_b}{z_a}
		; q, t
		\right)
		\Bigg]
		= 
		\frac{
			(1 - q^2)(1 - t)(1 - q^2\xi)(1 - q^{-1}t\xi)
		}{
			(1 - tq)(1 - q)(1 - tq\xi)(1 - \xi)
		}
	\end{align}
	\normalsize
	\item $i_c$ is a hyper number, $i_{c+1}$ is an ordinary number, $i_d$ is an ordinary number such that $i_{c+1} > i_d$. In this case, we have $i_c > i_{c+1} > i_d$. Thus, 
	\footnotesize
	\begin{align}
		&(
		\dualmap
		\comp 
		\bigg|_{
			\substack{
				q_1 = q, \\
				q_2 = q^{-1}t,\\
				q_3 = t^{-1} \\
			}
		}
		)
		\Bigg[
		\underbrace{	\prod_{1 \leq a < b \leq k}				}_{a,b \in \{c,c+1,d\}}
		\cals{D}^{(i_a,i_b)}\left(
		\frac{z_b}{z_a}
		; q, t
		\right)
		\Bigg]
		= 
		\frac{
			(1 - q^2)(1 - t)(1 - q^2\xi)(1 - q^{-1}t\xi)
		}{
			(1 - tq)(1 - q)(1 - tq\xi)(1 - \xi)
		}
	\end{align}
	\normalsize
\end{itemize}
We can see that in all cases, there will be no factor $(1 - t\xi)$ appearing in the numerator or denominator. Thus, we have completed the proof. 
\end{proof}

\begin{prop}
\label{prp-a4-1634-27jul-sun}
Suppose that $T(i_1,\dots,i_k;\lambda)$ is a reverse SSYTT with shape $\lambda$ where $|\lambda| = k$. Then, we have 
\begin{align}
	(
	\dualmap
	\comp 
	\bigg|_{
		\substack{
			q_1 = q, \\
			q_2 = q^{-1}t,\\
			q_3 = t^{-1} \\
		}
	}
	)
	\Bigg[
	\prod_{1 \leq a < b \leq k}			
	\cals{D}^{(i_a,i_b)}\left(
	\frac{z_b}{z_a}
	; q, t
	\right)
	\Bigg]
	=
	\frac{1}{(1 - t\xi)^{\ell(\lambda)-1}} \times \cals{F}, 
\end{align}
where $\cals{F}$ is a nonzero elements in $\bb{Q}(q,t,\xi)$ and no factor $(1 - t\xi)$ appears in the numerator or denominator of $\cals{F}$. 
\end{prop}
\begin{proof}
To determine the total number of $(1-t\xi)$ remaining after cancellations, it is sufficient to consider only the $\cals{D}^{(i,j)}\left(
\frac{z_b}{z_a}
; q, t
\right)$ 
arising from pairs of boxes that are adjacent in a column or diagonally. Moreover, by using \textbf{Proposition \ref{prpa3-1521-27jul}}, we know that we only need to examine the first column of $T(i_1,\dots,i_k;\lambda)$. That is, what we need to analyze are any adjacent pairs of boxes in the first column only, i.e. 
\begin{align}
	\begin{ytableau}
		\scriptstyle	i_1  \\ 
		\scriptstyle	i_{\lambda_1 + 1} 
	\end{ytableau}
	\hspace{0.2cm}
	,
	\begin{ytableau}
		\scriptstyle	i_{\lambda_1 + 1} \\
		\scriptstyle	i_{\lambda_1 + \lambda_2 + 1}
	\end{ytableau}
	\hspace{0.2cm}
	,
	\cdots
\end{align}
Since the argument used to analyze each adjacent pair of boxes is the same, here we will show the details only for the case of the pair of boxes 
\begin{align}
	\begin{ytableau}
		\scriptstyle	i_1  \\ 
		\scriptstyle	i_{\lambda_1 + 1} 
	\end{ytableau}
	\,\, . 
\end{align}

Since $T(i_1,\dots,i_k;\lambda)$ is a reverse SSYTT, the possible cases are 
\begin{itemize}
	\item $i_1 > i_{\lambda_1 + 1}$. In this case, we have 
	\footnotesize
	\begin{align}
		(
		\dualmap
		\comp 
		\bigg|_{
			\substack{
				q_1 = q, \\
				q_2 = q^{-1}t,\\
				q_3 = t^{-1} \\
			}
		}
		)
		\Bigg[
		\underbrace{		\prod_{1 \leq a < b \leq k}					}_{
			\substack{
				a = 1, \,\, b = \lambda_1 + 1
			}
		}
		\cals{D}^{(i_a,i_b)}\left(
		\frac{z_b}{z_a}
		; q, t
		\right)
		\Bigg]
		=
		\frac{
			(1 - q\xi)(1 - q^{-1}t\xi)
		}{
			(1 - t\xi)(1 - \xi)
		}
	\end{align}
	\normalsize
	\item $i_1 = i_{\lambda_1 + 1} = \text{hyper-number}$. In this case, we have 
	\footnotesize
	\begin{align}
		(
		\dualmap
		\comp 
		\bigg|_{
			\substack{
				q_1 = q, \\
				q_2 = q^{-1}t,\\
				q_3 = t^{-1} \\
			}
		}
		)
		\Bigg[
		\underbrace{		\prod_{1 \leq a < b \leq k}					}_{
			\substack{
				a = 1, \,\, b = \lambda_1 + 1
			}
		}
		\cals{D}^{(i_a,i_b)}\left(
		\frac{z_b}{z_a}
		; q, t
		\right)
		\Bigg]
		= 
		\frac{
			(1 - qt^{-1}\xi)
			(1 - q^{-1}t\xi)
		}{
			(1 - t\xi)
			(1 - t^{-1}\xi)
		}
	\end{align}
	\normalsize
	\item $i_1 = i_{\lambda_1 + 1} = \text{super-number}$. In this case, we have 
	\footnotesize
	\begin{align}
		(
		\dualmap
		\comp 
		\bigg|_{
			\substack{
				q_1 = q, \\
				q_2 = q^{-1}t,\\
				q_3 = t^{-1} \\
			}
		}
		)
		\Bigg[
		\underbrace{		\prod_{1 \leq a < b \leq k}					}_{
			\substack{
				a = 1, \,\, b = \lambda_1 + 1
			}
		}
		\cals{D}^{(i_a,i_b)}\left(
		\frac{z_b}{z_a}
		; q, t
		\right)
		\Bigg]
		= 
		\frac{
			(1 - q\xi)
			(1 - q^{-1}\xi)
		}{
			(1 - t\xi)
			(1 - t^{-1}\xi)
		}
	\end{align}
	\normalsize
\end{itemize}
We can immediately see that in all cases, there will be one factor of $(1 - t\xi)$ in the denominator. Since there are a toal of $\ell(\lambda) - 1$ adjacent pairs of boxes in the first column, we have 
\begin{align}
	(
	\dualmap
	\comp 
	\bigg|_{
		\substack{
			q_1 = q, \\
			q_2 = q^{-1}t,\\
			q_3 = t^{-1} \\
		}
	}
	)
	\Bigg[
	\prod_{1 \leq a < b \leq k}			
	\cals{D}^{(i_a,i_b)}\left(
	\frac{z_b}{z_a}
	; q, t
	\right)
	\Bigg]
	=
	\frac{1}{(1 - t\xi)^{\ell(\lambda)-1}} \times \cals{F}, 
\end{align}
where $\cals{F}$ is a nonzero elements in $\bb{Q}(q,t,\xi)$ and no factor $(1 - t\xi)$ appears in the numerator or denominator of $\cals{F}$. 
\end{proof}

From Propositions \ref{prp-a1-1339-27jul} and \ref{prp-a4-1634-27jul-sun}, we get the following corollary. 

\begin{cor}
If $\lambda \in \operatorname{Par}$  and $T(i_1,\dots,i_k;\lambda)$ is a reverse SSYTT with shape $\lambda$, then 

\footnotesize
\begin{align}
	&\lim_{\xi \rightarrow t^{-1}}\,\,
	(
	\dualmap
	\comp 
	\bigg|_{
		\substack{
			q_1 = q, \\
			q_2 = q^{-1}t,\\
			q_3 = t^{-1} \\
		}
	}
	)
	\Bigg[
	y_{i_1}\cdots y_{i_k}
	u_{i_1}\cdots u_{i_k}
	\\
	&\hspace{3.2cm}\times
	\Delta\left(	q_3^{-\frac{1}{2}}\frac{z_{\lambda_1 + 1}}{z_1}				\right)^{-1}
	\times 
	\cdots 
	\times 
	\Delta\left(	q_3^{-\frac{1}{2}}\frac{z_{\sum_{j = 1}^{\ell(\lambda) - 1}\lambda_j+1}}{z_{\sum_{j = 1}^{\ell(\lambda) - 2}\lambda_j+1}}				\right)^{-1}
	\times
	\prod_{1 \leq a < b \leq k}
	\cals{D}^{(i_a,i_b)}\left(
	\frac{z_b}{z_a}
	; q, t
	\right)
	\Bigg]
	\notag 
	\\
	&\neq 0 
	\notag 
\end{align}
\normalsize
\end{cor}

\subsection{Breaking Pair, Breaking Triangle, and Breaking Band}

In this subsection, we will define the concepts of breaking pair, breaking triangle, and breaking band. The definitions of these concepts will be similar to those we provided in \cite{CSW2025}, except that in \cite{CSW2025} we discussed the case of reverse SSYBT, while here we will discuss these concepts in the case of reverse SSYTT, which is a more general case. 

\begin{dfn}[Breaking pair]
Let $\lambda \in \operatorname{Par}(k)$, let $\{i_1,\dots,i_k\} \in \{1,\dots,N+M+L\}^k$ such that $T(i_1,\dots,i_k;\lambda)$ is a Young tableau with shape $\lambda$ where every adjacent pair of boxes in any row satisfies the reverse SSYTT condition. Consider an adjacent pair of boxes 
\begin{align}
	\begin{ytableau}
		i_a
		\\
		i_b
	\end{ytableau}
	\label{eqn-a29-1554-26jan}
\end{align}
in a column. In this situation, we say that \eqref{eqn-a29-1554-26jan} is a \textbf{breaking pair} if it violates the reverse SSYTT condition. In cases where \eqref{eqn-a29-1554-26jan} is not a breaking pair, we say it is a non-breaking pair. 
\end{dfn}

\begin{dfn}[Breaking triangle]
Let $\lambda \in \operatorname{Par}(k)$, let $\{i_1,\dots,i_k\} \in \{1,\dots,N+M+L\}^k$ such that $T(i_1,\dots,i_k;\lambda)$ is a Young tableau with shape $\lambda$ where every adjacent pair of boxes in any row satisfies the reverse SSYTT condition. Let 
\begin{align}
	\begin{ytableau}
		i_a
		\\
		i_b
	\end{ytableau}
\end{align}
be a breaking pair of $T(i_1,\dots,i_k;\lambda)$ that is not in the first column. Then, we call 
\begin{align}
	\begin{ytableau}
		i_{a-1} & i_a
		\\
		\none & i_b
	\end{ytableau}
\end{align}
a \textbf{breaking triangle} of $T(i_1,\dots,i_k;\lambda)$. 
\end{dfn}

Next, suppose that 
\begin{align}
	\begin{ytableau}
		i_{a-1} & i_a
		\\
		\none & i_b
	\end{ytableau}
\end{align}
is a breaking triangle of $T(i_1,\dots,i_k;\lambda)$. Recall from the definition of a breaking triangle that we assume every row in $T(i_1,\dots,i_k;\lambda)$ does not violate the reverse SSYTT conditions. Therefore, we can conclude that one of the cases listed below will occur:
\begin{itemize}
	\item $i_{a-1} > i_a$
	\item $i_{a-1} = i_a = \text{ ordinary number }$
	\item $i_{a-1} = i_a = \text{ hyper number }$
\end{itemize}
On the other hand, from the definition of a breaking triangle, we know that 
\begin{align}
	\begin{ytableau}
		i_a
		\\
		i_b
	\end{ytableau}
\end{align}
will definitely violate the reverse SSYTT condition. Therefore, we can conclude that one of the cases listed below will occur:
\begin{itemize}
	\item $i_a < i_b$ or
	\item  $i_a = i_b = \text{ordinary number}$ 
\end{itemize}

From this explanation, we get that all possible cases are:
\begin{enumerate}[(1)]
	\item \label{case1-1429}$i_a < i_b < i_{a-1}$
	\item  \label{case2-1429} $i_a< i_b = i_{a-1}$
	\item  \label{case3-1429}$i_a < i_{a-1} < i_b$
	\item  \label{case4-1429}$\text{ordinary number} = i_a = i_{a-1} < i_b$
	\item  \label{case5-1429}$\text{ordinary number} = i_a = i_b < i_{a-1}$
	\item  \label{case6-1429}$\text{ordinary number} = i_a = i_b = i_{a-1}$
	\item \label{case7-1429} $\text{hyper number} = i_{a-1} = i_a < i_b$
\end{enumerate}
We have already calculated the value of 
\begin{align}
	&(
	\dualmap
	\comp 
	\bigg|_{
		\substack{
			q_1 = q, \\
			q_2 = q^{-1}t,\\
			q_3 = t^{-1} \\
		}
	}
	)
	\Bigg[
	\underbrace{		\prod_{1 \leq c < d \leq k}			}_{c,d \in \{a-1,a,b\}}
	\cals{D}^{(i_c,i_d)}\left(
	\frac{z_d}{z_c}
	; q, t
	\right)
	\Bigg]
\end{align}
for the breaking triangle in cases \ref{case1-1429} - \ref{case6-1429} in \cite{CSW2025}. Therefore, here we will only calculate the case \ref{case7-1429}. Since $\text{hyper number} = i_{a-1} = i_a < i_b$, we get that 
\begin{align}
	&(
	\dualmap
	\comp 
	\bigg|_{
		\substack{
			q_1 = q, \\
			q_2 = q^{-1}t,\\
			q_3 = t^{-1} \\
		}
	}
	)
	\Bigg[
	\underbrace{		\prod_{1 \leq c < d \leq k}			}_{c,d \in \{a-1,a,b\}}
	\cals{D}^{(i_c,i_d)}\left(
	\frac{z_d}{z_c}
	; q, t
	\right)
	\Bigg]
	= 
	\frac{
		(1 - t^{-1}q^2)(1 - t)(1 - q^{-1}\xi)(1 - t^{-1}q^2\xi)
	}{
		(1 - tq)(1 - t^{-1}q)(1 - t^{-1}\xi)(1 - q\xi)
	},
\label{eqn-a-21-1814-27jul}
\end{align}
and get that 
\begin{align}
	\begin{ytableau}
		i_{a-1}
		\\
		i_{b-1}
	\end{ytableau} 
\label{eqn-a-22-1814-27jul}
\end{align}
will also be a breaking pair.

This analysis motivates us to define types of breaking triangles as we will define in the definition below. 

\begin{dfn}
Suppose that 
\begin{align}
	\begin{ytableau}
		i_{a-1} & i_a
		\\
		\none & i_b
	\end{ytableau}
\label{eqna23-1759-27jul}
\end{align}
is a breaking triangle of $T(i_1,\dots,i_k;\lambda)$. We say that it is an \textbf{unstoppable breaking triangle} if one of the following conditions holds:
\begin{enumerate}[(1)]
	\item 	$i_a< i_b = i_{a-1}$,
	\item  $i_a < i_{a-1} < i_b$,
	\item  $\text{ordinary number} = i_a = i_{a-1} < i_b$,
	\item  $\text{ordinary number} = i_a = i_b = i_{a-1}$.
	\item $\text{hyper number} = i_{a-1} = i_a < i_b$
\end{enumerate}
Conversely, if the breaking triangle \eqref{eqna23-1759-27jul} is not an unstoppable breaking triangle, then we say that it is a \textbf{stoppable breaking triangle}.
\end{dfn}

\begin{rem}
From the explanation above, we deduce that a breaking triangle is a stoppable breaking triangle if one of the following conditions holds:
\begin{enumerate}[(1)]
	\item	$i_a < i_b < i_{a-1}$,
	\item  $\text{ordinary number} = i_a = i_b < i_{a-1}$. 
\end{enumerate}
\end{rem}

From equations \eqref{eqn-a-21-1814-27jul}, \eqref{eqn-a-22-1814-27jul} above, and the \textbf{Propositions A.8.} and \textbf{A.9.} of \cite{CSW2025}, 
the \textbf{Propositions \ref{prp-a10-1816-27jul}} and \textbf{\ref{prp-a11-1816-27jul}} below immediately follow.

\begin{prop}
\label{prp-a10-1816-27jul}
If 
\begin{align}
	\begin{ytableau}
		i_{a-1} & i_a
		\\
		\none & i_b
	\end{ytableau}
\end{align}
is an unstoppable breaking triangle of $T(i_1,\dots,i_k;\lambda)$, then we get that 
\begin{align}
	\begin{ytableau}
		i_{a-1}
		\\
		i_{b-1}
	\end{ytableau} 
\end{align}
is a breaking pair of $T(i_1,\dots,i_k;\lambda)$. 
\end{prop}

\begin{prop}
\label{prp-a11-1816-27jul}
If 
\begin{align}
	\begin{ytableau}
		i_{a-1} & i_a
		\\
		\none & i_b
	\end{ytableau}
	\label{eqn-a26-1829-27jul}
\end{align}
is an unstoppable breaking triangle of $T(i_1,\dots,i_k;\lambda)$, then there will be no factor of $(1-t\xi)$ appearing in the numerator and denominator of 
\begin{align}
	&(
	\dualmap
	\comp 
	\bigg|_{
		\substack{
			q_1 = q, \\
			q_2 = q^{-1}t,\\
			q_3 = t^{-1} \\
		}
	}
	)
	\Bigg[
	\underbrace{		\prod_{1 \leq c < d \leq k}			}_{c,d \in \{a-1,a,b\}}
	\cals{D}^{(i_c,i_d)}\left(
	\frac{z_d}{z_c}
	; q, t
	\right)
	\Bigg].
	\label{eqn-a-27-1830-27jul}
\end{align}
Conversely, if \eqref{eqn-a26-1829-27jul} is a stoppable breaking triangle, then there will be a factor of $(1-t\xi)$ appearing in the numerator of \eqref{eqn-a-27-1830-27jul}, and only one such factor. 
\end{prop}

\begin{dfn}[Breaking band]
Let $\lambda \in \operatorname{Par}(k)$, let $\{i_1,\dots,i_k\} \in \{1,\dots,N+M+L\}^k$ such that $T(i_1,\dots,i_k;\lambda)$ is a Young tableau with shape $\lambda$ where every adjacent pair of boxes in any row satisfies the reverse SSYTT condition. Given that 
\begin{align}
	\begin{ytableau}
		i_a & i_{a+1}
		\\
		i_b & i_{b+1}
	\end{ytableau}
	\cdots
	\begin{ytableau}
		i_{a + \beta}
		\\
		i_{b + \beta}
	\end{ytableau}
	\label{eqn-a50-1051-27jan}
\end{align}
is part of the Young tableau $T(i_1,\dots,i_k;\lambda)$, we say that \eqref{eqn-a50-1051-27jan} is a \textbf{breaking band} starting from column $j$ and ending at column $i$ (where $ j \geq i$) if all the conditions stated below are true:
\begin{enumerate}[(1)]
\item For each $\gamma \in \{0,\dots,\beta\}$
\begin{align}
	\begin{ytableau}
		i_{a+\gamma}
		\\
		i_{b+\gamma}
	\end{ytableau} 
\end{align}
is a breaking pair.
\item The breaking pair 
\begin{align}
	\begin{ytableau}
		i_{a+\beta}
		\\
		i_{b+\beta}
	\end{ytableau} 
\end{align}
is in column $j$.
\item The breaking pair 
\begin{align}
	\begin{ytableau}
		i_{a}
		\\
		i_{b}
	\end{ytableau} 
\end{align}
is in column $i$. 
\item The pair of boxes 
\begin{align}
	\begin{ytableau}
		i_{a+\beta + 1}
		\\
		i_{b+\beta + 1}
	\end{ytableau} 
\end{align}
(if it exists) is not a breaking pair. 
\item The pair of boxes 
\begin{align}
	\begin{ytableau}
		i_{a - 1}
		\\
		i_{b - 1}
	\end{ytableau} 
\end{align}
(if it exists) is not a breaking pair. 
\end{enumerate}
\end{dfn}

\begin{prop}
Assume that
\begin{align}
	\begin{ytableau}
		i_a & i_{a+1}
		\\
		i_b & i_{b+1}
	\end{ytableau}
	\cdots
	\begin{ytableau}
		i_{a + \beta}
		\\
		i_{b + \beta}
	\end{ytableau}
	\label{eqna56-1256-27jan}
\end{align}
is a breaking band starting from column $j$ and ending at column $i$. Then, the following statements are true:
\begin{enumerate}[(1)]
\item If $i > 1$, then 
\begin{align}
	&(
	\dualmap
	\comp 
	\bigg|_{
		\substack{
			q_1 = q, \\
			q_2 = q^{-1}t,\\
			q_3 = t^{-1} \\
		}
	}
	)
	\Bigg[
	\underbrace{		\prod_{1 \leq c < d \leq k}			}_{c,d \in \{a-1,a,\dots,a+\beta, b, \dots, b+\beta\}}
	\cals{D}^{(i_c,i_d)}\left(
	\frac{z_d}{z_c}
	; q, t
	\right)
	\Bigg]
	= 
	(1-t\xi)^n
	\times
	\text{other factor}, 
	\label{a58-1147-27jan}
\end{align}
for some integer $n \in \bb{Z}^{\geq 1}$ and no factor $(1 - t\xi)$ appears in the numerator or denominator of \say{other factor} in \eqref{a58-1147-27jan}. 
\item If $i = 1$, then 
\begin{align}
	&(
	\dualmap
	\comp 
	\bigg|_{
		\substack{
			q_1 = q, \\
			q_2 = q^{-1}t,\\
			q_3 = t^{-1} \\
		}
	}
	)
	\Bigg[
	\underbrace{		\prod_{1 \leq c < d \leq k}			}_{c,d \in \{a,\dots,a+\beta, b, \dots, b+\beta\}}
	\cals{D}^{(i_c,i_d)}\left(
	\frac{z_d}{z_c}
	; q, t
	\right)
	\Bigg]
	= 
	(1-t\xi)^n
	\times
	\text{other factor}, 
	\label{a59-1147-27jan}
\end{align}
for some integer $n \in \bb{Z}^{\geq 0}$ and no factor $(1 - t\xi)$ appears in the numerator or denominator of \say{other factor} in \eqref{a59-1147-27jan}. 
\end{enumerate}
\end{prop}
\begin{proof}
The proof is identical to that of \textbf{Proposition A.11} in \cite{CSW2025}. This is because the proof of \textbf{Proposition A.11} in \cite{CSW2025} relies solely on facts about breaking pairs and breaking bands discussed in preceding propositions. As explained above, all statements remain true in the case of a tritableau.
\end{proof}

\begin{thm}
\label{thma14-1742-31jul}
Let $\lambda \in \operatorname{Par}(k)$, let $\{i_1,\dots,i_k\} \in \{1,\dots,N+M+L\}^k$ such that $T(i_1,\dots,i_k;\lambda)$ is a Young tableau with shape $\lambda$ where every adjacent pair of boxes in any row satisfies the reverse SSYTT condition. If $T(i_1,\dots,i_k;\lambda)$ contains at least one breaking band, then
\begin{align}
	&(
	\dualmap
	\comp 
	\bigg|_{
		\substack{
			q_1 = q, \\
			q_2 = q^{-1}t,\\
			q_3 = t^{-1} \\
		}
	}
	)
	\Bigg[
	\prod_{1 \leq c < d \leq k}
	\cals{D}^{(i_c,i_d)}\left(
	\frac{z_d}{z_c}
	; q, t
	\right)
	\Bigg]
	= 
	\frac{1}{	(1-t\xi)^n	}
	\times
	\text{other factors}
	\label{a72-eqn-1600-27jan}
\end{align}
where $n \leq \ell(\lambda) - 2$ and no factor $(1 - t\xi)$ appears in the numerator or denominator of \say{other factor} in \eqref{a72-eqn-1600-27jan}. 
\end{thm}
\begin{proof}
The proof is identical to that of \textbf{Theorem A.12} in \cite{CSW2025}. This is because the proof of \textbf{Theorem A.12} in \cite{CSW2025} relies solely on \textbf{Proposition A.11} in \cite{CSW2025} and facts about breaking pairs and breaking bands discussed in preceding propositions. As explained above, all statements remain true in the case of a tritableau.
\end{proof}

As a consequence of \textbf{Proposition \ref{prp-a1-1339-27jul}}, 
\textbf{Lemma \ref{lemma2-1359-27jul}}, 
and \textbf{Theorem \ref{thma14-1742-31jul}}, we have 

\footnotesize
\begin{align}
	&
	\lim_{\xi \rightarrow t^{-1}}\,\,
	(
	\dualmap
	\comp 
	\bigg|_{
		\substack{
			q_1 = q, \\
			q_2 = q^{-1}t,\\
			q_3 = t^{-1} \\
		}
	}
	)
	\left(
	\cals{N}_{\lambda}(z_1,\dots,z_k )
	\times
	\prod_{1 \leq i < j \leq k}f^{\vec{c}}_{11}\left(\frac{z_j}{z_i} \right)
	\times
	\langle 0 |\widetilde{T}^{\vec{c},\vec{u}}_{1}(z_1 )\cdots \widetilde{T}^{\vec{c},\vec{u}}_{1}(z_k )|0\rangle
	\right)
	\label{eqn-a38-1102-1aug}
	\\
	&= 
	\underbrace{				
		\sum_{i_1 = 1}^{N+M+L}
		\cdots
		\sum_{i_k = 1}^{N+M+L}
	}_{
		T(i_1,\dots,i_k) \in 
		\operatorname{RSSYTT}(N,M,L;\lambda)
	}
	\Biggl\{
	u_{i_1}\cdots u_{i_k}
	\times
	\left(
	\frac{
		q^{\frac{1}{2}} - q^{-\frac{1}{2}}
	}{
		t^{-\frac{1}{2}} - t^{\frac{1}{2}}
	}
	\right)^{|T_1|}
	\times
	\left(
	\frac{
		(q^{-1}t)^{\frac{1}{2}} - (qt^{-1})^{\frac{1}{2}}
	}{
		t^{-\frac{1}{2}} - t^{\frac{1}{2}}
	}
	\right)^{|T_2|}
	\notag 
	\\
	&\hspace{3.3cm}\times
	\lim_{\xi \rightarrow t^{-1}}\,\,
	(
	\dualmap
	\comp 
	\bigg|_{
		\substack{
			q_1 = q, \\
			q_2 = q^{-1}t,\\
			q_3 = t^{-1} \\
		}
	}
	)
	\bigg[
	\cals{N}_\lambda(z_1,\dots,z_k)
	\times
	\prod_{1 \leq a < b \leq k}
	\cals{D}^{(i_a,i_b)}\left(
	\frac{z_b}{z_a}
	; q, t
	\right)
	\bigg]
	\Biggr\}
	\notag 
\end{align}
\normalsize

In the next subsection, we show that for each $T(i_1,\dots,i_k;\lambda) \in 
\operatorname{RSSYTT}(N,M,L;\lambda)$, we have 
\footnotesize
\begin{align}
	&\lim_{\xi \rightarrow t^{-1}}\,\,
	(
	\dualmap
	\comp 
	\bigg|_{
		\substack{
			q_1 = q, \\
			q_2 = q^{-1}t,\\
			q_3 = t^{-1} \\
		}
	}
	)
	\bigg[
	\cals{N}_\lambda(z_1,\dots,z_k)
	\times
	\prod_{1 \leq a < b \leq k}
	\cals{D}^{(i_a,i_b)}\left(
	\frac{z_b}{z_a}
	; q, t
	\right)
	\bigg]
	\label{a39-31jul-1751}
	\\
	&= 
	(
	\widetilde{\Psi}_{\lambda}^{(q,t^{-1})}
	\comp 
	\bigg|_{
		\substack{
			q_1 = q, \\
			q_2 = q^{-1}t,\\
			q_3 = t^{-1} \\
		}
	}
	)
	\bigg[
	\prod_{1 \leq a < b \leq k}
	\cals{C}^{(i_a,i_b)}\left(
	\frac{z_b}{z_a}
	; q, t
	\right)
	\bigg]
	\notag 
\end{align}
\normalsize

\subsection{Derivation of \eqref{a39-31jul-1751}}

Note that 
\footnotesize
\begin{align}
	&\cals{N}_\lambda(z_1,\dots,z_k)
	\times 
	\prod_{1 \leq a < b \leq k}
	\cals{D}^{(i_a,i_b)}\left(
	\frac{z_b}{z_a}
	; q, t
	\right)
	\\
	&= 
	\prod_{c = 1}^{\ell(\lambda)}
	\underbrace{		\prod_{1 \leq a < b \leq k}				}_{
		\operatorname{row}(a) = \operatorname{row}(b) = c
	}
	\cals{C}^{(i_a,i_b)}\left(
	\frac{z_b}{z_a}
	; q, t
	\right)
	\times 
	\prod_{a = 1}^{k - 1}
	\underbrace{		\prod_{b \in \{1,\dots,k\}}				}_{
		\substack{
			(1) \,\, b > a 
			\\
			(2) \,\, \operatorname{row}(a) < \operatorname{row}(b)
			\\
			(3) \,\, i_a \leq i_b
		}
	}
	\cals{C}^{(i_a,i_b)}\left(
	\frac{z_b}{z_a}
	; q, t
	\right).
	\notag  
\end{align}
\normalsize
Thus, 
\footnotesize
\begin{align}
	&(
	\dualmap
	\comp 
	\bigg|_{
		\substack{
			q_1 = q, \\
			q_2 = q^{-1}t,\\
			q_3 = t^{-1} \\
		}
	}
	)
	\bigg[
	\cals{N}_\lambda(z_1,\dots,z_k)
	\times
	\prod_{1 \leq a < b \leq k}
	\cals{D}^{(i_a,i_b)}\left(
	\frac{z_b}{z_a}
	; q, t
	\right)
	\bigg]
	\\
	&= 
	\prod_{c = 1}^{\ell(\lambda)}
	\underbrace{		\prod_{1 \leq a < b \leq k}				}_{
		\operatorname{row}(a) = \operatorname{row}(b) = c
	}
	(
	\dualmap
	\comp 
	\bigg|_{
		\substack{
			q_1 = q, \\
			q_2 = q^{-1}t,\\
			q_3 = t^{-1} \\
		}
	}
	)
	\bigg[
	\cals{C}^{(i_a,i_b)}\left(
	\frac{z_b}{z_a}
	; q, t
	\right)
	\bigg]
	\times 
	\prod_{a = 1}^{k - 1}
	\underbrace{		\prod_{b \in \{1,\dots,k\}}				}_{
		\substack{
			(1) \,\, b > a 
			\\
			(2) \,\, \operatorname{row}(a) < \operatorname{row}(b)
			\\
			(3) \,\, i_a \leq i_b
		}
	}
	(
	\dualmap
	\comp 
	\bigg|_{
		\substack{
			q_1 = q, \\
			q_2 = q^{-1}t,\\
			q_3 = t^{-1} \\
		}
	}
	)
	\bigg[
	\cals{C}^{(i_a,i_b)}\left(
	\frac{z_b}{z_a}
	; q, t
	\right)
	\bigg]
	\notag 
\end{align}
\normalsize

Before proceeding with further analysis, we will discuss the facts that will be used. 

\begin{prop}
\label{prpa15-1033-1aug}
Suppose $a,b$ are boxes in a tableau $T(i_1,\dots,i_k;\lambda)$ (not necessarily reverse SSYTT) such that $a < b$. 
\begin{enumerate}[(1)]
	\item If $\operatorname{row}(a) = \operatorname{row}(b)$ and $i_a > i_b$, then 
	\footnotesize
	\begin{align}
		\bigg|_{
			\substack{
				q_1 = q, \\
				q_2 = q^{-1}t,\\
				q_3 = t^{-1} \\
			}
		}
		\cals{C}^{(i_a,i_b)}\left(
		\frac{z_b}{z_a}
		; q, t
		\right)
		= 
		\frac{
			\left(1 - q\frac{z_b}{z_a}\right)
			\left(1 - q^{-1}t\frac{z_b}{z_a}\right)
		}{
			\left(1 - t\frac{z_b}{z_a}\right)
			\left(1 - \frac{z_b}{z_a}\right)
		}
	\end{align}
	\normalsize
	\item If $\operatorname{row}(a) = \operatorname{row}(b)$ and $i_a = i_b = \text{hyper number}$, then 
	\footnotesize
	\begin{align}
		\bigg|_{
			\substack{
				q_1 = q, \\
				q_2 = q^{-1}t,\\
				q_3 = t^{-1} \\
			}
		}
		\cals{C}^{(i_a,i_b)}\left(
		\frac{z_b}{z_a}
		; q, t
		\right)
		= 
		\frac{
			\left(1 - qt^{-1}\frac{z_b}{z_a}\right)
			\left(1 - q^{-1}t\frac{z_b}{z_a}\right)
		}{
			\left(1 - t\frac{z_b}{z_a}\right)
			\left(1 - t^{-1}\frac{z_b}{z_a}\right)
		}
	\end{align}
	\normalsize
	\item If $\operatorname{row}(a) = \operatorname{row}(b)$ and $i_a = i_b = \text{ordinary number}$, then 
	\footnotesize
	\begin{align}
		\bigg|_{
			\substack{
				q_1 = q, \\
				q_2 = q^{-1}t,\\
				q_3 = t^{-1} \\
			}
		}
		\cals{C}^{(i_a,i_b)}\left(
		\frac{z_b}{z_a}
		; q, t
		\right)
		= 1 
	\end{align}
	\normalsize
	\item If $\operatorname{row}(a) < \operatorname{row}(b)$ and $i_a > i_b$, then 
	\footnotesize
	\begin{align}
		\bigg|_{
			\substack{
				q_1 = q, \\
				q_2 = q^{-1}t,\\
				q_3 = t^{-1} \\
			}
		}
		\cals{C}^{(i_a,i_b)}\left(
		\frac{z_b}{z_a}
		; q, t
		\right)
		= 1
	\end{align}
	\normalsize
	\item If $\operatorname{row}(a) < \operatorname{row}(b)$ and $i_a < i_b$, then 
	\footnotesize
	\begin{align}
		\bigg|_{
			\substack{
				q_1 = q, \\
				q_2 = q^{-1}t,\\
				q_3 = t^{-1} \\
			}
		}
		\cals{C}^{(i_a,i_b)}\left(
		\frac{z_b}{z_a}
		; q, t
		\right)
		= 
		\frac{
			\left(1 - q^{-1}\frac{z_b}{z_a}\right)
			\left(1 - qt^{-1}\frac{z_b}{z_a}\right)
			\left(1 - t\frac{z_b}{z_a}\right)
		}{
			\left(1 - q\frac{z_b}{z_a}\right)
			\left(1 - q^{-1}t\frac{z_b}{z_a}\right)
			\left(1 - t^{-1}\frac{z_b}{z_a}\right)
		}
		\label{C7-1704-21jul}
	\end{align}
	\normalsize
	\item If $\operatorname{row}(a) < \operatorname{row}(b)$ and $i_a = i_b = \text{hyper number}$, then 
	\footnotesize
	\begin{align}
		\bigg|_{
			\substack{
				q_1 = q, \\
				q_2 = q^{-1}t,\\
				q_3 = t^{-1} \\
			}
		}
		\cals{C}^{(i_a,i_b)}\left(
		\frac{z_b}{z_a}
		; q, t
		\right)
		= 
		\frac{
			\left(1 - qt^{-1}\frac{z_b}{z_a}\right)
			\left(1 - \frac{z_b}{z_a}\right)
		}{
			\left(1 - t^{-1}\frac{z_b}{z_a}\right)
			\left(1 - q\frac{z_b}{z_a}\right)
		}
		\label{C8-1704-21jul}
	\end{align}
	\normalsize
	\item If $\operatorname{row}(a) < \operatorname{row}(b)$ and $i_a = i_b = \text{super number}$, then 
	\footnotesize
	\begin{align}
		\bigg|_{
			\substack{
				q_1 = q, \\
				q_2 = q^{-1}t,\\
				q_3 = t^{-1} \\
			}
		}
		\cals{C}^{(i_a,i_b)}\left(
		\frac{z_b}{z_a}
		; q, t
		\right)
		= 
		\frac{
			\left(1 - q^{-1}\frac{z_b}{z_a}\right)
			\left(1 - \frac{z_b}{z_a}\right)
		}{
			\left(1 - q^{-1}t\frac{z_b}{z_a}\right)
			\left(1 - t^{-1}\frac{z_b}{z_a}\right)
		}
		\label{C9-1704-21jul}
	\end{align}
	\normalsize
	\item If $\operatorname{row}(a) < \operatorname{row}(b)$ and $i_a = i_b = \text{ordinary number}$, then 
	\footnotesize
	\begin{align}
		\bigg|_{
			\substack{
				q_1 = q, \\
				q_2 = q^{-1}t,\\
				q_3 = t^{-1} \\
			}
		}
		\cals{C}^{(i_a,i_b)}\left(
		\frac{z_b}{z_a}
		; q, t
		\right)
		= 
		\frac{
			\left(1 - t\frac{z_b}{z_a}\right)
			\left(1 - \frac{z_b}{z_a}\right)
		}{
			\left(1 - q\frac{z_b}{z_a}\right)
			\left(1 - q^{-1}t\frac{z_b}{z_a}\right)
		}
		\label{C10-1704-21jul}
	\end{align}
	\normalsize
	\item If $\operatorname{row}(a) = \operatorname{row}(b)$ and $i_a < i_b$, then 
	\footnotesize
	\begin{align}
		\bigg|_{
			\substack{
				q_1 = q, \\
				q_2 = q^{-1}t,\\
				q_3 = t^{-1} \\
			}
		}
		\cals{C}^{(i_a,i_b)}\left(
		\frac{z_b}{z_a}
		; q, t
		\right)
		= 
		\frac{
			\left(1 - q^{-1}\frac{z_b}{z_a}\right)
			\left(1 - qt^{-1}\frac{z_b}{z_a}\right)
		}{
			\left(1 - t^{-1}\frac{z_b}{z_a}\right)
			\left(1 - \frac{z_b}{z_a}\right)
		}
	\end{align}
	\normalsize
	\item If $\operatorname{row}(a) = \operatorname{row}(b)$ and $i_a = i_b = \text{supernumber}$, then 
	\footnotesize
	\begin{align}
		\bigg|_{
			\substack{
				q_1 = q, \\
				q_2 = q^{-1}t,\\
				q_3 = t^{-1} \\
			}
		}
		\cals{C}^{(i_a,i_b)}\left(
		\frac{z_b}{z_a}
		; q, t
		\right)
		= 
		\frac{
			\left(1 - q^{-1}\frac{z_b}{z_a}\right)
			\left(1 - q\frac{z_b}{z_a}\right)
		}{
			\left(1 - t^{-1}\frac{z_b}{z_a}\right)
			\left(1 - t\frac{z_b}{z_a}\right)
		}
	\end{align}
	\normalsize
\end{enumerate}
\end{prop}
\begin{proof}
The proof is a direct result from the definition of $\cals{C}^{(i_a,i_b)}\left(
\frac{z_b}{z_a}
; q, t
\right)$ 
given in equation \eqref{eqn36-1808-31jul}. 
\end{proof}

\begin{prop}
\label{a16-1103-1aug}
Let $T(i_1,\dots,i_k;\lambda) \in 
\operatorname{RSSYTT}(N,M,L;\lambda)$. Then, the following statements hold:
\begin{enumerate}[(1)]
\item \label{kor1-1027}
If $a,b$ are boxes in $T(i_1,\dots,i_k;\lambda)$ located in the same row such that $a < b$, then 
\begin{align}
	\lim_{\xi \rightarrow t^{-1}}\,\,
	(
	\dualmap
	\comp 
	\bigg|_{
		\substack{
			q_1 = q, \\
			q_2 = q^{-1}t,\\
			q_3 = t^{-1} \\
		}
	}
	)
	\Big[
	\cals{C}^{(i_a,i_b)}\left(
	\frac{z_b}{z_a}
	; q, t
	\right)
	\Big]
	=
	(
	\widetilde{\Psi}_{\lambda}^{(q,t^{-1})}
	\comp 
	\bigg|_{
		\substack{
			q_1 = q, \\
			q_2 = q^{-1}t,\\
			q_3 = t^{-1} \\
		}
	}
	)
	\Big[
	\cals{C}^{(i_a,i_b)}\left(
	\frac{z_b}{z_a}
	; q, t
	\right)
	\Big]
\end{align}
\item \label{kor2-1027}
If $a,b$ are boxes in $T(i_1,\dots,i_k;\lambda)$ that satisfy the following conditions:
\begin{enumerate}[(i)]
\item $a < b $,
\item $\operatorname{row}(a) < \operatorname{row}(b)$,
\item $i_a \leq i_b$,
\end{enumerate}
then 
\begin{align}
	\lim_{\xi \rightarrow t^{-1}}\,\,
	(
	\dualmap
	\comp 
	\bigg|_{
		\substack{
			q_1 = q, \\
			q_2 = q^{-1}t,\\
			q_3 = t^{-1} \\
		}
	}
	)
	\Big[
	\cals{C}^{(i_a,i_b)}\left(
	\frac{z_b}{z_a}
	; q, t
	\right)
	\Big]
	=
	(
	\widetilde{\Psi}_{\lambda}^{(q,t^{-1})}
	\comp 
	\bigg|_{
		\substack{
			q_1 = q, \\
			q_2 = q^{-1}t,\\
			q_3 = t^{-1} \\
		}
	}
	)
	\Big[
	\cals{C}^{(i_a,i_b)}\left(
	\frac{z_b}{z_a}
	; q, t
	\right)
	\Big]
\end{align}
\end{enumerate}
\end{prop}
\begin{proof}
\eqref{kor1-1027} Assume $a,b$ are boxes in $T(i_1,\dots,i_k;\lambda)$ such that $a < b$ and $\operatorname{row}(a) = \operatorname{row}(b)$. 
Since $T(i_1,\dots,i_k)$ is a reverse SSYTT, the only possible cases are 
\begin{itemize}
	\item $\operatorname{row}(a) = \operatorname{row}(b)$ and $i_a > i_b$
	\item $\operatorname{row}(a) = \operatorname{row}(b)$ and $i_a = i_b = \text{hyper number}$
	\item $\operatorname{row}(a) = \operatorname{row}(b)$ and $i_a = i_b = \text{ordinary number}$
\end{itemize}
Since $a < b$ and $\operatorname{row}(a) = \operatorname{row}(b)$, under the map 
\begin{align}
	(
	\dualmap
	\comp 
	\bigg|_{
		\substack{
			q_1 = q, \\
			q_2 = q^{-1}t,\\
			q_3 = t^{-1} \\
		}
	}
	)
\end{align}
$\displaystyle \frac{z_b}{z_a}$ is mapped to $q^b$ for some $b \in \bb{Z}^{\geq 1}$. From \textbf{Proposition \ref{prpa15-1033-1aug}}, we can see that in all three cases, substituting $\displaystyle \frac{z_b}{z_a}$ with $q^b$ will never make the denominator zero. Furthermore, the numerator also does not become zero. Thus, we have proved statement \eqref{kor1-1027}. 
\vspace{0.2cm}

\eqref{kor2-1027} Suppose 
$a,b$ are boxes in $T(i_1,\dots,i_k;\lambda)$ that satisfy the following conditions:
\begin{enumerate}[(i)]
	\item $a < b $,
	\item $\operatorname{row}(a) < \operatorname{row}(b)$,
	\item $i_a \leq i_b$,
\end{enumerate}
Since $a < b$ and $\operatorname{row}(a) < \operatorname{row}(b)$, under the map 
\begin{align}
	(
	\dualmap
	\comp 
	\bigg|_{
		\substack{
			q_1 = q, \\
			q_2 = q^{-1}t,\\
			q_3 = t^{-1} \\
		}
	}
	)
\end{align}
$\displaystyle \frac{z_b}{z_a}$ is mapped to $\xi^aq^b$ for some $a \in \bb{Z}^{\geq 1}$ and $b \in \bb{Z}$. We can see from \textbf{Proposition \ref{prpa15-1033-1aug}} that the only cases that causes a singularity from the substitution $\xi = t^{-1}$ are 
\begin{itemize}
	\item
	$\operatorname{row}(a) < \operatorname{row}(b)$ and $i_a < i_b$, which gives 
	\footnotesize
	\begin{align}
		\bigg|_{
			\substack{
				q_1 = q, \\
				q_2 = q^{-1}t,\\
				q_3 = t^{-1} \\
			}
		}
		\cals{C}^{(i_a,i_b)}\left(
		\frac{z_b}{z_a}
		; q, t
		\right)
		= 
		\frac{
			\left(1 - q^{-1}\frac{z_b}{z_a}\right)
			\left(1 - qt^{-1}\frac{z_b}{z_a}\right)
			\left(1 - t\frac{z_b}{z_a}\right)
		}{
			\left(1 - q\frac{z_b}{z_a}\right)
			\left(1 - q^{-1}t\frac{z_b}{z_a}\right)
			\left(1 - t^{-1}\frac{z_b}{z_a}\right)
		}, 
	\end{align}
	\normalsize
	\item 
	$\operatorname{row}(a) < \operatorname{row}(b)$ and $i_a = i_b = \text{super number}$, which gives 
	\footnotesize
	\begin{align}
		\bigg|_{
			\substack{
				q_1 = q, \\
				q_2 = q^{-1}t,\\
				q_3 = t^{-1} \\
			}
		}
		\cals{C}^{(i_a,i_b)}\left(
		\frac{z_b}{z_a}
		; q, t
		\right)
		= 
		\frac{
			\left(1 - q^{-1}\frac{z_b}{z_a}\right)
			\left(1 - \frac{z_b}{z_a}\right)
		}{
			\left(1 - q^{-1}t\frac{z_b}{z_a}\right)
			\left(1 - t^{-1}\frac{z_b}{z_a}\right)
		}, 
	\end{align}
	\normalsize
	\item 
	$\operatorname{row}(a) < \operatorname{row}(b)$ and $i_a = i_b = \text{ordinary number}$, which gives 
	\footnotesize
	\begin{align}
		\bigg|_{
			\substack{
				q_1 = q, \\
				q_2 = q^{-1}t,\\
				q_3 = t^{-1} \\
			}
		}
		\cals{C}^{(i_a,i_b)}\left(
		\frac{z_b}{z_a}
		; q, t
		\right)
		= 
		\frac{
			\left(1 - t\frac{z_b}{z_a}\right)
			\left(1 - \frac{z_b}{z_a}\right)
		}{
			\left(1 - q\frac{z_b}{z_a}\right)
			\left(1 - q^{-1}t\frac{z_b}{z_a}\right)
		}. 
	\end{align}
	\normalsize
\end{itemize}
This is because the factor $\left(1 - q^{-1}t\frac{z_b}{z_a}\right)$ in the denominators. Note that a singularity from the factor $\left(1 - q^{-1}t\frac{z_b}{z_a}\right)$ can only occur when boxes $a$ and $b$ are diagonally adjacent, as shown in the equation \eqref{A57-1048-1aug} below. 
\begin{align}
	\begin{ytableau}
		i_{a}  &   \\
		\none & i_{b}
	\end{ytableau}
\label{A57-1048-1aug}
\end{align}
However, since $T(i_1,\dots,i_k;\lambda)$ is a reverse SSYTT, we must have either 
$i_a > i_b$ or $i_a = i_b = \text{ hyper number}$. This implies that the conditions for this type of singularity can never be met. Therefore, we have shown that 

\footnotesize
\begin{align}
	(
	\dualmap
	\comp 
	\bigg|_{
		\substack{
			q_1 = q, \\
			q_2 = q^{-1}t,\\
			q_3 = t^{-1} \\
		}
	}
	)
	\Big[
	\cals{C}^{(i_a,i_b)}\left(
	\frac{z_b}{z_a}
	; q, t
	\right)
	\Big]
\end{align}
\normalsize
has no singularity at $\xi = t^{-1}$. This implies that 

\footnotesize
\begin{align}
	\lim_{\xi \rightarrow t^{-1}}\,\,
	(
	\dualmap
	\comp 
	\bigg|_{
		\substack{
			q_1 = q, \\
			q_2 = q^{-1}t,\\
			q_3 = t^{-1} \\
		}
	}
	)
	\Big[
	\cals{C}^{(i_a,i_b)}\left(
	\frac{z_b}{z_a}
	; q, t
	\right)
	\Big]
	=
	(
	\widetilde{\Psi}_{\lambda}^{(q,t^{-1})}
	\comp 
	\bigg|_{
		\substack{
			q_1 = q, \\
			q_2 = q^{-1}t,\\
			q_3 = t^{-1} \\
		}
	}
	)
	\Big[
	\cals{C}^{(i_a,i_b)}\left(
	\frac{z_b}{z_a}
	; q, t
	\right)
	\Big]. 
\end{align}
\normalsize
\end{proof}

From equation \eqref{eqn-a38-1102-1aug} and \textbf{Proposition \ref{a16-1103-1aug}}, we obtain that 
\footnotesize
\begin{align}
	&
	\lim_{\xi \rightarrow t^{-1}}\,\,
	(
	\dualmap
	\comp 
	\bigg|_{
		\substack{
			q_1 = q, \\
			q_2 = q^{-1}t,\\
			q_3 = t^{-1} \\
		}
	}
	)
	\left(
	\cals{N}_{\lambda}(z_1,\dots,z_k )
	\times
	\prod_{1 \leq i < j \leq k}f^{\vec{c}}_{11}\left(\frac{z_j}{z_i} \right)
	\times
	\langle 0 |\widetilde{T}^{\vec{c},\vec{u}}_{1}(z_1 )\cdots \widetilde{T}^{\vec{c},\vec{u}}_{1}(z_k )|0\rangle
	\right)
	\\
	&= 
	\underbrace{				
		\sum_{i_1 = 1}^{N+M+L}
		\cdots
		\sum_{i_k = 1}^{N+M+L}
	}_{
		T(i_1,\dots,i_k) \in 
		\operatorname{RSSYTT}(N,M,L;\lambda)
	}
	\Biggl\{
	u_{i_1}\cdots u_{i_k}
	\times
	\left(
	\frac{
		q^{\frac{1}{2}} - q^{-\frac{1}{2}}
	}{
		t^{-\frac{1}{2}} - t^{\frac{1}{2}}
	}
	\right)^{|T_1|}
	\times
	\left(
	\frac{
		(q^{-1}t)^{\frac{1}{2}} - (qt^{-1})^{\frac{1}{2}}
	}{
		t^{-\frac{1}{2}} - t^{\frac{1}{2}}
	}
	\right)^{|T_2|}
	\notag 
	\\
	&\hspace{5.3cm}\times
	(
	\widetilde{\Psi}_{\lambda}^{(q,t^{-1})}
	\comp 
	\bigg|_{
		\substack{
			q_1 = q, \\
			q_2 = q^{-1}t,\\
			q_3 = t^{-1} \\
		}
	}
	)
	\bigg[
	\prod_{1 \leq a < b \leq k}
	\cals{C}^{(i_a,i_b)}\left(
	\frac{z_b}{z_a}
	; q, t
	\right)
	\bigg]
	\Biggr\}
	\notag 
\end{align}
\normalsize
Thus, we have proved \textbf{Lemma \ref{lemm42-1252-25jul}}.


\end{document}